\documentclass[11pt]{article}
\usepackage[margin=1in]{geometry}\usepackage{amsmath,amsfonts,amsthm,amssymb}
\usepackage{url}
\usepackage{color}
\usepackage{makecell}
\usepackage[usenames,dvipsnames,svgnames,table]{xcolor}
\usepackage[colorlinks=true, linkcolor=red, urlcolor=blue, citecolor=gray]{hyperref}
\usepackage{algorithm}
\usepackage{algpseudocode}
\usepackage{framed}
\usepackage{float}
\makeatletter

\DeclareMathOperator*{\E}{\mathbb{E}}
\let\Pr\relax
\DeclareMathOperator*{\Pr}{\mathbb{P}}
\DeclareMathOperator{\tr}{tr}

\DeclareMathOperator{\rank}{rank}

\DeclareMathOperator{\nnz}{nnz}

\DeclareMathOperator{\grad}{\nabla}
\DeclareMathOperator*{\argmin}{arg\,min}

\newcommand{\Sum}{\mathcal{S}}
\newcommand{\R}{\mathbb{R}}

\newcommand{\algA}{\mathcal{A}}
\newcommand{\algB}{\mathcal{B}}
\newcommand{\algW}{\mathcal{W}}

\newcommand{\poly}{\mathop\mathrm{poly}}

\newcommand{\eqdef}{\mathbin{\stackrel{\rm def}{=}}}
\newcommand{\norm}[1]{\left \|#1 \right \|}
\newcommand{\abs}[1]{\left|#1\right|}

\newcommand{\bv}[1]{\mathbf{#1}}

\newcommand{\indicVec}{\bv{1}}
\newcommand{\onesVec}{\bv{1}}
\newcommand{\zeroVec}{\bv{0}}

\newif\iftodo
\todotrue

\iftodo
\newcommand{\todo}[1]{\textcolor{blue}{TODO: #1}}
\newcommand{\sidford}[1]{{\bf \color{green} Aaron : #1}}

\else
\newcommand{\todo}[1]{}
\newcommand{\sidford}[1]{}
\fi

\makeatletter

\newtheorem*{rep@theorem}{\rep@title}
\newcommand{\newreptheorem}[2]{%
\newenvironment{rep#1}[1]{%
 \def\rep@title{#2 \ref{##1}}%
 \begin{rep@theorem}}%
 {\end{rep@theorem}}}
\makeatother
\newtheorem{theorem}{Theorem}
\newreptheorem{theorem}{Theorem}

\newtheorem{corollary}[theorem]{Corollary}
\newtheorem{lemma}[theorem]{Lemma}
\newtheorem*{lemma*}{Lemma}
\newreptheorem{lemma}{Lemma}
\newreptheorem{claim}{Claim}
\newtheorem{fact}[theorem]{Fact}
\newtheorem{claim}[theorem]{Claim}
\newtheorem{definition}[theorem]{Definition}

  \usepackage{nth}
  \usepackage{intcalc}

\title{Spectrum Approximation Beyond Fast Matrix Multiplication: \\ Algorithms and Hardness}
\author{
 Cameron Musco \\
  MIT \\
  \texttt{cnmusco@mit.edu}
\and
Praneeth Netrapalli\\
Microsoft Research\\
\texttt{praneeth@microsoft.com}
\and
Aaron Sidford\\
Stanford University\\
 \texttt{sidford@stanford.edu}
\and
Shashanka Ubaru\\
University of Minnesota\\
\texttt{ubaru001@umn.edu}
\and
David P. Woodruff\\
Carnegie Mellon University \\
\texttt{dwoodruf@cs.cmu.edu}
}
\date{}

\begin{document}

\begin{titlepage}
  \maketitle
% !TEX root = normEstimation.tex
\begin{abstract}
	Understanding the singular value spectrum of a matrix $\bv{A} \in \R^{n \times n}$ is a fundamental task in countless numerical computation and data analysis applications.
	In matrix multiplication time, it is possible to perform a full SVD of $\bv{A}$ and directly compute the singular values $\sigma_1,\ldots,\sigma_n$. However, little is known about algorithms that break this runtime barrier.
	
	Using tools from stochastic trace estimation, polynomial approximation, and fast linear system solvers, we show how to efficiently isolate different ranges of $\bv{A}$'s spectrum and approximate the number of singular values in these ranges. We thus effectively compute an \emph{approximate histogram} of the spectrum, which can stand in for the true singular values in many applications.

	We use our histogram primitive to give the first algorithms for approximating a wide class of symmetric matrix norms and spectral sums \emph{faster than the best known runtime for matrix multiplication}.
	For example, we show how to obtain a $(1 + \epsilon)$ approximation to the Schatten $1$-norm (i.e. the nuclear or trace norm) in just $\tilde O((\nnz(\bv{A})n^{1/3} + n^2)\epsilon^{-3})$ time for $\bv{A}$ with uniform row sparsity or $\tilde O(n^{2.18} \epsilon^{-3})$ time for dense matrices. The runtime scales smoothly for general Schatten $p$-norms, notably becoming $\tilde O (p \nnz(\bv{A}) \epsilon^{-3})$ for any real $p \ge 2$.
	
	At the same time, we show that the complexity of spectrum approximation is inherently tied to fast matrix multiplication in the small $\epsilon$ regime. 
	We use fine-grained complexity to 
	give conditional lower bounds for spectrum approximation, showing that achieving milder $\epsilon$ dependencies in our algorithms would imply triangle detection algorithms for general graphs running in faster than state of the art matrix multiplication time. This further implies, through a reduction of \cite{williams2010subcubic}, that highly accurate spectrum approximation algorithms running in subcubic time can be used to give subcubic time matrix multiplication. As an application of our bounds,
we show that precisely computing all effective resistances in a graph
in less than matrix multiplication time is likely difficult, barring a major algorithmic breakthrough. 
\end{abstract}

\thispagestyle{empty}
\end{titlepage}

% !TEX root = normEstimation.tex
\section{Introduction}\label{sec:intro}

Given $\bv A \in \R^{n \times d}$, a central primitive in numerical computation and data 
analysis is to compute $\bv{A}$'s spectrum: the  singular values
$\sigma_1(\bv{A}) \ge \ldots \ge \sigma_{d}(\bv{A}) \ge 0$. 
These values can reveal matrix structure and low effective dimensionality,
which can be exploited in a wide range of spectral data analysis
methods~\cite{jolliffe2002principal,ubaru2016fast}. 
The singular values are also used as 
tuning parameters in many numerical algorithms performed on $\bv{A}$~\cite{golub2012matrix}, 
and in general,
to determine some of the most well-studied matrix functions~\cite{higham2008functions}. 
For example, for any $f: \R^+ \rightarrow \R^+$, we can define the \emph{spectral sum}:
\begin{align*}
\Sum_f(\bv{A}) \eqdef \sum_{i=1}^d f(\sigma_i(\bv{A})).
\end{align*}
Spectral sums often serve as snapshots of $\bv{A}$'s spectrum and are important in many applications.
They encompass, for example, the log-determinant, the trace inverse, 
the Schatten $p$-norms, including the important nuclear norm, and general
Orlicz norms (see Section \ref{sec:prior} for details).

While the problem of computing a few of the largest or smallest
singular values of $\bv{A}$ has been exhaustively studied~\cite{parlett1998symmetric,saad2011numerical},
much less is known about algorithms that approximate the full spectrum, and in particular,
allow for the computation of summary statistics such as spectral sums. In $n^\omega$ time,
it is possible to perform a full SVD and compute the singular values
exactly.\footnote{ 
Note that an exact SVD is incomputable even with exact arithmetic \cite{trefethen1997numerical}.
Nevertheless, direct methods for the SVD obtain superlinear convergence 
rates and hence are often considered to be `exact'.} Here, and throughout, $\omega \approx 2.3729$ denotes the \emph{current}
best exponent of fast matrix multiplication \cite{williams2012multiplying}. However, even if
one simply desires, for example, a constant factor approximation to the
nuclear norm $\norm{\bv{A}}_1$, no $o(n^\omega)$ time algorithm is known. 
We study the question of spectrum approximation, asking whether obtaining
an accurate picture of $\bv{A}$'s spectrum is 
truly as hard as matrix multiplication, or if it is possible to break this barrier. 
We focus on spectral sums as a motivating application. 

\subsection{Our Contributions}

\paragraph{Upper Bounds:}
On the upper bound side, we show that significant information about $\bv{A}$'s spectrum 
can be determined in $o(n^\omega)$ time, for the current value of $\omega$. 
We show how to compute a histogram of the spectrum, which gives approximate 
counts of the number of squared singular values in the ranges
$[(1-\alpha)^{t} \sigma_1^2(\bv{A}), (1-\alpha)^{t-1} \sigma_1^2(\bv{A})]$ for
some width parameter $\alpha$ and for $t$ ranging from $0$ to some maximum $T$. 
Specifically  our algorithm satisfies the following:

\begin{theorem}[Histogram Approximation -- Informal]\label{thm:histogramIntro} Given $\bv{A} \in \R^{n \times d}$, let $b_t$ be the number of squared singular values of $\bv{A}$ on the range $[(1-\alpha)^{t} \sigma_1^2(\bv{A}), (1-\alpha)^{t-1} \sigma_1^2(\bv{A})]$. Then given error parameter $\epsilon > 0$, with probability $99/100$, Algorithm \ref{algo:histogram} outputs for all $t \in \{0,...,T \}$, $\tilde b_t$ satisfying:
\begin{align*}
(1-\epsilon) b_t \le \tilde b_t \le (1+\epsilon) b_t + \epsilon (b_{t-1} + b_{t+1}).
\end{align*}
%where $\Delta_t \le \lceil \log_{(1-\alpha)}  \lambda \rceil \cdot \epsilon_2 (b_{t-1} + b_{t+1})$ for all $t$ and $\E \Delta_t = \epsilon_2 (b_{t-1} + b_{t+1})$.
For input parameter $k \in \{1,...,d\}$, let $\bar \kappa \eqdef \frac{k\sigma_k^2(\bv{A}) + \sum_{i=k+1}^d \sigma_i^2(\bv{A})}{d \cdot (1-\alpha)^T}$ and $\hat \kappa \eqdef \frac{\sigma_{k+1}^2(\bv{A})}{(1-\alpha)^T}$. Let  $d_s(\bv{A})$  be the maximum number of nonzeros  in a row of $\bv{A}$.
The algorithm's runtime can be bounded by:
\small{
\begin{align*}
\tilde O \left (\frac{\nnz(\bv{A})k + dk^{\omega-1} + \sqrt{\nnz(\bv{A}) [d\cdot d_s(\bv{A}) + dk ] \bar \kappa}}{\poly(\epsilon,\alpha)} \right ) \text{\hspace{.5em} or \hspace{.5em}} \tilde O \left (\frac{\nnz(\bv{A})k + dk^{\omega-1} + (\nnz(\bv{A})+dk) \lceil \sqrt{\hat \kappa}\rceil}{\poly(\epsilon,\alpha)} \right )
\end{align*}}
\normalsize
 for sparse $\bv{A}$ or $ \tilde O \left (\frac{nd^{\gamma-1} + n^{1/2}d^{3/2} \sqrt{\bar \kappa})}{\poly(\epsilon,\alpha)} \right )$ for dense $\bv{A}$, where $d^\gamma$ is the time it takes to multiply a $d \times d$ matrix by a $d\times k$ matrix using fast matrix multiplication.
\end{theorem}

This primitive is useful on its own, giving an accurate summary of $\bv{A}$'s spectrum 
which can be used in many downstream applications. Setting the parameter $k$ appropriately to balance costs (see overview in Section \ref{sec:approach}), we use it to
give the first $o(n^\omega)$ algorithms for computing $(1\pm \epsilon)$ relative
error approximations to a broad class of spectral sums for functions $f$, which are a) smooth 
and b) quickly growing, so that very small singular values cannot make a significant contribution
to $\Sum_f(\bv{A})$.
This class includes for example the Schatten $p$-norms for all $p > 0$, the SVD entropy,
the Ky Fan norms, and many general Orlicz norms.

For a summary of our $p$-norm results see Table \ref{resultsTable}.
Focusing for simplicity on square matrices, with uniformly sparse rows, and assuming $\epsilon$, $p$ are constants,
our algorithms approximate $\norm{\bv{A}}_p^p$ in $\tilde O(\nnz (\bv{A}) )$ time
for any real $p \ge 2$.\footnote{For any $\bv{A} \in \R^{n \times d}$, $\nnz(\bv{A})$ denotes the number of nonzero entries in $\bv{A}$.} For $p \le 2$, we achieve
\small $\tilde O \left (\nnz(\bv{A}) n^{\frac{1/p-1/2}{1/p+1/2}} + n^{\frac{4/p-1}{2/p+1}}\sqrt{\nnz(\bv{A})}\right )$ \normalsize runtime.
In the important case of  $p=1$, this becomes
$\tilde O \left (\nnz(\bv{A}) n^{1/3} + n\sqrt{\nnz(\bv{A})}\right)$. 
Note that $n\sqrt{\nnz(\bv{A})} \le n^2$, and for sparse enough $\bv{A}$, this bound is subquadratic. For dense $\bv{A}$, we use 
fast matrix multiplication, achieving time 
$\tilde O \left (n^{\frac{2.3729 - .0994p}{1 + .0435p}} \right )$ 
for all $p < 2$. For $p=1$, this gives $\tilde O(n^{2.18})$. 
Even without fast matrix multiplication, the runtime is $\tilde O(n^{2.33})$, and so $o(n^\omega)$ for $\omega \approx 2.3729$.

\paragraph{Lower Bounds:}
On the lower bound side, we show that obtaining $o(n^\omega)$ time
spectrum approximation algorithms with very high accuracy may be difficult.
Our runtimes all depend polynomially on the error $\epsilon$, and we show that improving 
this dependence, e.g., to $\log(1/\epsilon)$, or even to a better polynomial, would give faster
algorithms for the well studied Triangle Detection problem.

Specifically, for a broad class of spectral sums, including all Schatten $p$-norms with $p \neq 2$, 
SVD entropy, $\log \det(\bv{A})$, $\tr(\bv{A}^{-1})$, and $\tr(\exp(\bv{A}))$, we show that
any $(1 \pm \epsilon)$ approximation algorithm running in $O (n^\gamma \epsilon^{-c})$ time 
yields an algorithm for triangle detection running in $O (n^{\gamma+O(c)})$ time. 
For $\gamma < \omega$ and sufficiently small $c$, such an algorithm would improve the state
of the art in triangle detection, which currently requires $\Theta(n^\omega)$ time on dense graphs.
Furthermore, 
through a reduction of \cite{williams2010subcubic}, any subcubic time triangle detection 
algorithm yields a subcubic time algorithm for Boolean Matrix Multiplication (BMM). 
Thus, any spectral sum algorithm achieving subcubic runtime and $\frac{1}{\epsilon^c}$ 
accuracy for small enough constant $c$, must (implicitly) implement fast matrix multiplication.
This is in stark contrast to the fact that, for $c =3$, for many spectral sums, including all Schatten $p$-norms with $p \ge 1/2$, we obtain subcubic, and in fact $o(n^\omega)$ for $\omega = 2.3729$, runtimes without using fast matrix multiplication (see Table \ref{resultsTable} for precise $\epsilon$ dependencies).

\begin{table}
\small{
\begin{center}
  \begin{tabular}{| c | c | c | c | c |}
    \hline
    \textbf{$p$ range} & \textbf{Sparsity} & \textbf{Approx.} & \textbf{Runtime} & \textbf{Theorem} \\ \hline
    %$\norm{\bv{A}}_p^p$, for $p > 2$, dense & $(1\pm \epsilon)$ & $\tilde O \left (n^2 p \epsilon^{-3} \right )$ & \\\hline
    $p > 2$ & uniform & $(1+ \epsilon)$ & $\tilde O \left (\nnz (\bv{A}) \cdot p/\epsilon^{3} \right )$ & Thm \ref{thm:largep} \\\hline
    $p \le 2$ & uniform & $(1+ \epsilon)$ & \footnotesize $\tilde O \left (\frac{1}{f(p,\epsilon)} \left [\nnz(\bv{A}) n^{\frac{1/p-1/2}{1/p+1/2}} + n^{\frac{4/p-1}{2/p+1}}\sqrt{\nnz(\bv{A})}\right ] \right )$ \small & Thm \ref{thm:smallp} \\\hline
    $p \le 2$ & dense & $(1+ \epsilon)$ & \footnotesize $\tilde O \left (\frac{1}{f(p,\epsilon)} n^{\frac{2.3729-.0994p}{1+.0435p}} \right )$ or $\tilde O \left (\frac{1}{f(p,\epsilon)} \cdot  n^{\frac{3 + p/2}{1 + p/2}} \right )$ w/o FMM \small & Thm \ref{finalthm:dense}\\\hline
    $p > 0$ & general & $(1+\epsilon)$ & $\tilde O \left (\frac{1}{f(p,\epsilon)} \left [\nnz(\bv{A}) n^{\frac{1}{1+p}} + n^{1+\frac{2}{1+p}} \right ] \right )$ & Thms \ref{thm:largep}, \ref{thm:smallp} \\ \hline
    $p > 2$ & general & $1/\gamma$ & $\tilde O(p \nnz(\bv{A}) \cdot  n^\gamma )$ & Thm \ref{uniformSparsity} \\ \hline%\Xhline{2\arrayrulewidth}
    $p < 2$ & general & $1/\gamma$ & $\tilde O \left(\frac{1}{p^3} \left [\nnz(\bv{A})n^{\frac{1/p-1/2}{1/p+1/2}+\gamma/2}+\sqrt{\nnz(\bv{A})} \cdot n^{\frac{4/p-1}{2/p-1}}\right ]\right)$ & Thm \ref{uniformSparsity} \\\hline
  \end{tabular}
\end{center}

\caption{\small{Summary of our results for approximating the Schatten $p$-norms. We define $f(p,\epsilon) = \min \{1,p^3\} \cdot \epsilon^{\max\{3,1+1/p\}}$, which appears as a factor in many of the bounds.
The uniform sparsity assumption is that the maximum row sparsity $d_s(\bv{A}) \le \frac{\xi}{n} \nnz(\bv{A})$ for some constant $\xi$. In our theorems, we give general runtimes, parameterized by $\xi$. When we do not have the uniform sparsity assumption, we are still able to give a $(1 + \epsilon)$ approximation in $\tilde O(\epsilon^{-3}\nnz(\bv{A})\sqrt{n} + n^2)$ time for example for $\norm{\bv{A}}_1$. We can also give $1/\gamma$ approximation for any constant $\gamma < 1$ by paying an $n^{\gamma/2}$ factor in our runtime.
Note that for dense matrices, for all $p$ we obtain $o(n^\omega)$ runtime, or $o(n^3)$ runtime if we do not use fast matrix multiplication. 
}}

\label{resultsTable}}
\end{table}
\normalsize

Our lower bounds hold even for well-conditioned matrices and structured matrices 
like symmetric diagonally dominant (SDD) systems, both of which admit nearly linear time algorithms for system solving \cite{spielman2004nearly}. 
This illustrates a dichotomy between linear algebraic primitives like 
applying $\bv{A}^{-1}$ to a vector and spectral summarization tasks like precisely 
computing $\tr(\bv{A}^{-1})$, which in some sense require more global 
information about the matrix. Our analysis has ramifications regarding natural open
problems in graph theory and numerical computation. For example, for graph Laplacians,
we show that accurately computing all \emph{effective resistances} yields an 
accurate algorithm for computing $\tr(\bv{A}^{-1})$ of certain matrices, 
which is enough to give triangle detection. 

\subsection{Related Work on Spectral Sums}\label{sec:prior}

The applications of approximate spectral sum computation are very broad. 
When $\bv{A}$ is positive semidefinite (PSD) and $f(x) = \log(x)$, $\Sum_f(\bv{A})$ is
the log-determinant, which is important in machine learning and inference
applications \cite{rasmussen2004gaussian,davis2007information,friedman2008sparse}.
For $f(x) = 1/x$, $\Sum_f(\bv{A})$ is the trace of the inverse, used in uncertainty
quantification \cite{bekas2009low} and quantum chromodynamics \cite{stathopoulos2013hierarchical}.

When $f(x) = x^p$, $S_f(\bv{A}) = \norm{\bv{A}}_p^p$ where $\norm{\bv{A}}_p$ is the
Schatten $p$-norm of $\bv{A}$.
%Of particular importance is the
Computation of the Schatten $1$-norm, also known as the nuclear or trace norm, is required in a wide variety of applications. It is often used in place of the matrix rank in matrix completion algorithms and other convex relaxations of rank-constrained optimization problems \cite{cr12,dtv11,jain2013low,netrapalli2014non}. It
appears as the `graph energy' in theoretical chemistry \cite{gutman1992total,gutman2001energy}, 
the `singular value bound' in differential privacy \cite{hlm10,lm12}, and in 
rank aggregation and collaborative ranking \cite{lu2014individualized}.

Similar to the nuclear norm, general Schatten $p$-norms are used in convex
relaxations for rank-constrained optimization~\cite{nie2012low,nie2012robust}. They also appear in image processing applications such as
denoising and background subtraction~\cite{xie2015weighted}, 
classification~\cite{luo2014schatten}, restoration~\cite{xiehyperspectral}, and 
feature extraction~\cite{du2015two}.

When $f(x) = -x \log x$ (after $\bv{A}$ is normalized by 
$\norm{\bv{A}}_1$), $\Sum_f(\bv{A})$ is the SVD entropy~\cite{alter2000singular},
which is used in feature selection~\cite{varshavsky2006novel,banerjee2014feature},
financial data analysis~\cite{caraiani2014predictive,gu2015does}, 
and genomic data~\cite{alter2000singular} applications.

Despite their importance, prior to our work, few algorithms for fast computation of
spectral sums existed. Only a few special cases of the Schatten $p$-norms were
known to be computable efficiently in $o(n^\omega)$ time. These include the 
Frobenius norm ($p=2$) which is trivially computed in $O(\nnz(\bv{A}))$ time,
the spectral norm ($p=\infty$) which can be estimated via the Lanczos method
in $\tilde O(\nnz(\bv{A}) \epsilon^{-\frac{1}{2}})$ time \cite{kuczynski1992estimating}, 
and the Schatten-$p$ norms for {\it even integers} 
$p > 2$, or general integers with PSD $\bv{A}$. These norms can be approximated 
in $O(\nnz(\bv{A}) \epsilon^{-2} )$ time via trace estimation \cite{w14,bcky16}, since
when $p$ is even or $\bv{A}$ is PSD, ${\bf A}^p$ is PSD and so its trace equals $\norm{\bv{A}}_p^p$.

There are a number of works which consider estimating matrix norms in 
sublinear space and with a small number of passes over $\bv{A}$ \cite{lnw14,akr15,lw16,lw16b,bcky16,lw17b}.
However, in these works, the main focus is on space complexity,
and no non-trivial runtime bounds are given. We seem to be the first to tackle 
the arguably more fundamental problem of obtaining the best time complexity for
simple norms like the Schatten-$p$ norms, irrespective of space and pass complexity.

Another interesting line of works tries to estimate the 
Schatten-$p$ norms of an underlying covariance matrix from a
small number of samples from the distribution \cite{kv16}, or 
from entrywise sampling under various incoherence assumptions \cite{ko17}.
This model is different from ours, as we do not assume an underlying distribution
or any incoherence properties. Moreover, even with such assumptions,
these algorithms also only give non-trivial sample complexity
when either $\bv{A}$ is PSD and $p$ is an integer, or $\bv{A}$ is a 
general matrix but $p$ is an even integer, which as mentioned above 
are easy to handle from the perspective of time complexity alone.

A number of works have focused on computing spectral sums when $\bv{A}$ has 
bounded condition number, and relative error results exist for example for the 
log-determinant, $\tr(\text{exp}(\bv{A}))$, $\tr(\bv{A}^{-1})$, and 
the Schatten $p$-norms~\cite{boutsidis2015randomized,han2016approximating,shash17}. 
We are the first to give relative error results in $o(\omega)$ time for 
general matrices, without the condition number dependency.
Our histogram approach resembles spectral filtering and spectral density estimation techniques 
that have been considered in the numerical computation 
literature~\cite{zhang2015distributed,di2016efficient,lin2016approximating,ubaru2016fast,
shash17,ubaru2017fast}. However, this literature typically requires assuming gaps between the singular values and existing work is not enough 
to give relative error spectral sum approximation for general matrices

\subsection{Algorithmic Approach}\label{sec:approach}

\paragraph{Spectral Sums via Trace Estimation:}

A common approach to spectral sum approximation is to reduce to a trace estimation problem 
involving the PSD matrix $\bv{A}^T\bv{A}$, using the fact that the trace of this matrix equals the 
sum of its singular values. In fact, this has largely been the only known technique, other than 
the full SVD, for obtaining aggregate information about $\bv{A}$'s singular 
values \cite{hutchinson1990stochastic,stathopoulos2013hierarchical,wu2015estimating,DBLP:journals/corr/WimmerWZ14,DBLP:journals/corr/Roosta-KhorasaniA13,DBLP:journals/corr/FitzsimonsORF16,boutsidis2015randomized,han2016approximating}.
The idea is, letting $g(x) = f(x^{1/2})$, we have $\Sum_f(\bv{A}) = \Sum_g(\bv{A}^T\bv{A}).$
Writing the SVD $\bv{A}^T\bv{A} = \bv{U}\bv{\Lambda}\bv{U}^T$,  and
defining the matrix function $g(\bv{A}^T\bv{A}) \eqdef \bv{U} g(\bv{\Lambda})\bv{U}^T$ 
where $\left[g(\bv{\Lambda})\right]_{i,i} = g([\bv{\Lambda}]_{i,i})$, we 
have $\Sum_g(\bv{A}^T\bv{A}) = \tr(g(\bv{A}^T\bv{A}))$ since, if $g(\cdot)$ is 
nonnegative, $g(\bv{A}^T\bv{A})$ is PSD and its trace equals the sum of its singular values. 

It is well known that this trace can be approximated up to $(1 \pm \epsilon)$ 
accuracy by averaging $\tilde O(\epsilon^{-2})$ samples of the form $\bv{x}^T g(\bv{A}^T\bv{A}) \bv{x}$
where $\bv{x}$ is a random Gaussian or sign vector \cite{hutchinson1990stochastic,avron2011randomized}.
 While $g(\bv{A}^T\bv{A})$ cannot be explicitly computed without a full SVD, 
 a common approach is to approximate $g$ with a low-degree polynomial 
 $\phi$ \cite{boutsidis2015randomized,han2016approximating}. If $\phi$ has degree $q$,
 one can apply $\phi(\bv{A}^T\bv{A})$ to any vector $\bv{x}$ in $O(\nnz(\bv{A})\cdot q)$ time, 
 and so estimate its trace in just  $O(\nnz(\bv{A})\cdot \frac{q}{\epsilon^2})$ time. 
Unfortunately, for many of the functions most important in applications, 
e.g., $f(x) = x^p$ for odd $p$, $f(x) = x \log x$, $f(x) = x^{-1}$, $g(x)$ has a discontinuity at $x = 0$ and \emph{cannot} be approximated well by a low-degree polynomial near zero. While the approximation only needs to be good on the range $[\sigma_n^2(\bv{A}), \sigma_1^2(\bv{A})]$, the required degree $q$ will still typically depend on $\sqrt{\kappa}$ where $\kappa \eqdef \frac{\sigma_1^2(\bv{A})}{\sigma_n^2(\bv{A})}$ is the condition number, which can be unbounded for general matrices. 

\paragraph{Singular Value Deflation for Improved Conditioning:}
Our first observation is that, for many functions, it is not necessary to approximate $g(x)$ 
on the full spectral range. For example, for $g(x) = x^{p/2}$ (i.e., when $\Sum_g(\bv{A}^T\bv{A}) =\norm{\bv{A}}_p^p$), setting $\lambda = (\frac{\epsilon}{n} \norm{\bv{A}}_p^p)^{1/p}$, we have:
$$\sum_{\{i | \sigma_i(\bv{A}) \le \lambda\}} \sigma_i(\bv{A})^p \le n \cdot \frac{\epsilon}{n} \norm{\bv{A}}_p^p \le \epsilon \norm{\bv{A}}_p^p.$$
Hence we can safely `ignore' any $\sigma_i(\bv{A}) \le \lambda$ and still obtain a 
relative error approximation to $\Sum_g(\bv{A}^T\bv{A}) = \norm{\bv{A}}_p^p$. The larger $p$ is, the larger we can set $\lambda$ (corresponding to $(1-\alpha)^T$ in Theorem \ref{thm:histogramIntro})  to be, since, after powering, the singular values below this threshold do not contribute significantly to $\norm{\bv{A}}_p^p$.
For $\norm{\bv{A}}_p^p$, our `effective condition number' for 
approximating $g(x)$ becomes
$\hat \kappa = \frac{\sigma_1^2(\bv{A})}{\lambda^2} = (\frac{n}{\epsilon})^{2/p} \cdot \frac{\sigma_1^2(\bv{A})}{\norm{\bv{A}}_p^2}$. Unfortunately, in the worst case, we may have $\sigma_1(\bv{A}) \approx \norm{\bv{A}}_p$ and hence $\sqrt{\hat \kappa} = (\frac{n}{\epsilon})^{1/p}$. Hiding $\epsilon$ dependences, this gives runtime $\tilde O( \nnz(\bv{A}) \cdot n)$ when $p=1$.

To improve the effective condition number, we can apply \emph{singular vector deflation}.
Our above bound on $\hat \kappa$ is only tight when the first singular value is very large and so dominates $\norm{\bv{A}}_p$. We can remedy this by flattening $\bv{A}$'s spectrum by deflating off the top $k$ singular vectors (corresponding to $k$  in Theorem \ref{thm:histogramIntro}), and including their singular values in the spectral sum directly.

Specifically, letting $\bv{P}_k$ be the projection onto the top $k$ singular vectors of $\bv{A}$, we consider the deflated matrix $ \bv{\bar A} \eqdef \bv{A}(\bv{I} - \bv{P}_k)$, which has $\sigma_1(\bv{\bar A}) = \sigma_{k+1}(\bv{A})$. Importantly, $\sigma_{k+1}^p(\bv{A}) \le \frac{1}{k} \norm{\bv{A}}_p^p$, and thus this singular value cannot dominate the $p$-norm.
As an example, considering $p=1$ and ignoring $\epsilon$ dependencies, our effective condition number after deflation is 
\begin{align}\label{eq:initial_condition}
\hat \kappa = \frac{n^2 \cdot \sigma_{k+1}^2(\bv{A})}{\norm{\bv{A}}_1^2} \le \frac{n^2}{k^2}
\end{align}
The runtime required to approximate $\bv{P}_k$ via
an iterative method (ignoring possible gains from fast matrix multiplication) is roughly $O(\nnz(\bv{A})k + nk^{2})$. We then require $\tilde O(\nnz(\bv{A}) \sqrt{\hat \kappa} + nk\sqrt{\hat \kappa})$ time to approximate
the polynomial trace of $\bv{\bar A}^T\bv{\bar A}$. The $nk\sqrt{\hat \kappa}$ term 
comes from projecting off the top singular directions with each application of 
$\bv{\bar A}^T\bv{\bar A}$. Setting $k = \sqrt{n}$ to balance the costs, we 
obtain runtime $\tilde O(\nnz(\bv{A})\sqrt{n} + n^2)$. 

 For $p \neq 1$ a similar argument gives runtime $\tilde O(\nnz(\bv{A})n^{\frac{1}{p+1}} + n^{2 + \frac{1}{p+1}})$. This is already a significant improvement over a full SVD. As $p$ grows larger, the runtime approaches $\tilde O(\nnz(\bv{A}))$ reflecting the fact that for larger $p$ we can ignore a larger and larger portion of the small singular values in $\bv{A}$ and correspondingly deflate off fewer and fewer top values. 
 
 Unfortunately, we get stuck here. Considering the important Schatten-$1$ norm, 
 for a matrix with $\sqrt{n}$ singular values each equal to $\sqrt{n}$ and $\Theta(n)$ 
 singular values each equal to $1$, the tail of small singular values contributes a 
 constant fraction of $\norm{\bv{A}}_1 = \Theta(n)$. However, there is no good polynomial
 approximation to $g(x) = x^{1/2}$ on the range $[1,n]$ with degree $o(\sqrt{n})$ 
 (recall that we pick this function since $\Sum_g(\bv{A}^T\bv{A}) = \norm{\bv{A}}_1$). 
 So to accurately approximate $g(\bv{A}^T\bv{A})$, we either must deflate off all $\sqrt{n}$ 
 top singular values, requiring $\Theta(\nnz(\bv{A})\sqrt{n})$ time, or apply a $\Theta(\sqrt{n})$
 degree polynomial approximation, requiring the same amount of time.
 
\paragraph{Further Improvements with Stochastic Gradient Descent:}
To push beyond this barrier, we look to \emph{stochastic gradient} methods for linear systems. When using polynomial approximation, our bounds depend on the condition number of the interval over which we must approximate $g(\bv{A}^T\bv{A})$, after ignoring the smallest singular values and deflating off the largest. This is analogous to the condition number dependence of iterative linear system solvers like conjugate gradient or accelerated gradient descent, which approximate $f(\bv{A}^T\bv{A})$ for $f = 1/x$ using a polynomial of $\bv{A}^T\bv{A}$.

However, recent advances in convex optimization offer an alternative. Stochastic gradient methods \cite{johnson2013accelerating,shalev2014accelerated} sample one row, $\bv{a}_i$, of $\bv{A}$ at a time, updating the current iterate by adding a multiple of $\bv{a}_i$. They trade a larger number of iterations for updates that take $O(\nnz(\bv{a}_i))$ time, rather than $O(\nnz(\bv{A}))$ time to multiply $\bv{A}$ by a vector. At the same time, these methods give much finer  dependencies on the singular value spectrum. Specifically, it is possible to approximately apply $(\bv{A}^T\bv{A})^{-1}$ to a vector with the number of iterations dependent on the \emph{average condition number:} $$\bar \kappa = \frac{\frac{1}{n}\sum_{i=1}^n \sigma_i^2(\bv{A})}{\sigma^2_{n}(\bv{A})}.$$
$\bar \kappa$ is always at  most the standard condition number, $\kappa = \frac{\sigma_1^2(\bv{A})}{\sigma_n^2(\bv{A})}$. It  can be significantly smaller when $\bv{A}$ has a quickly decaying spectrum, and hence $\frac{1}{n} \sum_{i=1}^n \sigma_i^2(\bv{A}) \ll \sigma_1^2(\bv{A})$. Further, the case of a quickly decaying spectrum with a few large and many small singular values is \emph{exactly the hard case for our earlier approach}. If we can understand how to translate improvements on linear system solvers to spectral sum approximation, we can handle this hard case.

\paragraph{From Linear System Solvers to Histogram Approximation:}
The key idea to translating the improved average condition number bounds for linear systems to our problem of approximating $\Sum_f(\bv{A})$ is to note that linear system solvers can be used to apply threshold functions to $\bv{A}^T\bv{A}$. 

Specifically, given any vector $\bv{y}$, we can first compute $\bv{A}^T\bv{A} \bv{y}$. We can then apply a fast system solver to approximate $(\bv{A}^T\bv{A} + \lambda \bv{I})^{-1}\bv{A}^T\bv{A} \bv{y}$. The matrix function $r_\lambda(\bv{A}^T\bv{A}) \eqdef (\bv{A}^T\bv{A} + \lambda \bv{I})^{-1}\bv{A}^T\bv{A}$ has a number of important properties. All its singular values are between $0$ and $1$. Further, any singular value in $\bv{A}^T\bv{A}$ with value $\ge \lambda$ is mapped to a singular value in $r_\lambda(\bv{A}^T\bv{A})$ which is $\ge 1/2$. Correspondingly, any singular value $< \lambda$ is mapped to $<1/2$.

 Thus, we can apply a low degree polynomial approximation to a step function at $1/2$ to $r_\lambda(\bv{A}^T\bv{A})$ to obtain $s_\lambda(\bv{A}^T\bv{A})$, which approximates a step function at $\lambda$ \cite{frostig2016principal}. For some steepness parameter $\gamma$ which affects the degree of the polynomial approximation, for $x \ge (1+\gamma) \lambda$ we have $s_\lambda(x) \approx \ 1$ and for $x < (1-\gamma)\lambda$,  $s_\lambda(x) \approx 0$. On the intermediate range $x \in [(1-\gamma)\lambda,(1+\gamma)\lambda]$, $s_\lambda(x)$ falls somewhere between $0$ and $1$.

By composing these approximate threshold functions at different values of $\lambda$, it is possible to `split' our spectrum into a number of small spectral windows. For example, $s_a(\bv{A}^T\bv{A}) \cdot (\bv{I}- s_b(\bv{A}^T\bv{A}))$ is $\approx 1$ on the range $[a,b]$ and $\approx0$ outside this range, with some ambiguity near $a$ and $b$.

Splitting our spectrum into windows of the form $[(1-\alpha)^{t-1}, (1-\alpha)^t]$ for a width parameter $\alpha$, and applying trace estimation on each window lets us produce an approximate spectral histogram. Of course, this histogram is not exact and in particular, the `blurring' of our windows at their boundaries can introduce significant error. However, by applying a random shifting technique and setting the steepness parameter $\gamma$ small enough (i.e., $1/\poly(\alpha,\epsilon)$), we can ensure that most of the spectral weight falls outside these boundary regions with good probability, giving Theorem \ref{thm:histogramIntro}. 

\paragraph{From Histogram Approximation to Spectral Sums:}

If $\alpha$ is small enough, and $f(\cdot)$ (the function in the spectral sum)
and correspondingly $g(\cdot)$ (where $g(x)=f(x^{1/2})$) are smooth enough, we can 
approximate $\Sum_f(\bv{A})=\Sum_g(\bv{A}^T\bv{A})$
by simply summing over each bucket in the histogram, approximating $g(x)$ by its value at
one end of the bucket. 

This technique can be applied for any spectral sum.
The number of windows required (controlled by $\alpha$) and the histogram accuracy $\epsilon$ scale with the smoothness of $f(\cdot)$ and the desired accuracy in computing the sum, 
introducing polynomial runtime dependencies on these parameters. 

However, the most important factor determining the final runtime is
the smallest value $\lambda$ (corresponding to $(1-\alpha)^T$ in Theorem \ref{thm:histogramIntro}) which we must include in our histogram in order to 
approximate $S_f(\bv{A})$. The cost of computing the last bucket of the histogram is proportional to the cost of applying $s_\lambda(\bv{A}^T\bv{A})$, and hence of approximately 
computing $(\bv{A}^T\bv{A} + \lambda \bv{I})^{-1} \bv{A}^T\bv{A} \bv{y}$. Using stochastic gradient descent 
this depends on the average condition number of $(\bv{A}^T\bv{A} + \lambda \bv{I})$.

Again considering the Schatten $1$-norm for illustration, we can ignore any singular values with $\sigma_i(\bv{A}) \le \frac{\epsilon}{n} \norm{\bv{A}}_1$. Hiding $\epsilon$ dependence, this means that in our histogram, we must include any singular values of $\bv{A}^T\bv{A}$ with value $\sigma_i(\bv{A}^T\bv{A}) = \sigma_i^2(\bv{A}) \ge \frac{1}{n^2} \norm{\bv{A}}_1^2$. This gives us effective average condition number after deflating off the top $k$ singular values:
\begin{align}\label{eq:second_condition}
\bar \kappa = \frac{n^2 \sum_{i=k+1}^n \sigma_i^2(\bv{A})}{n\norm{\bv{A}}_1^2} \le \frac{n^2 \sigma_{k+1}(\bv{A}) \cdot \sum_{i=k+1}^n \sigma_i(\bv{A})}{n\norm{\bv{A}}_1^2} \le \frac{n}{k}
\end{align}
where the last inequality follows from the observation that $\sigma_{k+1}(\bv{A}) \le \frac{1}{k} \norm{\bv{A}}_1$ and $\sum_{i=k+1}^n \sigma_i(\bv{A}) \le \norm{\bv{A}}_1$. Comparing to \eqref{eq:initial_condition}, this bound is better by an $n/k$ factor.

Ignoring details and using a simplification of the runtimes in Theorem \ref{thm:histogramIntro},
we obtain an
algorithm running in $\tilde O(\nnz(\bv{A})k + nk^{2})$ time 
to deflate $k$ singular vectors, plus 
$\tilde O \left (\nnz(\bv{A}) \sqrt{\bar \kappa} + \sqrt{\nnz(\bv{A}) n k \bar \kappa} \right )$
time to approximate the spectral sum over the deflated matrix. Choosing $k$ to balance
these costs, gives our final runtimes. For the nuclear norm, using the bound on
$\bar \kappa$ from \eqref{eq:second_condition}, we set $k = n^{1/3}$ which gives 
$\bar \kappa = n^{2/3}$ and runtime $\tilde O(\nnz(\bv{A}) n^{1/3} + n^{3/2}\sqrt{d_s})$
where $d_s \le n$ is the maximum row sparsity. For dense $\bv{A}$ this is $\tilde O(n^{2.33})$, which 
is faster than state of the art matrix multiplication time. It can be further accelerated using 
fast matrix multiplication methods. See details in Section~\ref{sec:upper}.

Returning to our original hard example for intuition, we have $\bv{A}$ with $\sqrt{n}$ singular values at $\sqrt{n}$  and $\Theta(n)$ singular values at $1$. Even without deflation, we have (again ignoring $\epsilon$ dependencies) $\bar \kappa = \frac{\sum_{i=1}^n \sigma_i^2(\bv{A})}{n \lambda} = \frac{n \norm{\bv{A}}_F^2}{\norm{\bv{A}}_1^2}.$
Since $\norm{\bv{A}}_F^2 = \Theta(n^{3/2})$ and 
$\norm{\bv{A}}_1^2 = \Theta(n^2)$, this gives $\bar \kappa = \Theta( \sqrt{n}).$ Thus,
 we can actually approximate 
$\norm{\bv{A}}_1$ in just $\tilde O(\nnz(\bv{A}) n^{1/4})$ time for this matrix.

With average condition number dependence, our performance is limited by a new hard case. Consider $\bv{A}$ with $n^{1/3}$ singular values at $n^{2/3}$ and $\Theta(n)$ at $1$. The average condition number without deflation is $\frac{n \norm{\bv{A}}_F^2}{\norm{\bv{A}}_1^2} = \Theta \left ( \frac{n^{5/3}}{n}\right) = \Theta(n^{2/3})$ giving $\sqrt{\bar \kappa} =  \Theta(n^{1/3})$. Further, we can see that unless we deflate off nearly all $n^{1/3}$ top singular vectors, we do not improve this bound significantly.

% !TEX root = normEstimation.tex
\subsection{Lower Bound Approach}
\label{sec:lower_bound_approach}
We now shift focus to our lower bounds, which explore the fundamental limits of spectrum approximation using fine-grained complexity approaches.
Fine-grained complexity has had much success for graph problems,
string problems, and problems in other areas (see, e.g., \cite{w15} for a survey),
and is closely tied to understanding the complexity of matrix multiplication. However,
to the best of our knowledge it has not been applied
broadly to problems in linear algebra. 

Existing hardness results for linear
algebraic problems tend to apply to restricted computational models such as
arithmetic circuits \cite{baur1983complexity},
bilinear circuits
or circuits with bounded coefficients and number of divisions 
\cite{morgenstern1973note,r03,s03,rs03}, algorithms for
dense linear systems that can only add multiples of rows to each other \cite{kk65,k70}, and 
algorithms with restrictions on the dimension of certain manifolds defined in terms of the input
\cite{w70,w87,d13}.
In contrast, we obtain conditional lower bounds for arbitrary polynomial time algorithms
by showing that faster algorithms for them imply faster algorithms for canonical hard problems in
fine-grained complexity.

\paragraph{From Schatten $3$-norm to Triangle Detection:}
We
start with the fact that the number of triangles in any unweighted graph $G$ is equal to $\tr(\bv{A}^3)/6$, where $\bv{A}$ is the adjacency matrix. Any algorithm for approximating $\tr(\bv{A}^3)$ to high enough accuracy therefore gives an algorithm for detecting if a graph has at least one triangle. 

$\bv{A}$ is not PSD, so $\tr(\bv{A}^3)$ is actually not a function of $\bv{A}$'s singular values -- it depends on the signs of $\bv{A}$'s eigenvalues. However, the graph Laplacian given by $\bv{L}  = \bv{D}-\bv{A}$ where $\bv{D}$ is the diagonal degree matrix, is PSD and we have:
$$\norm{\bv{L}}_3^3 = \tr(\bv{L}^3) = \tr(\bv{D}^3) - 3\tr(\bv{D}^2\bv{A}) + 3\tr(\bv{D}\bv{A}^2) - \tr(\bv{A}^3).$$
$\tr(\bv{D}^2 \bv{A}) = 0$ since $\bv A$ has an all $0$ diagonal. Further, it is not hard to see that $\tr(\bv{D}\bv{A}^2) = \tr(\bv{D}^2)$. So this term and $\tr(\bv{D}^3)$ are easy to compute exactly. Thus, if we approximate $\norm{\bv L}_3^3$ up to additive error $6$, we can determine if $\tr(\bv{A}^3) = 0$ or $\tr(\bv{A}^3) \ge 6$ and so detect if $G$ contains a triangle. $\norm{\bv{L}}_3^3 \le 8n^4$ for any unweighted graph on $n$ nodes, and hence computing this norm  up to $(1 \pm \epsilon)$ relative error for $\epsilon = 3/(6n^4)$ suffices to detect a triangle. If we have an $O(n^\gamma \epsilon^{-c})$ time $(1\pm \epsilon)$ approximation algorithm for the Schatten $3$-norm, we can thus perform triangle detection in $O(n^{\gamma+4c})$ time.

Our strongest algorithmic result for the Schatten $3$-norm requires just $\tilde O(n^2/\epsilon^3)$ time for dense matrices. Improving the $\epsilon$ dependence to $o(1/\epsilon^{(\omega-2)/4} ) = O(1/\epsilon^{.09})$ for the current value of $\omega$, would yield an algorithm for triangle detection running in $o(n^\omega)$ time for general graphs, breaking a longstanding runtime barrier for this problem. Even a $1/\epsilon^{1/3}$ dependence would give a sub-cubic time triangle detection algorithm, and hence could be used to give a subcubic time matrix multiplication algorithm via the reduction of \cite{williams2010subcubic}.

\paragraph{Generalizing to Other Spectral Sums}

We can generalize the above approach to the Schatten $4$-norm by adding $\lambda$ self-loops to each node of $G$, which corresponds to replacing $\bv{A}$ with $\lambda \bv{I} + \bv{A}$. We then consider $\tr((\lambda \bv{I} + \bv{A})^4) = \norm{\lambda \bv{I} + \bv{A}}_4^4$. This is the sum over all vertices of the number of paths that start at $v_i$ and return to $v_i$ in four steps. All of these paths are either (1) legitimate  four cycles, (2) triangles combined with self loops, or (3) combinations of self-loops and two-step paths from a vertex $v_i$ to one of its neighbors and back. The number of type (3) paths is exactly computable using the node degrees and number of self loops. Additionally, if the number of self loops $\lambda$ is large enough, the number of type $(2)$ paths will dominate the number of type $(1)$ paths, even if there is just a single triangle in the graph. Hence, an accurate approximation to $\norm{\lambda \bv{I} + \bv{A}}_4^4$ will give us the number of type $(2)$ paths, from which we can 
easily compute the number of triangles.

This argument extends to a very broad class of spectral sums by considering a power series expansion of $f(x)$ and showing that for large enough $\lambda$, $\tr \left ( f(\lambda \bv{I }+ \bv{A}) \right )$ is dominated by $\tr(\bv{A}^3)$ along with some exactly computable terms. Thus, an accurate approximation to this spectral sum allows us to determine the number of triangles in $G$. This approach works for any $f(x)$ that can be represented as a power series, with reasonably well-behaved coefficients on some interval of $\R^+$, giving bounds for all $\norm{\bv A}_p$ with $p \neq 2$, the SVD entropy, $\log \det (\bv{A})$, $\tr(\bv{A}^{-1})$, and $\tr(\exp(\bv{A}))$. 

We further show that approximating $\tr(\bv{A}^{-1})$ for the $\bv{A}$ used in our lower bound can be reduced to computing all effective resistances of a certain graph Laplacian up to $(1 \pm \epsilon)$ error. Thus, we rule out highly accurate (with $1/\epsilon^c$ dependence for small $c$) approximation algorithms for all effective resistances, despite the existence of linear time system solvers (with $\log(1/\epsilon)$ error dependence) for Laplacians \cite{spielman2004nearly}.
Effective resistances and leverage scores are quantities that have recently been crucial to achieving algorithmic improvements to fundamental problems like graph sparsification \cite{spielmanS08} and regression \cite{LiMP13,CohenLMMPS15}. While crude multiplicative approximations to the quantities suffice for these problems, more recently computing these quantities has been used to achieve breakthroughs in solving maximum flow and linear programming \cite{LeeS14}, cutting plane methods \cite{LeeSW15}, and sampling random spanning trees \cite{KelnerM09, MadryST15}. In each of these settings having more accurate estimates would be a natural route to either simplify or possibly improve existing results; we show that this is unlikely to be successful if the precision requirements are two high.

\subsection{Paper Outline}

\medskip
\noindent \textbf{Section \ref{sec:prelim}: Preliminaries}.
We review notations that will be used throughout.

\medskip
\noindent \textbf{Section \ref{sec:stepApprox}: Spectral Window Approximation}.
We show how to approximately restrict the spectrum of a matrix to a small window, 
which will be our main primitive for accessing the spectrum.

\medskip
\noindent \textbf{Section \ref{sec:stepDis}: Spectral Histogram Approximation}.
We show how our spectral window algorithms can be used to compute an approximate spectral histogram. We give applications to approximating general spectral sums, including the Schatten $p$-norms, Orlicz norms, and Ky Fan norms.

\medskip
\noindent \textbf{Section \ref{sec:lower}: Lower Bounds}.
We prove lower bounds showing that highly accurate spectral sum algorithms can be used to give  algorithms for triangle detection and matrix multiplication.

\medskip
\noindent \textbf{Section \ref{sec:poly}: Improved Algorithms via Polynomial Approximation}.
We demonstrate how to tighten $\epsilon$ dependencies in our runtimes using a more general polynomial approximation approach.

\medskip
\noindent \textbf{Section \ref{sec:upper}: Optimized Runtime Bounds}.
We instantiate the techniques of Section \ref{sec:poly} give our best runtimes for the Schatten $p$-norms and SVD entropy.

% !TEX root = normEstimation.tex
\section{Preliminaries}\label{sec:prelim}
Here we outline notation and conventions used throughout the paper.

\medskip
\noindent 
\textbf{Matrix Properties:} For $\bv{A} \in \R^{n \times d}$ we assume without loss of generality that $d \le n$. We let $\sigma_1(\bv{A}) \geq \ldots \geq \sigma_d(\bv{A}) \geq 0$ denote the matrix's singular values, $\nnz(\bv{A})$ denote the number of non-zero entries, and $d_s(\bv{A})$ denote the maximum number of non-zero entries in any row. Note that $d_s(\bv{A}) \in [\nnz(\bv{A}) / n, d]$. 

\medskip
\noindent 
\textbf{Fast Matrix Multiplication:}
Let $\omega \approx 2.3729$ denote the current best exponent of fast matrix multiplication \cite{williams2012multiplying,gall2017improved}. Additionally,
let $\omega(\gamma)$ denote the exponent such that it takes $O\left (d^{\omega(\gamma)}\right)$ time to multiply a $d \times d$ matrix by a $d\times d^\gamma$ matrix for any $\gamma \le 1$. $\omega(\gamma) = 2$ for $\gamma < {\alpha}$ where $\alpha > 0.31389$ and $\omega(\gamma) = 2 + (\omega-2)\frac{\gamma - \alpha}{1-\alpha}$ for $\gamma \ge {\alpha}$ \cite{le2012faster,gall2017improved}. For $\gamma =1$, $\omega(\gamma) = \omega$.

\medskip
\noindent 
\textbf{Asymptotic Notation:}
We use $\tilde O(\cdot)$ notation to hide poly-logarithmic factors in the input parameters, including dimension, failure probability, and error $\epsilon$. We use `with high probability' or `w.h.p.' to refer to events happening with probability at least $1-1/d^c$ for some constant $c$, where $d$ is our smaller input dimension.

\medskip
\noindent 
\textbf{Other:}
We denote $[d] \eqdef \{0,\ldots,d\}$. For any $\bv{y} \in \mathbb{R}^d$ and PSD $\bv{N} \in \mathbb{R}^{d\times d}$, we denote $\norm{\bv{y}}_{\bv{N}} \eqdef \sqrt {\bv{y}^T \bv{N} \bv{y}}.$
% !TEX root = normEstimation.tex
\section{Approximate Spectral Windows via Ridge Regression}\label{sec:stepApprox}

In this section, we give state-of-the-art results for approximating spectral windows over $\bv{A}$. As discussed, our algorithms will split $\bv{A}$'s spectrum into small slices using these window functions, performing trace estimation to estimate the number of singular values on each window and producing an approximate spectral histogram.

In Section~\ref{sec:tools:step} we show how to efficiently apply smooth approximations to threshold functions of the spectrum provided we have access to an algorithm for solving regularized regression problems with the matrix. In Section~\ref{sec:tools:regress} we then provide the fastest known algorithms for the regression problems in the given parameter regimes using both stochastic gradient methods and traditional solvers.
Departing from our high level description in Section \ref{sec:approach}, we actually incorporate singular vector deflation directly into our system solvers to reduce condition number. This simplifies our final algorithms but has the same effect as the deflation techniques discussed in Section \ref{sec:approach}.
Finally,  in Section~\ref{sec:tools:window} we provide algorithms and runtime analysis for applying smooth approximations to window functions of the spectrum, which is the main export of this section.

\subsection{Step Function Approximation}
\label{sec:tools:step}

To compute a window over $\bv{A}$'s spectrum, we will combine two threshold functions at the boundaries of the window. We begin by discussing how to compute these threshold functions.

Let $s_\lambda: [0,1] \rightarrow [0,1]$ be the threshold function at $\lambda$. $s_\lambda(x) = 1$ for $x \in [\lambda,1]$ and $0$ for $x \in [0,\lambda)$. For some gap $\gamma$ we define a soft step function by:
\begin{definition}[Soft Step Function]\label{def:softstep}
$s_\lambda^\gamma: [0,1] \rightarrow [0,1]$ is a $\gamma$-soft step at $\lambda > 0$ if:
\begin{align}\label{polyApproxGoal}
s_\lambda^\gamma(x)= \begin{cases} 0 \text{ for }x \in [0,(1-\gamma)\lambda] \\
1 \text{ for }x \in [\lambda,1]
\end{cases}\text{ and  }
s_\lambda^\gamma(x) \in [0,1]\text{ for }x \in [(1-\gamma)\lambda, \lambda].
\end{align}
\end{definition}
We use the strategy from \cite{frostig2016principal}, which, for $\bv{A}$ with $\norm{\bv{A}}_2 \le 1$ shows how to efficiently multiply a $\gamma$-soft step $s^\gamma_\lambda(\bv{A}^T\bv{A})$ by any $\bv{y} \in \mathbb{R}^d$ using ridge regression. The trick is to first approximately compute $\bv{A}^T\bv{A}(\bv{A}^T\bv{A} + \lambda \bv{I})^{-1} \bv{y} = r_\lambda(\bv{A}^T\bv{A}) \bv{y}$ where $r_\lambda(x) \eqdef \frac{x}{x + \lambda}$.
Then, note that $s_{1/2}(r_\lambda(x)) = s_\lambda(x)$. Additionally, the symmetric step function $s_{1/2}$ can be well approximated with a low degree polynomial. Specifically, there exists a polynomial of degree $O(\gamma^{-1} \log(1/(\gamma \epsilon)))$ that is within additive $\epsilon$ of a true $\gamma$-soft step at $1/2$ and can be applied stably such that any error in computing $r_\lambda(\bv{A}^T \bv{A})$ remains bounded. The upshot, following from Theorem 7.4 of \cite{allen2016faster} is:

\begin{lemma}[Step Function via Ridge Regression]\label{ridge2step}
Let $\algA(\bv{A},\bv{y},\lambda,\epsilon)$ be an algorithm
that on input $\bv{A} \in \R^{n \times d}$, $\bv{y} \in \R^d$, $\lambda,\epsilon > 0$ returns $\bv{x} \in \R^{d}$ such that $\norm{\bv{x}-(\bv{A}^T\bv{A} + \lambda \bv{I})^{-1} \bv{y}}_2 \le \epsilon \norm{\bv{y}}_2$ with high probability. 
%Then for any $\bv{A} \in \R^{n \times d}$ with $\norm{\bv{A}}_2 \le 1$, $\gamma, \epsilon,\lambda > 0$, and $\bv{x} \in \R^d$. 
Then there is an algorithm $\algB(\bv{A},\bv{y},\lambda,\gamma,\epsilon)$ which on input $\bv{A} \in \R^{n \times d}$ with $\norm{\bv{A}}_2 \le 1$, $\bv{y} \in \R^d$, $\lambda \in (0,1)$, and $\gamma, \epsilon > 0$, returns $\bv{x} \in \R^d$ with 
$$\norm{\bv{x} - s^\gamma_\lambda(\bv{A}^T\bv{A})\bv{y}}_2 \le \epsilon \norm{\bv{y}}_2$$
where $s^\gamma_\lambda$ is a $\gamma$-soft step at $\lambda$ (i.e. satisfies Definition \ref{def:softstep}). $\algB(\bv{A},\bv{y},\lambda,\gamma,\epsilon)$ requires $O  (\gamma^{-1} \log(1/\epsilon\gamma))$ calls to $\algA(\bv{A},\bv{y},\lambda,\epsilon')$ along with $O(\nnz(\bv{A}))$ additional runtime, where $\epsilon' = \poly (1/ (\gamma\epsilon))$.
\end{lemma}

\subsection{Ridge Regression}
\label{sec:tools:regress}

Given Lemma~\ref{ridge2step}, to efficiently compute $s^\gamma_\lambda(\bv{A}^T\bv{A}) \bv{y}$ for $s^\gamma_\lambda( \cdot )$ satisfying Definition \ref{def:softstep}, it suffices to quickly approximate $(\bv{A}^T\bv{A} + \lambda \bv{I})^{-1} \bv{y}$ (i.e. to provide the algorithm $\algA(\bv{A},\bv{y},\lambda,\epsilon)$ used in the lemma). In this section we provide two theorems which give the state-of-the-art ridge regression running times achievable in our parameter regime, using sampling, acceleration, and singular value deflation.

Naively, computing $(\bv{A}^T\bv{A} + \lambda \bv{I})^{-1} \bv{y}$ using an iterative system solver involves a dependence on the condition number $\sigma_1^2(\bv{A})/\lambda$. In our theorems, this condition number is replaced by a deflated condition number depending on $\sigma^2_{k}(\bv{A})$ for some input parameter $k \in [d]$. We achieve this improved dependence following the techniques presented in~\cite{gonen2016solving}. We first approximate the top $k$ singular vectors of $\bv{A}$ and then construct a preconditioner based on this approximation, which significantly flattens the spectrum of the matrix. By using this preconditioner in conjunction with a stochastic gradient based linear system solver, we further enjoy an average condition number dependence.
The following theorem summarizes the results.

\begin{theorem}[Ridge Regression -- Accelerated Preconditioned SVRG]\label{accPreCondSvrgMainBody} For any $\bv A \in \mathbb{R}^{n \times d}$ and $\lambda > 0$, let $\bv{M}_\lambda \eqdef \bv{A}^T\bv{A} + \lambda \bv{I}$. Let $ \bar \kappa \eqdef \frac{k \sigma_k^2(\bv{A}) + \sum_{i=k+1}^d \sigma_i^2(\bv{A})}{d\lambda}$ where $k \in [d]$ is an input parameter.
There is an algorithm that builds a preconditioner for $\bv{M}_\lambda$ using precomputation time $\tilde O(\nnz(\bv{A})k+ dk^{\omega-1})$ for sparse $\bv{A}$ or $\tilde O(nd^{\omega(\log_d k)-1})$ time for dense $\bv{A}$, and for any input $\bv{y} \in \R^d$, returns
$\bv{x}$ such that with high probability
$\norm{\bv{x}-\bv{M}_\lambda^{-1} \bv{y}}_{\bv{M}_\lambda} \le \epsilon \norm{\bv{y}}_{\bv{M}_\lambda^{-1}}$
in
$$\tilde O \left (\nnz(\bv{A}) + \sqrt{\nnz(\bv{A}) [d\cdot d_s(\bv{A}) + dk ] \bar \kappa} \right )$$
time for sparse $\bv{A}$ or 
$\tilde O \left (nd + n^{1/2}d^{3/2} \sqrt{\bar \kappa}) \right)$
time for dense $\bv{A}$.
\end{theorem}

\begin{proof} We give a proof in Appendix \ref{sec:solverAppendix}.
Note that the $\epsilon$ dependence in the runtime is $\log(1/\epsilon)$ and so is hidden by the $\tilde O(\cdot )$ notation.
\end{proof}

When $\bv{A}$ is dense, the runtime of Theorem \ref{accPreCondSvrgMainBody} is essentially the best known. Due to its average condition number dependence, the method always outperforms traditional iterative methods, like conjugate gradient, up to $\log$ factors. However, in the sparse case, traditional approaches can give faster runtimes if the rows of $\bv{A}$ are not uniformly sparse and $d_s(\bv{A})$ is large. We have the following, also proved in Appendix \ref{sec:solverAppendix} using the same deflation-based preconditioner as in Theorem \ref{accPreCondSvrgMainBody}:

\begin{theorem}[Ridge Regression -- Preconditioned Iterative Method]\label{preCondIterMainBody} For any $\bv A \in \mathbb{R}^{n \times d}$ and $\lambda > 0$, let $\bv{M}_\lambda \eqdef \bv{A}^T\bv{A} + \lambda \bv{I}$ and $ \hat \kappa \eqdef  \frac{\sigma_{k+1}^2(\bv{A})}{\lambda}$ where $k \in [d]$ is an input parameter.
There is an algorithm that builds a preconditioner for $\bv{M}_\lambda$ using precomputation time $\tilde O(\nnz(\bv{A})k+ dk^{\omega-1})$, and for any input $\bv{y} \in \R^d$, returns
$\bv{x}$ such that with high probability
$\norm{\bv{x}-\bv{M}_\lambda^{-1} \bv{y}}_{\bv{M}_\lambda} \le \epsilon \norm{\bv{y}}_{\bv{M}_\lambda^{-1}}$ in
$\tilde O \left ((\nnz(\bv{A})+dk) \lceil \sqrt{\hat \kappa}\rceil \right )$ time.
\end{theorem}

\subsection{Overall Runtimes For Spectral Windows}
\label{sec:tools:window}

Combined with Lemma~\ref{ridge2step}, the ridge regression
routines above let us efficiently compute soft step functions of $\bv{A}$'s spectrum. Composing step functions  then gives our key computational primitive: the ability to
approximate soft window functions that restrict $\bv{A}$'s spectrum to a specified range. We first define our notion of soft window functions and then discuss  runtimes. The corresponding Theorem~\ref{thm:window} is our main tool for spectrum approximation in the remainder of the paper.

\begin{definition}[Soft Window Function]\label{def:window}
Given $a, b > 0$ with $a < b$, and $\gamma \in [0,1]$,
$h^\gamma_{[a,b]} :[0,1]\rightarrow [0,1]$ is a $\gamma$-soft window for the range $[a,b]$ if:
\begin{align*}
h^\gamma_{[a,b]}(x)  = \begin{cases} 1 \text{ for }x \in [a,b]\\
0 \text{ for }x \in [0,(1-\gamma)a] \cup [(1+\gamma)b,1]
\end{cases}\\\text{ and  }
h^\gamma_{[a,b]}(x)  \in [0,1]\text{ for }x \in [(1-\gamma)a, a] \cup [b,(1+\gamma) b].
\end{align*}
\end{definition}

\begin{theorem}[Spectral Windowing]\label{thm:window} For any $\bv{A} \in \mathbb{R}^{n \times d}$ with $\norm{\bv{A}}_2 \le 1$,  $\bv{y} \in \R^d$, and $a,b,\gamma,\epsilon \in (0,1]$, with $a < b$, there is an algorithm $\algW(\bv{A},\bv{y},a,b,\gamma,\epsilon)$ that returns $\bv{x}$ satisfying w.h.p.: 
$$\norm{\bv{x} - h^\gamma_{[a,b]}(\bv{A}^T\bv{A}) \bv{y}}_2 \le \epsilon \norm{\bv{y}}_2$$
where $h^\gamma_{[a,b]}$ is a soft window function satisfying Def.~\ref{def:window}. Let $\bar \kappa \eqdef \frac{k\sigma_k^2(\bv{A}) + \sum_{i=k+1}^d \sigma_i^2(\bv{A})}{d \cdot a}$ and $\hat \kappa \eqdef \frac{\sigma_{k+1}^2(\bv{A})}{a}$ where $k \in [d]$ is an input parameter.
The algorithm uses precomputation time $\tilde O(\nnz(\bv{A})k + dk^{\omega-1})$ for sparse $\bv{A}$ or $\tilde O(nd^{\omega(\log_d k)-1})$ for dense $\bv{A}$ after which given any $\bv{y}$ it returns $\bv{x}$ in time:
\begin{align*}
\tilde O \left (\frac{\nnz(\bv{A}) + \sqrt{\nnz(\bv{A}) [d\cdot d_s(\bv{A}) + dk ] \bar \kappa}}{\gamma} \right ) \text{\hspace{1em} or \hspace{1em}} \tilde O \left (\frac{(\nnz(\bv{A})+dk) \lceil \sqrt{\hat \kappa}\rceil}{\gamma} \right )
\end{align*} for sparse $\bv{A}$ or $ \tilde O \left (\frac{nd + n^{1/2}d^{3/2} \sqrt{\bar \kappa})}{\gamma} \right )$ for dense $\bv{A}$.
\end{theorem}
\begin{proof}
If $b \ge 1/(1+\gamma)$ then we can simply define $h^\gamma_{[a,b]}(x)  = s_a^{\gamma}(x)$ for any $s_a^\gamma$ satisfying Definition \ref{def:softstep}. Otherwise,
given soft steps $s_a^{\gamma}$ and $s_{(1+\gamma)b}^{\gamma/2}$ satisfying Definition \ref{def:softstep}, we can define $h^\gamma_{[a,b]}(x) = s_a^{\gamma}(x) \cdot (1-s_{(1+\gamma)b}^{\gamma/2} (x))$. Since $\frac{\gamma}{2} \le \frac{\gamma}{1+\gamma}$ we can verify that this will be a valid soft window function for $[a,b]$ (i.e. satisfy Definition \ref{def:window}).
Further, we have for any $\bv{y} \in \R^d$:
\begin{align}\label{windowComputation}
h^\gamma_{[a,b]}(\bv{A}^T\bv{A})\bv{y} = s_a^{\gamma}(\bv{A}^T\bv{A})(\bv{I} - s_{(1+\gamma)b}^{\gamma/2}(\bv{A}^T\bv{A}))\bv{y}.
\end{align}
We can compute $s_a^{\gamma}(\bv{A}^T\bv{A})\bv{y}$ and $s_{(1+\gamma)b}^{\gamma/2}(\bv{A}^T\bv{A})\bv{y}$ each up to error $\epsilon \norm{\bv{y}}_2$ via Lemma \ref{ridge2step}. This gives the error bound in the theorem, since we have both $\norm{s_a^{\gamma}(\bv{A}^T\bv{A})}_2 \le 1$ and $\norm{\bv{I} - s_{(1+\gamma) b}^{\gamma/2}(\bv{A}^T\bv{A})}_2 \le 1$ so the computation in \eqref{windowComputation} does not amplify error. Our runtime follows from combining Theorems \ref{accPreCondSvrgMainBody} and \ref{preCondIterMainBody} with $\lambda = a,b$ with Lemma \ref{ridge2step}. The errors in these theorems are measured with respect to $\norm{\cdot}_{\bv{M}_\lambda}$. To obtain the error in $\norm{\cdot}_2$ as used by Lemma \ref{ridge2step}, we simply apply the theorems with $\epsilon' = \epsilon \kappa(\bv{M}_\lambda)$ which incurs an additional $\log(\kappa(\bv{M}_\lambda))$ cost. Since $a < b$ the runtime is dominated by the computation of $s_a^{\gamma}(\bv{A}^T\bv{A})\bv{y}$, which depends on the condition number $\bar \kappa \eqdef \frac{k\sigma_k^2(\bv{A}) + \sum_{i=k+1}^d \sigma_i^2(\bv{A})}{d \cdot a}$ when using SVRG (Theorem \ref{accPreCondSvrgMainBody}) or $\hat \kappa \eqdef \frac{\sigma_{k+1}^2(\bv{A})}{a}$ for a traditional iterative solver (Theorem \ref{preCondIterMainBody}).
\end{proof}
% !TEX root = normEstimation.tex
\section{Approximating Spectral Sums via Spectral Windows}\label{sec:stepDis}

We now use the window functions discussed in Section \ref{sec:stepApprox} to compute an approximate spectral histogram of $\bv{A}$. 
We give our main histogram algorithm and approximation guarantee in Section \ref{sec:histogram}.
In Section \ref{sec:function_eval} we show how this guarantee translates to accurate spectral sum approximation for any smooth and sufficiently quickly growing function $f(x)$. In Section \ref{sec:specific} we apply this general result to approximating the Schatten $p$-norms for all real $p > 0$, bounded Orlicz norms, and the Ky Fan norms.

\subsection{Approximate Spectral Histogram}\label{sec:histogram}

    Our main histogram approximation method is given as Algorithm \ref{algo:histogram}. The algorithm is reasonably simple. Assuming $\norm{\bv{A}}_2 \le 1$ (this is w.l.o.g. as we can just scale the matrix), and given cutoff $\lambda$, below which we will not evaluate $\bv{A}'s$ spectrum, we split the range $[\lambda,1]$ into successive windows $R_0,...,R_T$ where $R_t = [a_1(1-\alpha)^t, a_1(1-\alpha)^{t-1}]$. Here $\alpha$ determines the width of our windows. In our final spectral approximation algorithms, we will set $\alpha = \Theta(\epsilon)$. $a_1$ is a random shift, which insures that, in expectation, the boundaries of our soft windows do not overlap too many singular values. This argument requires that most of the range $[\lambda,1]$ is not covered by boundary regions. Thus, we set the steepness parameter $\gamma = \Theta(\epsilon_2 \alpha)$ where $\epsilon_2$ will control the error introduced by the boundaries. Finally, we iterate over each window, applying trace estimation to approximate the singular value count in each window.
    
    In our final algorithms, the number of windows and samples required for trace estimation will be $\tilde O(\poly(1/\epsilon))$. The dominant runtime cost will come from computing the window for the lowest range $R_T$, which will incur a dependence on the condition number of $\bv{A}^T\bv{A}+ a_T \bv{I}$ with $a_T = \Theta(\lambda)$.

\begin{algorithm}
\caption{Approximate Spectral Histogram}%[1]
{\bf Input:} $\bv{A}\in\R^{n\times d}$ with $\norm{\bv{A}}_2 \le 1$, accuracy parameters $\epsilon_1,\epsilon_2 \in (0,1)$, width parameter $\alpha \in(0,1)$, and minimum singular value parameter $\lambda \in (0,1)$.
\\ %, $\gamma \in (0,\alpha)$.\\
{\bf Output:} Set of range boundaries $a_{T+1} < a_T < ... < a_1 < a_0$ and counts $\{\tilde b_0,...,\tilde b_T\}$ where $\tilde b_t$ approximates the number of squared singular values of $\bv{A}$ on $[a_{t+1},a_{t}]$. 
\begin{algorithmic}
  \State Set $\gamma = c_1 \epsilon_2 \alpha$, $T = \lceil \log_{(1-\alpha)} \lambda \rceil$, and $S = \frac{\log n}{c_2\epsilon_1^2}.$
  
  \State Set $a_0 = 1$ and choose $a_1$ uniformly at random in $[1-\alpha/4,1]$. 
  \State Set $a_t = a_1(1-\alpha)^{t-1}$ for $2 \le t \le T+1$.
  \For {$t = 0 :T$} \Comment{\textcolor{blue}{Iterate over histogram buckets.}}
  	\State Set $\tilde b_t = 0.$ \Comment{\textcolor{blue}{Initialize bucket size estimate.}}
  	\For{$s = 1:S$} \Comment{\textcolor{blue}{Estimate bucket size via trace estimation.}}
		\State Choose $\bv{y} \in \{-1,1\}^d$ uniformly at random.
  		\State Set $\tilde b_t = \tilde b_t + \frac{1}{S} \cdot \bv{y}^T\algW(\bv{A}^T\bv{A},\bv{y},a_{t+1},a_{t},\gamma,c_3\epsilon_1^2/n).$  \Comment{\textcolor{blue}{Apply soft window via Thm \ref{thm:window}.}}
		\State If $\tilde b_t \le 1/2$ set $\tilde b_t = 0$.\Comment{\textcolor{blue}{Round small estimates to ensure relative error.}}
	\EndFor
  \EndFor\\
 \Return $a_1$ and $\tilde b_t$ for $t=0:T$. \Comment{\textcolor{blue}{Output histogram representation.}}
\end{algorithmic}
\label{algo:histogram}
\end{algorithm}

\begin{theorem}[Histogram Approximation]\label{thm:histogram} Let $a_1$, $\tilde b_0,...,\tilde b_{T}$ be output by Algorithm \ref{algo:histogram}. Let $R_0 = [a_1,1]$, $R_t = [a_1(1-\alpha)^t,a_1(1-\alpha)^{t-1}]$ for $t \ge 1$, and $b_t = \left |\{i: \sigma_i^2(\bv{A}) \in R_t \}\right |$ be the number of squared singular values of $\bv{A}$ on the range $R_t$. Then, for sufficiently small constants $c_1,c_2,c_3$, with probability $99/100$, for all $t \in \{0,...,\lceil \log_{(1-\alpha)}  \lambda \rceil  \}$, $\tilde b_t$ output by Algorithm \ref{algo:histogram} satisfies:
\begin{align*}
(1-\epsilon_1) b_t \le \tilde b_t \le (1+\epsilon_1) b_t + \lceil \log_{(1-\alpha)}  \lambda \rceil \cdot \epsilon_2 (b_{t-1} + b_t + b_{t+1}).
\end{align*}
\end{theorem}
That is, 
 $\tilde b_t$ approximates the number of singular values of the range $R_t$ up to multiplicative $(1\pm \epsilon_1)$ error and additive error $\lceil \log_{(1-\alpha)} \lambda \rceil\cdot \epsilon_2 (b_{t-1} + b_t + b_{t+1})$. Note that by setting $\epsilon_2 \le \frac{\epsilon_1}{\lceil \log_{(1-\alpha)} \lambda \rceil}$, the error on each bucket is just multiplicative on its size plus the size of the two adjacent buckets, which contain singular values in nearby ranges. For simplicity we assume $\bv{A}$ passed to the algorithm has $\norm{\bv{A}}_2 \le 1$. This is without loss of generality: we can estimate $\norm{\bv{A}}_2$ in $\tilde O(\nnz(\bv{A}))$ time via the power or Lanczos methods \cite{kuczynski1992estimating,musco2015randomized}, and scale down the matrix appropriately. 
 
 Ignoring logarithmic and $\epsilon$ dependencies, the runtime of Algorithm \ref{algo:histogram} is dominated by the calls to $\algW$ for the bucket corresponding to the smallest singular values, with $a_{T} = \Theta(\lambda)$. This runtime is given by Theorem \ref{thm:window}. Since balancing the deflation parameter $k$ in that theorem with the minimum squared singular value $\lambda$ considered can be complex, we wait to instantiate full runtimes until employing Algorithm \ref{algo:histogram} for specific spectral sum computations.
 \begin{proof}
 We use the notation of Algorithm \ref{algo:histogram}, where  $a_0 = 1$, $a_1$ is chosen uniformly in $[1-\alpha/4,1]$ and $a_t = a_1(1-\alpha)^{t-1}$. With this notation, we have $R_t = [a_{t+1},a_t]$.
 Let $h^{\gamma}_{R_t}$ be a $\gamma$-soft window for $R_t$ (Definition \ref{def:window}) and let $\bar R_t =  [(1-\gamma)a_{t+1},(1+ \gamma)a_t]$ be the interval on which $h^\gamma_{R_t}$ is nonzero. $\tilde b_t$ is an estimation of the trace of such a window applied to $\bv{A}^T\bv{A}$. We first show that, if these traces are computed exactly, they accurately estimate the singular value counts in each range $R_t$. We have:
  \begin{align}\label{cnm1}
	\tr(h^{\gamma}_{R_t}(\bv{A}^T \bv{A})) &= \sum_{\sigma_i^2(\bv{A}) \in R_t} h^{\gamma}_{R_t}(\sigma_i^2(\bv{A})) + \sum_{\sigma_i^2(\bv{A}) \in \bar R_t \setminus R_t} h^{\gamma}_{ R_t}(\sigma_i^2(\bv{A}))\nonumber\\
	&= b_t +  \sum_{\sigma_i^2(\bv{A}) \in \bar R_t \setminus R_t} h^{\gamma}_{ R_t}(\sigma_i^2(\bv{A})).
\end{align}
	
We can bound the second term using the random shift $a_1$. Since $\gamma = c_1\epsilon_2 \alpha < \alpha$, each singular value falls within at most two extended ranges: $\bar R_t$ and $\bar R_{t \pm  1}$ for some $t$. Let $I$ be the set of indices whose singular values fall within two ranges.
We have $i \in I$ only if $a_1(1-\alpha)^t \in (1\pm \gamma) \sigma_i^2(\bv{A})$ for some $t$. Letting $d = \lceil \log_{\sigma_i^2(\bv{A})} (1-\alpha) \rceil$, this holds only if 
	$a_1 \in  (1\pm \gamma ) \left (\frac{\sigma_i^2(\bv{A})}{(1-\alpha)^d} \right )$, which occurs with probability at  most $\frac{8\gamma}{\alpha}$ since $a_1$ is chosen uniformly in the range $[1-\alpha/4,1]$. Thus we have $ \Pr[i\in I] \le \frac{8\gamma}{\alpha}$. 
	Further, $a_{t} \eqdef a_1(1-\alpha)^{t-1}$ is distributed uniformly in the range $[(1-\alpha/4)(1-\alpha)^{t-1},(1-\alpha)^{t-1}]$, so we know for certain that if the constant $c_1$ on $\gamma$ is set small enough:
	\begin{align}\label{fixingit}
	\bar R_t &\subset [(1-\gamma)(1-\alpha/4)(1-\alpha)^{t},(1+\gamma)(1-\alpha)^{t-1}]\nonumber\\
	&\subset [(1-\alpha/3)(1-\alpha)^{t}, (1+\alpha/3)(1-\alpha)^{t-1}]\nonumber\\
	&\subset [(1-\alpha)^{t+1}, (1-\alpha/4)(1-\alpha)^{t-2}].
	\end{align}
Let $M_t \eqdef [(1-\alpha)^{t+1}, (1-\alpha/4)(1-\alpha)^{t-2}]$. Note that $M_t$ is fixed (i.e. not a random variable).
\emph{Regardless of the random shift $a_1$}, by \eqref{fixingit}, we always have $\bar R_t \subset M_t$. We also have $M_t \subset R_{t-1}\cup R_t \cup R_{t+1}$.

%We also can see that $M_t \subset R_{t-1} \cup R_t \cup R_{t+1}$. 

Let $\mathbb{I}[i \in I]$ be $1$ if $i \in I$ and $0$ otherwise.
We can upper bound the second term of \eqref{cnm1} by:% replacing $\bar R_t$ with $M_t$. That is:
\begin{align}\label{eq:mbound1}
\sum_{\sigma_i^2(\bv{A}) \in \bar R_t \setminus R_t} h^{\gamma}_{ R_t}(\sigma_i^2(\bv{A})) &\le
\sum_{\sigma_i^2(\bv{A}) \in \bar R_{t}\setminus R_t} \mathbb{I}[i\in I]\nonumber\\
& \le \sum_{\sigma_i^2(\bv{A}) \in M_t \setminus R_t} \mathbb{I}[i\in I] \nonumber\\
& \le  \sum_{\sigma_i^2(\bv{A}) \in M_t} \mathbb{I}[i\in I].
\end{align}
The first bound follows from the fact that for $\sigma_i(\bv{A})^2  \in \bar R_t \setminus R_t$, we have $\sigma_i(\bv{A})^2 \in I$ and that $h^{\gamma}_{ R_t}(\sigma_i^2(\bv{A})) \le 1$.
The second bound follows from \eqref{fixingit}, which shows that $\bar R_t \subset M_t$. Let $m_t = |\{i: \sigma_i^2(\bv{A}) \in M_t\}|$ be the number of squared singular values falling in $M_t$. Note that like $M_t$, $m_t$ is fixed (i.e., not a random variable.) Thus, by  linearity of expectation, we have:
\begin{align*}
\E \left [ \sum_{\sigma_i^2(\bv{A}) \in M_t} \mathbb{I}[i\in I] \right ] = \sum_{\sigma_i^2(\bv{A}) \in M_t} \Pr[i\in I] = \frac{8\gamma}{\alpha} \cdot m_t.
\end{align*}
%
%\begin{align*}
%\E \left [ \sum_{\sigma_i^2(\bv{A}) \in \bar R_t \setminus R_t} h^{\gamma}_{ R_t}(\sigma_i^2(\bv{A})) \right ]\le  \sum_{\sigma_i^2(\bv{A}) \in R_{t-1} \cup R_{t+1}} \Pr[i\in I] \le \frac{8\gamma}{\alpha}(b_{t-1} + b_{t+1}).
%\end{align*} 
Letting $T = \lceil \log_{(1-\alpha)} \lambda \rceil$ as in Algorithm \ref{algo:histogram}, by a Markov bound, with probability $1-\frac{1}{200(T+1)}$:
\begin{align*}
\sum_{\sigma_i^2(\bv{A}) \in M_t} \mathbb{I}[i\in I] \le \frac{1600(T+1) \gamma}{\alpha}\cdot m_t \le T \epsilon_2 \cdot m_t
\end{align*}
if $c_1$ is set small enough. By a union bound this holds for all $t \in \{0,...,T \}$ simultaneously with probability $\ge 199/200$. Plugging back into \eqref{eq:mbound1}, we have, with probability  $\ge 199/200$, simultaneously  for all $t$:
\begin{align}\label{newBound2}
\sum_{\sigma_i^2(\bv{A}) \in \bar R_t \setminus R_t} h^{\gamma}_{ R_t}(\sigma_i^2(\bv{A}))&\le T \epsilon_2 \cdot m_t\nonumber\\ &\le T\epsilon_2 \cdot (b_{t-1}+b_t+b_{t+1})
\end{align}
where the second bound follows from the fact that, regardless of the setting of the shift $a_1$, $M_t \subset R_{t-1}\cup R_t \cup R_{t+1}$ so $m_t \le (b_{t-1}+b_t+b_{t+1})$. Plugging \eqref{newBound2} into \eqref{cnm1} we have:
\begin{align}\label{cnm2}
b_t \le \tr(h^{\gamma}_{R_t}(\bv{A}^T \bv{A})) \le b_t + T \epsilon_2 (b_{t-1} + b_t + b_{t+1}).
\end{align}
We conclude by showing that, before the final rounding step of Algorithm \ref{algo:histogram}, $\tilde b_t \in (1\pm \epsilon_1) \tr(h^{\gamma}_{R_t}(\bv{A}^T \bv{A})) + \sqrt{c_3} \epsilon_1$ with high probability. In the final rounding step, if $\tr(h^{\gamma}_{R_t}(\bv{A}^T \bv{A})) \le 1/4$ (which can only occur if $b_t = 0$) and $c_1,c_3$ are sufficiently small, then we will have $\tilde b_t \le 1/2$ and so will round down to $\tilde b_t = 0 = b_t$. Otherwise, the $\sqrt{c_3} \epsilon_1$ term will be absorbed into the relative error on $\tr(h^{\gamma}_{R_t}(\bv{A}^T \bv{A}))$.

Overall, combined with  \eqref{cnm2} we will have, with probability $\ge 99/100$ for all $t$: 
\begin{align*}
(1-\epsilon) b_t \le \tilde b_t \le (1+2\epsilon_1)b_t + (1+2\epsilon_1) T \epsilon_2 (b_{t-1} + b_t + b_{t+1})
\end{align*}
which gives the theorem if we adjust $\epsilon_1,\epsilon_2$ by making $c_1,c_2,c_3$ sufficiently small.% and make $c_3$ small enough. The $\sqrt{c_3}\epsilon_1$ term will either be absorbed into relative error if $b_t \ge 1$ and hence $ \tr(h^{\gamma}_{R_t}(\bv{A}^T \bv{A})) \ge 1$, or will be $\le 1/2$ and so $\tilde b_t$ will be rounded down to $0$ in the rounding step.
Thus we conclude by showing that in fact $\tilde b_t \in (1\pm \epsilon_1) \tr(h^{\gamma}_{R_t}(\bv{A}^T \bv{A})) + \sqrt{c_3} \epsilon_1$ with high probability.
Setting $S = \frac{\log n}{c_2 \epsilon_1^2}$ as in Algorithm \ref{algo:histogram}, for $\bv{y}_1,\ldots,\bv{y}_S$ chosen 
from $\{-1,1\}^d$, with high probability $\frac{1}{S} \sum_{i=1}^S \bv{y}_i^T h_{R_t}^\gamma (\bv{A}^T\bv{A}) \bv{y}_i \in (1\pm \epsilon_1) \tr(h_{R_t}^\gamma (\bv{A}^T\bv{A}))$ by a standard trace estimation result \cite{avron2011randomized}. 

Further, let $\bv{x}_i = \algW(\bv{A}^T\bv{A},\bv{y},a_{t+1},a_{t},\gamma,c_3\epsilon_1^2/n) \bv{y}_i$. By Theorem \ref{thm:window}: $\norm{\bv{x}_i - h_{R_t}^\gamma (\bv{A}^T\bv{A}) \bv{y}_i}_2 \le \frac{c_3\epsilon_1^2}{n^2} \norm{\bv{y}_i}_2$ which by Cauchy-Schwarz gives $| \bv{y}_i^T\bv{x}_i - \bv{y}_i^Th_{R_t}^\gamma (\bv{A}^T\bv{A})\bv{y}_i | \le  \sqrt{c_3/n}\epsilon_1 \norm{\bv{y}_i}_2 = \sqrt{c_3} \epsilon_1$ since $\norm{\bv{y}_i}_2 = \sqrt{n}$. Thus, overall we have, before the rounding step in which $\tilde b_t$ is set to $0$ if $\tilde b_t < 1/2$:
\begin{align*}
\tilde b_t &= \frac{1}{S} \sum_{i=1}^S \bv{y}_i^T \algW(\bv{A}^T\bv{A},\bv{y},a_{t+1},a_{t},\gamma,c_3\epsilon_1^2/n) \bv{y}_i \\
&\in (1\pm \epsilon_1) \tr(h^{\gamma}_{R_t}(\bv{A}^T \bv{A})) + \sqrt{c_3}\epsilon_1.
\end{align*}
 \end{proof}

\subsection{Application to General Spectral Sums}\label{sec:function_eval}

While Theorem \ref{thm:histogram} is useful in its own right, we now apply it to approximate a broad class of spectral sums. We need two assumptions on the sums that we approximate. First, for the histogram discretization to be relatively accurate, we need our function to be  relatively smooth. Second, it is expensive to compute the histogram over very small singular values of $\bv{A}$ (i.e. with $\lambda$ very small in Algorithm \ref{algo:histogram}) as this makes the condition number in Theorem \ref{thm:window} large. So it is important that small singular values cannot contribute significantly to our sum.
We start with the following definition:

\begin{definition}[Multiplicative Smoothness]\label{def:smooth} $f: \R^+ \rightarrow \R^+$ is $\delta_f$-multiplicatively smooth 
if for some $\delta_f \ge 1$, for all $x$, $|f'(x)| \le \delta_f \frac{f(x)}{x}$.
\end{definition}

We have the following claim, proven in Appendix \ref{sec:generalAppendix}.
\begin{claim}\label{claim:hypo}
	Let $f: \R^+ \rightarrow \R^+$ be a $\delta_f$-multiplicatively smooth function. For all $x,y \in \mathbb{R}^+$ and $c \in (0,\frac{1}{3\delta_f})$ $$y \in \left [(1-c)x, (1+c)x\right] \Rightarrow f(y) \in [(1-3\delta_f c) f(x), (1+3\delta_f c) f(x)].$$
\end{claim}

For the example of the Schatten $p$-norm, for $f(x) = x^p$, $f'(x) = p \cdot x^{p-1}$ and so $f(x)$ is $p$-multiplicatively smooth. 

We now give our general approximation theorem, showing that any spectral sum depending on sufficiently smooth and rapidly growing $f$ can be computed using Algorithm \ref{algo:histogram}:

\begin{framed}
\begin{theorem}[Spectral Sums Via Approximate Histogram]\label{thm:histogram_algorithm} Consider any $\bv{A} \in \mathbb{R}^{n \times d}$ and any function $f: \R^+ \rightarrow \R^+$ satisfying:
	\begin{itemize}
		\item Multiplicative Smoothness: For some $\delta_f \ge 1$, $f$ is $\delta_f$-multiplicatively smooth (Defn. \ref{def:smooth}).
		\item Small Tail: For any $\epsilon > 0$ there exists $\lambda_f(\epsilon)$ such that for $x \in [0,\lambda_f(\epsilon)]$, $f(x) \le \frac{\epsilon}{n}\Sum_f(\bv{A})$
	\end{itemize}
	Given error parameter $\epsilon \in (0,1)$ and spectral norm estimate $M \in [ \norm{\bv{A}}_2, 2\norm{\bv{A}}_2]$, for sufficiently small constant $c$, if we run Algorithm \ref{algo:histogram} on $\frac{1}{M} \bv{A}$ with input parameters $\epsilon_1, \epsilon_2 = c \epsilon$, $\alpha = c\epsilon/\delta_f$ and $\lambda = \lambda_f(c\epsilon)^2/M^2$ then with probability $99/100$, letting  $a_1$, $\tilde b_0,...,\tilde b_{T}$ be the outputs of the algorithm and $g(x) = f(x^{1/2})$:
	\begin{align*}
	(1-\epsilon) \Sum_f(\bv{A}) \le \sum_{t=0}^{T} g( M^2\cdot a_1(1-\alpha)^t) \cdot \tilde b_t \le (1+\epsilon) \Sum_f(\bv{A}).
	\end{align*}
For parameter $k \in [d]$, letting $\bar \kappa \eqdef \frac{k\sigma_k^2(\bv{A}) + \sum_{i=k+1}^d \sigma_i^2(\bv{A})}{d \cdot \lambda}$ and $\hat \kappa \eqdef \frac{\sigma_{k+1}^2(\bv{A})}{\lambda}$, the algorithm runs in
\begin{align*}
\tilde O \left (\nnz(\bv{A}) k + dk^{\omega -1} + \frac{\nnz(\bv{A}) + \sqrt{\nnz(\bv{A}) [d\cdot d_s(\bv{A}) + dk ] \bar \kappa}}{\epsilon^5/(\delta_f^2 \log(1/\lambda))} \right ) \\\text{ or } \tilde O \left (\nnz(\bv{A}) k + dk^{\omega -1} + \frac{(\nnz(\bv{A})+dk) \lceil \sqrt{\hat \kappa}\rceil}{\epsilon^5/(\delta_f^2 \log(1/\lambda))} \right )
\end{align*} time for sparse $\bv{A}$ or 
$ \tilde O \left (nd^{\omega(\log_d k)-1} + \frac{nd + n^{1/2}d^{3/2} \sqrt{\bar \kappa})}{\epsilon^5/(\delta_f^2 \log(1/\lambda))} \right )$ for dense $\bv{A}$.
\end{theorem}
\end{framed}

That is, we accurately approximate $\Sum_f(\bv{A})$ by discretizing over the histogram output by Algorithm \ref{algo:histogram}. Note that we can boost our probability of success to  $1-\delta$ by repeating the algorithm $\Theta(\log(1/\delta))$ times and taking the median of the outputs. 
\begin{proof} Let $\bv{\bar A} \eqdef \frac{1}{M} \bv{A}$. Note that $\norm{\bv{\bar A}}_2 \le 1$ and so it is a valid matrix on which to apply Algorithm \ref{algo:histogram} and Theorem \ref{thm:histogram}. Recall that we use the notation: $R_0 = [a_1,1]$,
$R_t = [a_1(1-\alpha)^t,a_1(1-\alpha)^{t-1}]$ for $t \ge 1$, $T = \lceil \log_{(1-\alpha)} \lambda \rceil$, and $b_t = \left |\{i: \sigma_i^2(\bv{\bar A}) \in R_t \} \right |$.

Since $g(x) = f(x^{1/2})$, by chain rule:
$$g'(x) = \frac{f'(x^{1/2})}{2x^{1/2}} \le \delta_f \frac{f(x^{1/2})}{2x} = \frac{\delta_f}{2} \frac{g(x)}{x}.$$
So $g$ is $\delta_f/2$ multiplicatively smooth. 
By this smoothness, and Claim \ref{claim:hypo}, for any $i$ with $\sigma_i^2(\bv{\bar A}) \in R_t$:
\begin{align*}
  g(M^2 \cdot a_1(1-\alpha)^t) &\in \left (1 \pm \frac{3\delta_f \alpha}{2}\right )  g( \sigma_i^2(\bv{A}))\\
 &\in (1\pm \epsilon/4) \cdot f( \sigma_i(\bv{A}))
\end{align*}
if the constant $c$ on $\alpha$ is set small enough. Small enough $c$ also ensures $ g(M^2a_1(1-\alpha)^{t-1}) \in (1\pm\epsilon/4) f( \sigma_i(\bv{A}))$ and $g(M^2a_1(1-\alpha)^{t+1})  \in (1\pm \epsilon/4) f( \sigma_i(\bv{A}))$. So, both the multiplicative and additive error terms in Theorem \ref{thm:histogram} do not hurt us significantly. For now, assume that $\epsilon_2 = c\epsilon/T$. We will first prove the result with this assumption and then show that it can be relaxed to $\epsilon_2 = c\epsilon$ as given in the theorem statement. Let $I$ be the set of indices with $\sigma_i^2(\bv{\bar A}) \in R_t$ for some $t \in \{0,...,T\}$. That is, the set of singular values covered by our histogram. Let $I_{T+1}$ be the set of indices with $\sigma_i^2(\bv{\bar A}) \in R_{T+1}$ -- that is, singular values which are not included in the histogram, but may be included in the additive error for the last bucket covering $R_{T}$. Applying Theorem \ref{thm:histogram}:
\begin{align}
 \sum_{t=0}^{T} g(M^2a_1(1-\alpha)^t) \cdot \tilde b_t &\le \sum_{t=0}^{T} g(M^2a_1(1-\alpha)^t) \cdot \left [(1+\epsilon_1) b_t + \epsilon_2 T(b_{t-1} + b_t+ b_{t+1}) \right ]\nonumber\\
 &\le (1+\epsilon/4 ) \sum_{t=0}^{T+1} \left [ (1+c \epsilon )\sum_{\sigma_i^2(\bv{\bar A}) \in R_t} f( \sigma_i(\bv{A})) +  3c\epsilon \sum_{\sigma_i^2(\bv{\bar A}) \in R_{t}} f( \sigma_i(\bv{A}))  \right ] \nonumber\\
 &\le (1+\epsilon) \sum_{i \in I \cup I_{T+1}} f( \sigma_i(\bv{A})) \nonumber\\
 &\le (1+\epsilon) \Sum_f(\bv{A})\label{camupper}
 \end{align}
 if we set $c$ small enough. The second inequality arises because each $\sigma_i(\bv{A})$ contributing to $b_t$ appears at most three times as an additive error term for $\tilde b_{t-1}$, $\tilde b_t$, and $\tilde b_{t+1}$. In this inequality we include bucket $R_{T+1}$ in the histogram, which only increases the right hand side.
 
 On the lower bound side, we use our small tail assumption, that for $i$ with $\sigma_i(\bv{A}) < \lambda_f(c\epsilon)$, $f(\sigma_i(\bv{A})) \le \frac{c\epsilon}{n} \Sum_f(\bv{A})$. Using the notation of Algorithm \ref{algo:histogram}, we have $a_{T+1} = a_1 (1-\alpha)^T \le (1-\alpha)^{\lceil \log_{(1-\alpha)} \lambda \rceil} \le \lambda \eqdef \lambda_f(c\epsilon)^2/M^2$. So for any $i$ with $\sigma_i(\bv{A}) \ge \lambda_f(c\epsilon)$, $\sigma_i^2(\bv{\bar A}) > a_{T+1}$ and thus falls in some bucket of our histogram so $i \in I$. We thus have:
 \begin{align*}
 \sum_{i\in I} f( \sigma_i(\bv{A})) \ge \Sum_f(\bv{A}) - \sum_{i: \sigma_i(\bv{A}) < \lambda_f(c\epsilon)} f( \sigma_i(\bv{A}))\\
 \ge  \Sum_f(\bv{A})  - n \cdot \frac{c\epsilon}{n}  \Sum_f(\bv{A}) = (1-c\epsilon)  \Sum_f(\bv{A}) .
 \end{align*}
 Applying Theorem \ref{thm:histogram} again, we have for sufficiently small $c$:
 \begin{align*}
 \sum_{t=0}^{T} g(a_1(1-\alpha)^t) \cdot \tilde b_t  &\ge (1-\epsilon/4 )(1-c \epsilon )\sum_{t=0}^{T} \sum_{\sigma_i^2(\bv{\bar A}) \in R_t} f( \sigma_i(\bv{A}))\\
 &= (1-\epsilon/4 )(1-c \epsilon )  \sum_{i\in I} f( \sigma_i(\bv{A}))\\
 & \ge (1-\epsilon)  \Sum_f(\bv{A}).
 \end{align*}
 We conclude the theorem by noting that we can actually set $\epsilon_2 = c\epsilon$ instead of $\epsilon_2 = c\epsilon/T$ as used above. The additive error term on each bucket was bounded using a Markov bound, and to union bound over $T$ buckets, we lost a factor of $T$. However, in expectation, the total contribution of additive error to our spectral sum estimation is just $O \left (\epsilon \sum_{i \in I \cup I_{T+1}} f(\sigma_1(\bv{A}))\right) = O \left (\Sum_f(\bv{A}) \right )$ and so by a Markov bound is $\le \epsilon/2 \Sum_f(\bv{A})$ with probability $99/100$ if we set our constants small enough.
 
It just remains to discuss runtime, i.e. to calculate the runtime of Algorithm \ref{algo:histogram} with inputs $\epsilon_1,\epsilon_2 = c\epsilon$, $\alpha = c\epsilon/\delta_f$ and $\lambda = \lambda_f(c\epsilon)^2/M^2$. The number of outer loops of the algorithm is $T = \lceil \log_{(1-\alpha)} \lambda \rceil = \Theta \left (\frac{\delta_f}{\epsilon } \log(1/\lambda) \right )$. The number of inner loops is $S = \tilde O(1/\epsilon^2)$. And finally, within each loop the most expensive operation is computing $\algW(\bv{A}^T\bv{A}, \bv{y}, a_{t+1}, a_t, \gamma, c_3 \epsilon_1^2/n)$. Our final runtime follows from plugging this into Theorem \ref{thm:window} noting that $\gamma = \Theta(\epsilon \alpha) = \Theta(\epsilon^2/\delta_f)$ and $a_{t+1} = \Omega(\lambda)$ for all $t \in  \{ 0,...,T \}$. Note that we perform the precomputation step to construct a preconditioner for $\bv{A}^T\bv{A} + \lambda \bv{I}$ just once, incurring cost $\tilde O(\nnz(\bv{A}) + dk^{\omega-1})$ or $\tilde O (nd^{\omega(\log_d k)-1})$.
\end{proof}

\subsection{Application to Schatten $p$, Orlicz, and Ky Fan Norms}\label{sec:specific}

Theorem \ref{thm:histogram_algorithm} is very general, allowing us to approximate any function satisfying a simple smoothness condition as long as the smaller singular values of $\bv{A}$ cannot contribute significantly to $\Sum_f(\bv{A})$. We now give some specific applications of the result.

\subsubsection*{Schatten $p$-Norms} 

Theorem \ref{thm:histogram_algorithm} already gives the fastest known algorithms for Schatten $p$-norm estimation. We will not go into all runtime tradeoffs now as our best runtimes will be worked out in detail in Sections \ref{sec:poly} and \ref{sec:upper}, but as an example:

\begin{corollary}[Schatten $p$-norms via Histogram Approximation]\label{cor:hist_schatten}
For any $\bv{A} \in \mathbb{R}^{n \times n}$ with uniformly sparse rows (i.e. $d_s(\bv{A}) = O(\nnz(\bv{A})/n)$), given error parameter $\epsilon \in (0,1)$ and $M \in [ \norm{\bv{A}}_2, 2\norm{\bv{A}}_2]$, if we run Algorithm \ref{algo:histogram} on $\frac{1}{M} \bv{A}$ with $\epsilon_1,\epsilon_2 = c\epsilon$, $\alpha = c\epsilon/ \max\{1,p\}$ and $\lambda = \frac{1}{M^2} \left (\frac{c\epsilon}{n} \norm{\bv{A}}^p_p \right )^{2/p}$ for sufficiently small constant $c$ then with probability $99/100$, letting $a_1, \tilde b_0,...,\tilde b_{T}$ be the outputs of the algorithm we have:
\begin{align*}
(1-\epsilon) \norm{\bv{A}}_p^p \le \sum_{t=0}^{T} \left [M^2 a_1(1-\alpha)^t\right ]^{p/2} \cdot \tilde b_t \le (1+\epsilon) \norm{\bv{A}}_p^p.
\end{align*}
Further the algorithm runs in time:
\begin{align*}
\tilde O\left (\frac{\nnz(\bv{A})p^2}{\epsilon^{5+1/p}}\right) \text{ for } p\ge 2 \text{\hspace{1em} and \hspace{1em}}
\tilde O \left (\frac{\nnz(\bv{A}) n^{\frac{1/p-1/2}{1/p+1/2}}+n^{\frac{5/p-1/2}{2/p+1}}\sqrt{d_s(\bv{A})}}{p\cdot \epsilon^{5+1/p}} \right )
\text{ for }p \le 2.
\end{align*}
For dense inputs this can be sped up to $\tilde O \left (\frac{n^{\frac{2.3729 - .1171p}{1 + .0346p}}}{p\cdot \epsilon^{5+1/p}} \right )$ using fast matrix multiplication or $\tilde O \left (\frac{n^{\frac{3+p/2}{1+p/2}}}{p\cdot \epsilon^{5+1/p}} \right )$ without fast matrix multiplication.
\end{corollary}
For constant $\epsilon, p>2$ the first runtime is $\tilde O(\nnz(\bv{A}))$, and for the nuclear norm ($p=1$), for constant $\epsilon$ the second runtime gives $\tilde O(\nnz(\bv{A})n^{1/3} + n^{3/2}\sqrt{d_s(\bv{A})})$ which is at worst $\tilde O(\nnz(\bv{A})n^{1/3} + n^{2})$. For dense matrices, the nuclear norm estimation time is $\tilde O(n^{2.18})$ using fast matrix multiplication. It is already $\tilde O(n^{2.33}) = o(n^\omega)$, for the current $\omega$, without using fast matrix multiplication. 

Note that we can compute the spectral norm approximation used to scale $\bv{A}$ via the Lanczos or power method in $\tilde O(\nnz(\bv{A}))$ time. $\lambda$ depends on $\norm{\bv{A}}_p$ which we are estimating. However, as we will discuss in the proof, we can use a rough estimate for $\norm{\bv{A}}_p$ which suffices. As a more general strategy, $\lambda$ can be identified via binary search. We can start with $\lambda = \sigma_1(\bv{A})^2/M^2$ and successively decrease $\lambda$ running Algorithm \ref{algo:histogram} up to the stated runtime bounds. If it does not finish in the allotted time, we know that we have set $\lambda$ too small. Thus, we can output the result with the smallest $\lambda$ such that the algorithm completes within in the stated bounds.
\begin{proof}
We invoke Theorem \ref{thm:histogram_algorithm} with $f(x) = x^p$. 
We have $f'(x) = p\frac{f(x)}{x}$ so $\delta_f = \max\{1,p\}$ and our setting of $\alpha = c\epsilon/\max\{1,p\}$ suffices. Additionally, for any $c$, we can set $
\lambda_f(c\epsilon) = \left ( \frac{c \epsilon}{n} \norm{\bv{A}}_p^p \right )^{1/p} = \frac{c^{1/p} \epsilon^{1/p}}{n^{1/p}} \norm{\bv{A}}_p$ and so our setting of $\lambda$ suffices. Thus the accuracy bound follows from Theorem \ref{thm:histogram_algorithm}. 

We now consider runtime. 
For $p \ge 2$:
\begin{align*}
\bar \kappa = \frac{k \sigma_k^2(\bv{A}) + \sum_{i=k+1}^n \sigma_i^2(\bv{A})}{n \lambda} \le \frac{n^{2/p-1}}{\epsilon^{2/p}} \cdot \frac{\norm{\bv{A}}_F^2}{\norm{\bv{A}}_p^2}.
\end{align*}
We can bound $\norm{\bv{A}}_F \le n^{1/2-1/p}\norm{\bv{A}}_p$ and so have $\bar \kappa \le \frac{1}{\epsilon^{2/p}}.$ For $p < 2$ we have:
\begin{align*}
\bar \kappa = \frac{n^{2/p-1}}{\epsilon^{2/p}} \cdot \frac{k \sigma_k^2(\bv{A}) + \sum_{i=k+1}^n \sigma_i^2(\bv{A})}{\norm{\bv{A}}_p^2} \le \frac{n^{2/p-1}}{\epsilon^{2/p}} \cdot \frac{\sigma_k^{2-p}(\bv{A}) \sum_{i=1}^n \sigma_i^p(\bv{A})}{\norm{\bv{A}}_p^2} = \frac{n^{2/p-1}}{\epsilon^{2/p}} \cdot \frac{\sigma_k^{2-p}(\bv{A})}{\norm{\bv{A}}_p^{2-p}}.
\end{align*}
Using the fact that $\sigma_k^p(\bv{A}) \le \frac{1}{k} \norm{\bv{A}}_p^p$ we have the tradeoff between $k$ and $\bar \kappa$:
\begin{align}
\bar \kappa \le \frac{1}{\epsilon^{2/p}} \left ( \frac{n}{k} \right )^{2/p-1}.\label{kKappa}
\end{align}

As mentioned, $\lambda$ depends on the value of $\norm{\bv{A}}_p $. We can simply lower bound $\norm{\bv{A}}_p^p$ by $k\sigma_k^p(\bv{A})$, which we estimate up to multiplicative error when performing deflation. We can use this lower bound to set $\lambda$. Our estimated $\lambda$ will only be smaller than the true  value, giving a better approximation guarantee and the above condition number bound will still hold.

We now analyze our runtime. Recall that for $f(x) = x^p$, $\delta_f = \max \{1,p \}$. Correspondingly, $\log(1/\lambda) = \tilde O( \max \{1, 1/p \})$ and so $\delta_f^2 \log(1/\lambda)= \max \{ p^2,1/p \}$. 
Plugging into the first runtime of Theorem \ref{thm:histogram_algorithm}, using the uniform sparsity assumption and the fact that $\sqrt{x+y} \le \sqrt{x} + \sqrt{y}$ we have:
$$\tilde O \left (\nnz(\bv{A}) k + nk^{\omega -1} + \frac{\nnz(\bv{A})\sqrt{\bar \kappa} + \sqrt{\nnz(\bv{A})nk \bar \kappa}}{\epsilon^5/(\max \{ p^2,1/p \}} \right ).$$
For $p \ge 2$ we just set $k=0$ and have $\tilde O(\nnz(\bv{A})p^2/\epsilon^{5+1/p})$ runtime by our bound $\bar \kappa \le \frac{1}{\epsilon^{2/p}}$. For $p \le 2$, not trying to optimize $\poly(1/\epsilon)$ terms, we write the runtime as 
\begin{align*}
\tilde O \left (nd_s(\bv{A}) k + nk^{\omega -1} + \frac{nd_s(\bv{A})\sqrt{\bar \kappa} + n\sqrt{d_s(\bv{A})k \bar \kappa}}{\epsilon^5p} \right )
\end{align*}
and balancing the first two coefficients on $n$, set $k = n^{\frac{1/p-1/2}{1/p+1/2}}$ which gives $\sqrt{\bar \kappa} = n^{\frac{1/p-1/2}{1/p+1/2}}$ by \eqref{kKappa} and so $nd_s(\bv{A})k = nd_s(\bv{A}) \sqrt{\bar \kappa}$. We then have $n\sqrt{d_s(\bv{A}) k \bar \kappa} = n\sqrt{d_s(\bv{A})} k^{3/2} = n^{\frac{5/p-1/2}{2/p+1}}\sqrt{d_s(\bv{A})} $. Finally, the $nk^{\omega-1}$ is dominated by the $n\sqrt{d_s(\bv{A})} k^{3/2}$ term so we drop it.

Finally, for dense $\bv{A}$ we apply the third runtime which gives
\begin{align*}
\tilde O \left (n^{\omega(\log_n k) } + \frac{n^2 \sqrt{\bar \kappa}}{\epsilon^5p} \right ) &= \tilde O \left (n^{\omega(\log_n k)} + \frac{n^{3/2+1/p}}{k^{1/p-1/2}\cdot \epsilon^{5+1/p}\cdot p} \right ) \\ &= \tilde O \left (n^{\omega(\log_n k)} + \frac{n^{3/2+1/p - (\log_n k) (1/p-1/2)}}{\epsilon^{5+1/p}\cdot p} \right ). 
\end{align*}
We now balance the terms, again not optimizing $\epsilon$ dependence. Writing $\gamma = \log_n k$, $\omega(\gamma) = 2$ for $\gamma < {\alpha}$ where $\alpha > 0.31389$ and $2 + (\omega-2)\frac{\gamma - \alpha}{1-\alpha}$ for $\gamma \ge {\alpha}$ \cite{gall2017improved}.
Assuming $\gamma > \alpha$ we set:
\begin{align*}
2 + (\omega-2) \frac{\gamma-\alpha}{1-\alpha} &= \frac{3}{2} + \frac{1}{p} - \frac{\gamma}{p} + \frac{\gamma}{2}%\\
\end{align*}
which gives $\gamma \approx \frac{1/p -.3294}{1/p +.0435} > \alpha$ for all $p < 2$ (so our assumption that $\gamma \ge \alpha$ was valid.) This yields total runtime $\tilde O \left (n^{\frac{2.3729 - .0994p}{1 + .0435p}}\right )$.
Without using fast matrix multiplication, the first term in the runtime becomes $n^2k$ and so we balance costs by setting:  $n^{2+\gamma} = n^{3/2 + 1/p - \gamma/p + \gamma/2}$ which gives $\gamma = \frac{1/p - 1/2}{1/p+1/2}$ and total runtime $\tilde O \left (n^{\frac{3+p/2}{1+p/2}} \right )$.

\end{proof}

\subsubsection*{Bounded Orlicz Norms}
An Orlicz norm of a vector $\bv{x}\in \mathbb{R}^n$ is given by $G(\bv{x}):=\sum_{i=1}^n g(|\bv{x}_{i}|)$, where $g(\cdot)$ is convex, nonnegative, and has $g(0)=0$. It is easy to see that applied to the vector of singular values, an Orlicz matrix norm is a special case of a spectral sum and can be approximated with Theorem \ref{thm:histogram_algorithm} under sufficient conditions. A simple example shows that Orlicz norms for $g(\cdot)$ bounded by an envelope of $x^{p_1}$ and $x^{p_2}$ can be approximated in a similar time as $\norm{\bv{A}}_{p_2}^{p_2}$ in Corollary \ref{cor:hist_schatten}.

\begin{corollary}[Bounded Orlicz norms via Histogram Approximation]\label{cor:hist_orlicz} For any convex function $g: \mathbb{R}^+ \rightarrow \mathbb{R}^+$, let $p_1$
 be the minimal real for which for all arguments $a, b$ to $g$ with $a \le b$, $\left (\frac{b}{a}\right )^{p_1} \ge \frac{g(b)}{g(a)}$. Let $p_2$ be the maximum positive real for which for all $a\le b$ we have  $\left(\frac{b}{a}\right)^{p_2} \le \frac{g(b)}{g(a)}$.
 
 Then for any $\bv{A} \in \mathbb{R}^{n \times n}$ with uniformly sparse rows (i.e. $d_s(\bv{A}) = O(\nnz(\bv{A})/n)$), given error parameter $\epsilon \in (0,1)$ and $M \in [ \norm{\bv{A}}_2, 2\norm{\bv{A}}_2]$, if we run Algorithm \ref{algo:histogram} on $\frac{1}{M} \bv{A}$ with $\epsilon_1,\epsilon_2 = c\epsilon$, $\alpha = c\epsilon/ \max\{1,p_1\}$ and $\lambda = \frac{1}{M^2} \left (\frac{c\epsilon/2}{n} \norm{\bv{A}}^{p_2}_{p_2} \right )^{2/{p_2}}$ for sufficiently small constant $c$ then with probability $99/100$, letting $a_1, \tilde b_0,...,\tilde b_{T}$ be the outputs of the algorithm we have:
\begin{align*}
(1-\epsilon) \Sum_g(\bv{A}) \le \sum_{t=0}^{T} g \left (M a_1^{1/2} (1-\alpha)^{t/2}\right ) \cdot \tilde b_t \le (1+\epsilon) \Sum_g(\bv{A}).
\end{align*}
Further the algorithm runs in time:
\small{
\begin{align*}
\tilde O\left (\frac{\nnz(\bv{A})(p_1+1)^2}{\epsilon^{5+1/{p_2}}}\right) \text{ for } p_2\ge 2 \text{\hspace{1em} and \hspace{1em}}
\tilde O \left ( \frac{(p_1+1)^2\cdot\nnz(\bv{A}) n^{\frac{1/{p_2}-1/2}{1/{p_2}+1/2}}+n^{\frac{5/{p_2}-1/2}{2/{p_2}+1}}\sqrt{d_s(\bv{A})}}{p_2\cdot \epsilon^{5+1/p_2}} \right )
\text{ for }p_2 \le 2.
\end{align*}}
\normalsize
For dense inputs this can be sped up to $\tilde O \left ((p_1+1)^2 \cdot \frac{n^{\frac{2.3729 - .0994p_2}{1 + .0435p_2}}}{p_2\cdot \epsilon^{5+1/p_2}} \right )$ via fast matrix multiplication.
 \end{corollary}
 Note that $\lambda$ depends on $\norm{\bv{A}}_{p_2}^{p_2}$ which can be estimated for example using Corollary \ref{cor:hist_schatten}, or via binary search as described in the proof of Corollary \ref{cor:hist_schatten}.
 \begin{proof}
 We invoke Theorem \ref{thm:histogram_algorithm} with $f(x) = g(x)$. While we do not show that $g(\cdot)$ exactly satisfies the multiplicative smoothness condition, the upper bound on $g(\cdot)$ directly implies the result of Claim \ref{claim:hypo} with $\delta_f = p_1$. So our setting of $\alpha = c\epsilon/\max \{1,p_1\}$ suffices as the bound in this Claim is the only smoothness condition used to prove Theorem \ref{thm:histogram_algorithm}.
 
Additionally, by the lower bound on $g(\cdot)$, for any $c$, we can set $
\lambda_f(c\epsilon) = \left ( \frac{c \epsilon/2}{n} \norm{\bv{A}}_{p_2}^{p_2} \right )^{1/{p_2}}$. For any $\sigma_i(\bv{A}) \ge \lambda_f(c\epsilon)$:
\begin{align*}
g(\sigma_i(\bv{A})) &\ge \left (\frac{\sigma_i(\bv{A})}{\lambda_f(c\epsilon)}\right)^{p_2} g(\lambda_f(c\epsilon))\\
&\ge \frac{2n}{c\epsilon} \cdot \frac{\sigma^{p_2}_i(\bv{A})}{\norm{\bv{A}}_{p_2}^{p_2}} \cdot g(\lambda_f(c\epsilon)).
\end{align*} which gives:
\begin{align*}
\sum_{\sigma_i(\bv{A}) \ge \lambda_f(c\epsilon)} g(\sigma_i(\bv{A})) &\ge \frac{2n\cdot g(\lambda_f(c\epsilon))}{c\epsilon\cdot \norm{\bv{A}}_{p_2}^{p_2}} \cdot \sum_{\sigma_i(\bv{A}) \ge \lambda_f(c\epsilon)} \sigma_i^{p_2}(\bv{A})\\
&\ge \frac{2n\cdot g(\lambda_f(c\epsilon))}{c\epsilon}\cdot (1-c\epsilon/2)
\end{align*}
where the last bound follows from the fact that by the setting of $\lambda_f(c\epsilon)$, $\frac{\sum_{\sigma_i(\bv{A}) \ge \lambda_f(c\epsilon)} \sigma_i^{p_2}(\bv{A})}{\norm{\bv{A}}_{p_2}^{p_2}} \ge 1-c\epsilon/2.$ 
Finally for any $\sigma_i(\bv{A}) \le \lambda_f(c\epsilon)$, $g(\sigma_i(\bv{A})) \le g(\lambda_f(c\epsilon))$. So overall we have:
\begin{align*}
\sum_{\sigma_i(\bv{A}) < \lambda_f(c\epsilon)} g(\sigma_i(\bv{A})) \le n \cdot g(\lambda_f(c\epsilon)) \le \frac{c\epsilon/2}{1-c\epsilon/2} \cdot \sum_{\sigma_i(\bv{A}) \ge \lambda_f(c\epsilon)} g(\sigma_i(\bv{A})) \le \frac{c\epsilon/2}{1-c\epsilon} S_g(\bv{A}) \le c\epsilon S_g(\bv{A}) .
\end{align*}
Thus our setting of $\lambda$ suffices and the accuracy bound follows from Theorem \ref{thm:histogram_algorithm}.

It remains to discuss runtime. We have $\log(1/\lambda) = \tilde O\left (\max \{1,1/{p_2} \} \right)$ and $\delta_f^2 = \max \{1,p_1^2 \} \le (p_1+1)^2$. The runtimes follow from Theorem \ref{thm:histogram_algorithm} via the same arguments used in Corollary \ref{cor:hist_schatten}.
 \end{proof}

\subsubsection*{Ky Fan Norms}

The Ky Fan $w$-norm of a matrix is the sum of its top $w$ singular values: $\norm{\bv{A}}_{KF(w)} \eqdef \sum_{i=1}^w \sigma_i(\bv{A})$ (note that these norms are typically called the `Ky Fan $k$-norms', however we use $w$ to avoid overloading notation on $k$). Such a norm is not strictly speaking a spectral sum. However it can still be approximated using our histogram method.
We have the following corollary of Theorem \ref{thm:histogram}:

\begin{corollary}[Ky Fan norms via Histogram Approximation]\label{cor:kyfan}
For any $\bv{A} \in \mathbb{R}^{n \times n}$, given rank $w$,  error parameter $\epsilon \in (0,1)$ and $M \in [ \norm{\bv{A}}_2, 2\norm{\bv{A}}_2]$, if we run Algorithm \ref{algo:histogram} on $\frac{1}{M} \bv{A}$ with $\epsilon_1 = c\epsilon$, $\epsilon_2 = \frac{c_2\epsilon^2}{\log(1/\lambda) }$, $\alpha = c\epsilon$ and $\lambda = \left (\frac{c\epsilon}{Mw} \norm{\bv{A}}_{KF(w)} \right )^{2}$ for sufficiently small constants $c,c_2$ then with probability $99/100$, letting $a_1, \tilde b_0,...,\tilde b_{T}$ be the outputs of the algorithm we have:
\begin{align*}
(1-\epsilon)\norm{\bv{A}}_{KF(w)} \le \sum_{i=1}^w M a_1^{1/2}(1-\alpha)^{ \tilde t(i)/2} \le (1+\epsilon)\norm{\bv{A}}_{KF(w)}
\end{align*}
where $\tilde t(i)$ is the smallest integer with $\sum_{t=0}^{\tilde t(i)} (1+2c\epsilon)\tilde b_t  \ge i$ or $\tilde t(i) = \infty$ if no such integer exists (and hence the $i^{th}$ term of the sum is just $0$). Further the algorithm runs in time:
\begin{align*}
\tilde O\left(\frac{\nnz(\bv{A})\sqrt{w} + nw}{\epsilon^7} \right).
\end{align*}
\end{corollary}
Note that to compute the top $w$ singular values explicitly would require $\tilde O \left (\nnz(\bv{A})w + nw^{\omega-1} \right )$ time using for example a block Krylov method \cite{musco2015randomized}. 
Also note that $\lambda$ depends on the norm we are attempting to compute. As discussed in Corollary \ref{cor:hist_schatten}, we can approximate $\lambda$ via binary search, successively decreasing it until our call to Algorithm \ref{algo:histogram} exceeds the stated runtime.
\begin{proof}
We can assume without loss of generality that $\norm{\bv{A}}_2 \le 1$ and $M = 1$ since rescaling will not effect our approximation factor.
By Theorem \ref{thm:histogram}, since we set $\epsilon_1 = c\epsilon$, $\epsilon_2 =  \frac{c_2\epsilon^2}{\log(1/\lambda) }$, and have $T = \lceil \log_{1-\alpha} \lambda \rceil = \Theta \left (\frac{\log(1/\lambda)}{\epsilon} \right )$:
\begin{align}\label{kyfanbBound}
(1-c\epsilon) b_t \le \tilde b_t \le b_t + 2c\epsilon (b_{t-1} + b_t + b_{t+1})
\end{align}
if we set $c_2$ small enough compared to $c$. This give $b_t \le \frac{1}{1-c\epsilon} \tilde b_t \le (1+2c\epsilon) \tilde b_t$ if $c$ is set small enough. Thus, since these scaled bucket sizes strictly overestimate the true bucket sizes we have $\tilde t(i) \le t(i)$, where $ t(i)$ is the smallest integer with $\sum_{t=0}^{ t(i)} b_t  \ge i$. This gives, since $\alpha = c\epsilon$ and since by our setting of $\lambda$, at most an $c\epsilon$ fraction of $\norm{\bv{A}}_{KF(w)}$ falls outside of the ranges $R_0,...,R_{T}$:
\begin{align*}
\norm{\bv{A}}_{KF(w)} = \sum_{i=1}^w \sigma_i(\bv{A}) \le (1 + 4c\epsilon) \sum_{i=1}^w M a_1^{1/2}(1-\alpha)^{ t(i)/2} \le (1+4c\epsilon) \sum_{i=1}^w M a_1^{1/2}(1-\alpha)^{ \tilde t(i)/2}.
\end{align*}
On the other side, let $\bv{\tilde v}$ be a vector that, for each $t \in \{0,...,T \}$ has $(1+4c\epsilon) \tilde b_t$ entries each set to $Ma_1^{1/2}(1-\alpha)^{t/2}$. Define the top-$w$ norm of $\bv{\tilde v}$ to be the sum of its largest $w$ entries, denoted by $\norm{\bv{\tilde v}}_{T(w)}$. Our estimate of $\norm{\bv{A}}_{KF(w)}$ is equal to $\norm{\bv{\tilde v}}_{T(w)}$. 
 Note that we can also add arbitrary zeros as padding entries to $\bv{\tilde v}$ and not change this norm.
 
Similarly, let $\bv{v'}$ be a vector with $(1+10c\epsilon) b_t$ values each set to $(1+2c\epsilon) \cdot a_1^{1/2}(1-\alpha)^{t/2}$ for $t \in \{0,...,T \}$ and $2\epsilon c \cdot b_{T+1}$ entries set to $(1+2c\epsilon)  \cdot a_1^{1/2}(1-\alpha)^{T/2} \le \lambda$. 

By \eqref{kyfanbBound}, $(1+2c\epsilon)\tilde b_t \le (1+6c\epsilon) b_t + 2c\epsilon b_{t-1} + 2c\epsilon b_{t+1}$. This fact combined with the fact that the entries in $\bv{v}'$ are scaled up by a $(1+2c\epsilon)$ factor ensure that the entries of $\bv{v}'$ dominate those of $\bv{\tilde v}$ and so $\norm{\bv{v}'}_{T(w)} \ge \norm{\bv{\tilde v}}_{T(w)}$.
Further, $\norm{\bv{v}'}_{T(w)} \le (1+13c\epsilon) \norm{\bv{A}}_{KF(w)}$ since we have scaled up each entry by at most $(1+2c\epsilon)$ factor and scaled up each bucket size $b_t$ by a $(1+10c\epsilon)$ factor. This bound gives our final multiplicative approximation after adjusting the constant $c$ on $\epsilon$.

It remains to discuss runtime. We invoke Theorem \ref{thm:window}, setting $k = w^{1/2}$. This gives:
$$\hat \kappa = \frac{\sigma_{k+1}^2(\bv{A})}{\lambda} = \frac{\sigma_{k+1}^2(\bv{A}) w^2}{\norm{\bv{A}}_{KF(w)}^2(\epsilon c)^2} \le \frac{w}{(c\epsilon)^2}$$
where we us the fact that for $k = w^{1/2}$, $\sigma_k(\bv{A}) \le \frac{1}{w^{1/2}} \norm{\bv{A}}_{KF(w)}$. With our settings of $\lambda$, $\alpha$ and $\epsilon_1,\epsilon_2$, Algorithm \ref{algo:histogram} performs $\tilde \Theta(1/\epsilon_1^2 \cdot 1/\alpha) = \tilde \Theta(1/\epsilon^3) $ calls to $\algW$, each which requires $\tilde \Theta (1/\gamma) = \tilde \Theta(\epsilon_2 \alpha) = \tilde \Theta(1/\epsilon^3)$  regression calls due to our setting of $\alpha$ and $\epsilon_2$. Plugging $\hat \kappa$ into the second runtime of Theorem \ref{thm:window} gives the corollary.
\end{proof}
% !TEX root = normEstimation.tex
\section{Lower Bounds}\label{sec:lower}

In this section we give hardness results for high accuracy spectrum approximation. Specifically, we show how to detect if an undirected unweighted graph contains a triangle using accurate approximation algorithms  for various important spectral sums such as the Schatten $p$-norms, log-determinant, the SVD entropy, and the trace inverse. Our spectral sum bounds further imply hardness for the important primitives of computing effective resistances in a graph or leverage scores in a matrix.

In the seminal work of \cite{williams2010subcubic} it was shown that any truly subcubic time algorithm for triangle detection yields a truly subcubic time algorithms for Boolean matrix multiplication (BMM). Consequently, these results show that computing any of these quantities too precisely is in a sense as difficult as BMM. Furthermore, as it is a longstanding open question whether or not there is any subcubic time combinatorial algorithm for BMM, i.e. an algorithm which avoids the powerful algebraic manipulations inherent in the fastest algorithms for BMM, these results can be viewed as showing that we do not expect simple iterative methods to yield precise algorithms for these problems without an immense breakthrough in linear algebra.

% cannot be done too efficiently without a major algorithmic breakthrough.

In Section \ref{sec:reductions} we give a general result on reducing spectral sums to triangle detection. Then in Section~\ref{sec:sum_hardness} we use this to show hardness for computing various well studied spectral sums. In Section~\ref{sec:effres_hard} we conclude by showing hardness for computing graph and numerical linear algebra primitives, i.e. effective resistances and leverage scores.

\subsection{Reductions From Triangle Detection}
\label{sec:reductions}

Here we provide our main technical tool for reducing spectral sum computation to triangle detection. As discussed in Section~\ref{sec:lower_bound_approach}, our reduction leverages the well known fact that the number of triangles in any unweighted graph $G$ is equal to $\tr(\bv{A}^3) / 6$ where $\bv{A}$ is the adjacency matrix for $G$. Consequently, given any function $f : \R^+ \rightarrow \R^+$ whose power series is reasonably behaved, we can show that for suitably small $\delta$ the quantity $\tr(f(\bv{I} + \delta \bv{A}))$ is dominated by the contribution of $\tr(\bv{A}^k)$ for $k \in (0, 3)$. Therefore computing $\tr(f(\bv{I} + \delta \bv{A}))$ approximately lets us distinguish between whether or not $\tr(\bv{A}^3) = 0$ or $\tr(\bv{A}^3) \geq 6$. 

We formalize this in the following theorem. As it simplifies the result, we focus on the case where $f$ is defined on the interval $(0,2)$, however, this suffices for our purposes and can be generalized via shifting and scaling of $x$.
 
%\begin{theorem}\label{triangle_lowerBound}
%Consider any $f: \R^+ \rightarrow \R^+$ such that for some $a,b \in R$ and set of coefficients $c_0,c_1,...$ we can write
%$$f(x) = \sum_{k=0}^\infty c_k (x-a)^k$$
% for all $x \in [a-b, a + b]$, where $c_3 \neq 0$ and $c_i \le c_3 \cdot h^{i-3}$ for all $i > 3$. 
%Then given $X \in (1\pm O(1/h)) \sum_{i=1}^n f(\sigma_i)$ 
%\end{theorem}
\begin{theorem}[Many Spectral Sums are as Hard as Triangle Detection]\label{triangle_lowerBound}
Let $f: \R^+ \rightarrow \R^+$ be an arbitrary function, such that for $x \in (0,2)$ we can express it as 
\begin{align}\label{seriesExpansion}
f(x) = \sum_{k=0}^\infty c_k (x-1)^k
\text{ where }
\left |\frac{c_k}{c_3} \right |\le h^{k-3}
\text{ for all } k > 3
\end{align}
Then given the adjacency matrix $\bv{A} \in \R^{n \times n}$ of any simple graph $G$ that has no self-loops and spectral sum estimate $X \in (1\pm \epsilon_1 / 9) \sum_{i=1}^n f(\sigma_i(\bv{I} - \delta \bv{A}))$ for scaling $\delta$ and accuracy $\epsilon$ satisfying
\[
\delta = \min \left \{\frac{1}{n} ~ , ~ \frac{1}{10 n^4 h} \right \}
\text{ and }
\epsilon_1 = \min \left \{1 ~ , ~  \left |\frac{c_3\delta^3 }{c_0 n } \right | ~ , ~ \left |\frac{c_3 \delta}{c_2n^2}\right | \right \}
\]
we can detect if $G$ has a triangle in $O(\nnz(\bv{A}))$ time. 

Consequently, given an algorithm which on input $\bv{B} \in \R^{n \times n}$ outputs $Y \in (1\pm \epsilon) \sum_{i=1}^n f(\sigma_i(\bv{B}))$ in $O(n^\gamma \epsilon^{-c})$ time we can produce an $O(n^2 + n^\gamma \epsilon_1^{-c}))$ time triangle detection algorithm.
\end{theorem}

\begin{proof}
%Let $A \in \R^{n\times n}$ be the adjacency matrix of any unweighted $n$ node graph $G$. $\tr(A^3) = 6 \Delta(G)$ where $\Delta(G)$ is the number of triangles in $G$. So in order to detect if $G$ has a triangle, we must determine of $\tr(A^3) \ge 1/6$ or $\tr(A^3) = 0$.   
Let $\bv{A}$, $G$, $\delta$, $\epsilon_1$, and $X$ be as in the theorem statement and let $\bv{B} \eqdef \bv{I} - \delta \bv{A}$. By Gershgorin's circle theorem, $\norm{\bv{A}}_2 \le n - 1$ and  since $\delta \le 1/n$, $\norm{\delta \bv{A}}_2 < 1$. Consequently $\bv{B}$ is symmetric PSD, $\sigma_i(\bv{B}) = \lambda_i(\bv{B}) \in (0, 2)$ for all $i \in [n]$, and therefore applying \eqref{seriesExpansion} yields: %so $\lambda_i(B) = \sigma_i(B)$. 
\begin{align*}
\sum_{i=1}^n f(\sigma_i(\bv{B})) = \sum_{i=1}^n f(1-\delta\lambda_i(\bv{A})) = \sum_{i=1}^n \sum_{k=0}^\infty c_k (\delta \lambda_i(\bv{A}))^k = \sum_{k=0}^\infty c_k \delta^k \tr(\bv{A}^k) \,.
\end{align*}
$\delta \le \frac{1}{10n^4 h}$ is enough to insure that the first three terms of this power series dominate. Specifically:
\begin{align*}
\left |\sum_{k=4}^\infty c_k \delta^k \tr(\bv{A}^k) \right |
%&
= \left | c_3\delta^3 \sum_{k=4}^\infty \frac{c_k}{c_3}\delta^{k-3} \tr(\bv{A}^k) \right |
%\\
%&
\le |c_3| \delta^3 \sum_{k=4}^{\infty} \frac{1}{10^{k-3}} \left (\frac{1}{n^4}\right )^{k-3} \tr(\bv{A}^k)
%\\
%&
\le  \frac{|c_3| \delta^3}{9} 
\end{align*}
where the last inequality uses the fact that $\tr(A^k) \le \|\bv{A}\|_2^{k - 2} \norm{\bv{A}}_F^2 \le n^k \le n^{4(k-3)}$  
for all $k > 3$. 
%For $k = 4$, $\tr(A^4) \le (\tr(A^2))^2 = \norm{A}_F^4 \le n^4$. For $k \ge 5$, $\tr(A^k) \le n^{k+1} \le n^{4(k-3)}$.
 Further, since $\tr(\bv{A}^0) = n$, $\tr(\bv{A}) = 0$, and $\tr(\bv{A}^2) = \norm{\bv{A}}_F^2 \leq n^2$ we have:
\begin{align*}
0 \leq c_0 \tr(\bv{A}^0) + c_1 \delta \tr(\bv{A}) + c_2 \delta^2(\tr(\bv{A}^2)) \le |c_3| \delta^3 \cdot \left ( \left |\frac{c_0n}{c_3\delta^3}\right | + \left | \frac{c_2n^2}{c_3\delta}  \right | \right ) \le \frac{|c_3| \delta^3}{\epsilon_1} .
\end{align*}
Now, clearly in $O(\nnz(\bv{A}))$ time we can compute $\tr(\bv{A}^2) = \|\bv{A}\|_F^2$ in $O(\nnz(\bv{A}))$ as well as:
\begin{align*}
X - c_0 n - c_2\delta^2 \tr(\bv{A}^2) &= 
c_3\delta^3 \tr(\bv{A}^3) \pm \frac{|c_3| \delta^3}{9}
\pm \frac{\epsilon_1}{9} \left(\frac{|c_3|\delta^3}{9} + c_3 \delta^3 \tr(\bv{A}^3) + \frac{|c_3| \delta^3}{\epsilon_1}\right)\\
&=  c_3 \delta^3 \left[
 \tr(\bv{A}^3) \left(1\pm \frac{1}{20}\right)
 \pm \frac{1}{3}
 \right]
\end{align*}
%\begin{align*}
%X - c_0 n - c_2\delta^2 \tr(\bv{A}^2) &= c_3\delta^3 \left ( \tr(\bv{A}^3) \pm \left (1/9 + \frac{\epsilon_1}{20} (1+1/\epsilon_1 + 1/9) \right) \right )\\
%&=  c_3\delta^3 \left ( \tr(A^3) \pm 1/9 \right )
%\end{align*}
So we can detect if $\tr(\bv{A}^3) = 0$ or if $\tr(\bv{A}^3) \ge 6$ and hence whether or not $G$ has a triangle.
\end{proof}

Note that in this reduction, so long as $\delta$ is small (i.e. $\leq 2n$) then $\bv{B} = \bv{I} - \delta \bv{A}$ is a very well conditioned matrix (its condition number is at most a constant). 
Consequently, our lower bounds apply even when approximately applying for example $\bv{B}^{-1}$ or $\bv{B}^{1/2}$ to high precision is inexpensive. The theorem (and the results in Section~\ref{sec:effres_hard}) suggests that the difficulty in computing spectral sums arises more from the need to measure the contribution from multiple terms precisely, than from the difficulty in manipulating $\bv{B}$ for the purposes of applying it to a single vector. 
%more in the need to ma
%This highlights that the difficulty of computing spectral sums arises more from the need to make precise measurements, then for the difficulty of applying the inverse of applying simple operations to $\bv{B}$. 

Also, note that the matrix $\bv{B}$ in this reduction is symmetrically diagonally dominant (SDD). So, even for these highly structured matrices which admit near linear time application of $\bv{B}^{-1}$ \cite{spielman2004nearly} as well as approximate factorization \cite{kyng2016approximate}, accurate spectral sums are difficult. We leverage this in Section~\ref{sec:effres_hard}. 

\subsection{Hardness for Computing Spectral Sums}
\label{sec:sum_hardness}

Here we use Theorem~\ref{triangle_lowerBound} to show hardness for various spectral sum problems. To simplify our presentation, we focus on the case of dense matrices, showing bounds of the form $\Omega(n^\gamma \epsilon^{-c})$. However, note that Theorem~\ref{triangle_lowerBound} and the approach we use also yields lower bounds on the running time for sparse matrices and can be stated in terms of $\nnz(\bv{A})$. 

\subsubsection*{Schatten $p$-norm for all $p \neq 2$}

For $x \in (0,2)$, using the Taylor Series about $1$ we can write
\begin{align}
x^p = \sum_{k=0}^\infty c_k (1-x)^{k}\text{ where } c_k = \frac{\prod_{i=0}^{k-1} (p-i)}{k!}
\end{align}
This series converges since $|c_k| \le 1$ for $k > p$ and for $x \in (0,2)$, $(1-x) < 1$.
Note that when $p$ is a non-negative integer, only the first $p$ terms of the expansion are nonzero. When $p$ is non-integral, the sum is infinite. % and hence $c_3 \neq 0$ for all $p\neq 1,2$. 
We will apply Theorem \ref{triangle_lowerBound} slightly differently for different values of $p$. We first give our strongest result:
\begin{corollary}[Schatten $3$-Norm Hardness]
Given an algorithm which on input $\bv{B} \in \R^{n\times n}$ returns $X \in (1\pm \epsilon) \norm{\bv{B}}_3^3$ in $O (n^\gamma \epsilon^{-c} )$ we can produce an algorithm that detects if an $n$-node graph contains a triangle in $O(n^{\gamma + 4c})$ time.
\end{corollary}
\begin{proof}
For $p = 3$, $c_k = 0$ for $k > 3$. So we apply Theorem \ref{triangle_lowerBound} with $h=0$ and hence $\delta = 1/n$ and $\epsilon_1 = \frac{c_3 \delta^3}{c_0 n} = \frac{1}{n^4}$.
\end{proof}
Note that for $p$ very close to $3$ a similar bound holds as $h \approx 0$. 
If $p = 3$ Theorem \ref{finalthm:dense} gives an algorithm running in $\tilde O(n^2/\epsilon^3)$ time. Significant improvement to the $\epsilon$ dependence in this algorithm therefore either requires loss in the $n$ dependence or would lead to $o(n^{w})$ time triangle detection for the current value of $\omega$.
We next extend to all $p \neq 1,2$.
\begin{corollary}[Schatten $p$-Norm Hardness, $p \neq 1,2$]\label{onehalfthreehalfs}
For any $p > 0$, $p \neq 1,2$, given algorithm $\algA$ which for any $\bv{B} \in \R^{n\times n}$ returns $X \in (1\pm \epsilon) \norm{\bv{B}}_p^p$ in $O(n^\gamma \epsilon^{-c})$ time, we can detect if an $n$-node graph contains a triangle in $O \left(n^{\gamma + 13c} \cdot \frac{p^{3c}}{\left |\min \{p,(p-1),(p-2) \}\right |^c} \right )$ time.
\end{corollary}
\begin{proof}
We have $\frac{c_k}{c_3} \le p^{k-3}$ for all $k > 3$ as well as
$\left |\frac{c_0}{c_3} \right | = \left |\frac{1}{p(p-1)(p-2)} \right | \le \left |\frac{1}{2\min \{p,(p-1),(p-2) \}} \right |$ and similarly $\left |\frac{c_2}{c_3} \right| \le \left |\frac{1}{2 \min \{p,(p-1) \}}\right |$. %Since $p \notin (1/2,3/2)$, $p,p-1 > 1/2$ and hence 
We apply Theorem \ref{triangle_lowerBound} with $\delta = \Theta \left ( \frac{1}{n^4 p}\right)$ and $\epsilon_1 = \frac{c_3\delta^3}{c_0 n} = \Theta \left (\frac{\left |\min \{p,(p-1),(p-2) \}\right |}{n^13 p^3}\right)$, which gives the result.
\end{proof}
In the typical case when $p << n$, the $p^{3c}$ term above is negligible.  The $\frac{1}{\left |\min \{p,(p-1),(p-2) \}\right |^c}$ term is meaningful however. Our bound becomes weak as $p$ approaches $2$ (and meaningless when $p =2$). This is unsurprising, as for $p$ very close to $2$, $\norm{\bv{B}}_p^p \approx \norm{\bv{B}}_F^2$, which can be computed exactly in $\nnz(\bv{B})$ time. The bound also becomes weak for $p \approx 1$, which is natural as our reduction only uses PSD $B$, for which $\norm{\bv{B}}_1 = \tr(\bv{B})$ which can be computed in $n$ time. However, we can remedy this issue by working with a (non-PSD) square root of $\bv{B}$ which is easy to compute:
\begin{corollary}[Schatten $p$-Norm Hardness, $p \approx 1$]\label{nearone}
For any $p$, given an algorithm which for any $\bv{B} \in \R^{p\times n}$ returns $X \in (1\pm \epsilon) \norm{\bv{B}}_p^p$ in $O \left (f(\nnz(\bv{B}),n) \cdot \frac{1}{\epsilon^c} \right)$ time, we can detect if an $n$-node graph with $m$ edges contains a triangle in $O \left(f(m,n) \cdot n^{13c} \cdot \frac{p^{3c/2}}{\left |\min \{p/2,(p/2-1),(p/2-2) \}\right |^c} + m \right )$ time.
\end{corollary}
Note that for $p \approx 1$, $\frac{p^{3c/2}}{\left |\min \{p/2,(p/2-1),(p/2-2) \}\right |^c}$ is just a constant. Again, the bound is naturally weak when $p \approx 2$ as $(p/2-1)$ goes to $0$
\begin{proof}
For $\bv{B} = I - \delta \bv{A}$ as in Theorem \ref{triangle_lowerBound}. Let $\bv{L} = \bv{D} - \bv{A}$ be the Laplacian of $G$ where $\bv{D}$ is the diagonal degree matrix. We can write $\bv{B} = \delta \bv{L} + \widehat{\bv{D}}$ where $\widehat{\bv{D}} = \bv{I} -\delta \bv{D}$ is PSD since $\delta \le 1/n$. Letting $\bv{M} \in \mathbb{R}^{{n \choose 2} \times n}$ be the vertex edge incidence matrix of $\bv{A}$, and $\bv{N} = [\delta^{1/2}\bv{M}^T, \widehat{\bv{D}}^{1/2}]$, we have $\bv{N} \bv{N}^T = \bv{B}$. Thus, $\norm{\bv{N}}_p = \norm{\bv{B}}_{p/2}^{p/2}$ and so approximating this norm gives triangle detection by Corollary \ref{onehalfthreehalfs}.
Note that $\nnz(\bv{N}) = O(\nnz(\bv{A}))$ and further $\bv{N}$ matrix can be formed in this amount of time, giving our final runtime claim. %\sidford{Is redundancy between this and somehting in effres later ... not worth deduping now - but maybe later.}
\end{proof}
%\todo{Give a specific example for this one comparing to upper bound for $\norm{A}_1$.}
Note that for $p = 1$, since $\bv{N}$ has maximum row sparsity $2$, we obtain a runtime via Theorem \ref{thm:smallp} of $\tilde O(\epsilon^{-3} \left (m n^{1/3} + n^{3/2}\right )) = o(n^{\omega})$ for the current value of $\omega$, even when $m = n^2$, implying that significantly improving this $\epsilon$ dependence would either improve matrix multiplication or come at a cost in the polynomials of the other parameters.

\paragraph{SVD Entropy:}

\begin{corollary}[SVD Entropy Hardness]
Given algorithm $\mathcal{A}$ which for any $\bv{B} \in \R^{n\times n}$ returns $X \in (1\pm \epsilon) \sum_{i=1}^n f(\sigma_i(\bv{B}))$ for $f(x) = x \log x$ in $O(n^\gamma \epsilon^{-c})$ time, we can detect if an $n$-node graph contains a triangle in $O(n^{\gamma + 6c})$ time.
\end{corollary}
\begin{proof}
For $x \in (0,2)$, using the Taylor Series about $1$ we can write $x \log x = \sum_{k=0}^\infty c_k (x-1)^k$
where $c_0 = 1 \log(1) = 0$, $c_1 = \log(1) + 1 = 1$, and $|c_k| = \frac{(k-2)!}{k!} \le 1$ for $k \ge 2$. %The series converges since $1-x < 1$ for $x \in (0,2)$. 
$c_k < c_3$ for all $k > 3$ and $\frac{c_0}{c_3} = 0$ while $\frac{c_2}{c_3} = \frac{1}{3}$. Applying Theorem \ref{triangle_lowerBound} with $\delta = \frac{1}{10n^4}$ and $\epsilon_1 = \frac{\delta}{3n^2} = \frac{1}{30n^6}$ gives the result.
\end{proof}

\paragraph{Log Determinant:}
\begin{corollary}[Log Determinant Hardness]
Given algorithm $\mathcal{A}$ which for any $\bv{B} \in \R^{n\times n}$ returns $X \in (1\pm\epsilon) \log(\det(\bv{B}))$ in $O (n^\gamma \epsilon^{-c})$ time, we can detect if an $n$-node graph contains a triangle in $O(n^{\gamma + 6c})$ time.
\end{corollary}
\begin{proof}
Using the Taylor Series about $1$ we can write
$\log x = \sum_{k=0}^\infty c_k (x-1)^k$
where $c_0 = 0$, $|c_i| = 1/i$ for $i \ge 1$. Therefore $c_k < c_3$ for all $k > 3$, $\frac{c_0}{c_3} = 0$, and $\frac{c_2}{c_3} = \frac{3}{2}$. Applying Theorem \ref{triangle_lowerBound} with $\delta = \frac{1}{10n^4}$ and $\epsilon_1 = \frac{\delta}{2n^2} = \frac{1}{20n^6}$ gives the result.
\end{proof}

In Appendix \ref{sec:hardnessAppendix}, Lemma \ref{detHardness} we show that a similar result holds for computing $\det(\bv{B}) = \prod_{i=1}^n \lambda_i(\bv{B})$.
In \cite{baur1983complexity} it is shown that, given an arithmetic circuit for computing $\det(\bv{B})$, one can generate a circuit of the same size (up to a constant) that computes $\bv{B}^{-1}$. This also yields a circuit for matrix multiplication by a classic reduction.\footnote{Matrix multiplication reduces to inversion by the fact that
\begin{tiny}
$\begin{bmatrix}
    \bv{I} & \bv{A} & \bv{0}\\
    \bv{0} & \bv{I} & \bv{B} \\
   \bv{0} & \bv{0} & \bv{I}
\end{bmatrix}^{-1} = \begin{bmatrix}
    \bv{I} & -\bv{A} & \bv{A}\bv{B} \\
    \bv{0} & \bv{I} & -\bv{B} \\
   \bv{0} & \bv{0} & \bv{I}
\end{bmatrix}$.
\end{tiny}See \cite{invertToMult}. }
Our results, combined with the reduction of \cite{williams2010subcubic} of Boolean matrix multiplication to triangle detection, show that a sub-cubic time algorithm for the approximating $\log(\det(\bv{B}))$ or $\det(\bv{B})$ up to sufficient accuracy, yields a sub-cubic time matrix multiplication algorithm, providing a reduction based connection between determinant and matrix multiplication analogous to the circuit based result of \cite{baur1983complexity}.

\paragraph{Trace of Exponential:}
\begin{corollary}[Trace of Exponential Hardness]
Given algorithm $\mathcal{A}$ which for any $\bv{B} \in \R^{n\times n}$ returns $X \in (1\pm\epsilon) \tr(\exp(\bv{B}))$ in $O (n^\gamma 
\epsilon^{-c})$ time, we can detect if an $n$-node graph contains a triangle in $O(n^{\gamma + 13c})$ time.
\end{corollary}
\begin{proof}
Using the Taylor Series about $1$ we can write $e^x = \sum_{k=0}^\infty \frac{e(x-1)^k}{k!}$. We have $\frac{c_0}{c_3} = 6$, $\frac{c_2}{c_3} = 3$, and for all $k \ge 3$, $c_k < c_3$. Applying Theorem \ref{triangle_lowerBound} with $\delta = \frac{1}{10n^4}$ and $\epsilon_1 = \frac{c_3 \delta^3}{c_0 n} = \frac{1}{6000 n^{13}}$ gives the result.
\end{proof}

\subsection{Leverage Score and Effective Resistance Hardness}
\label{sec:effres_hard}
%\paragraph{Trace of Inverse and Computing the Effective Resistances of an Unweighted Graph:}

Here we show hardness for precisely computing all effective resistances and leverage scores of a graph. Our main tool is the following result (which in turn is an easy corollary of Theorem~\ref{triangle_lowerBound}) for an algorithm that precisely computes the trace inverse of a strictly symmetric diagonally dominant (SDD) $\bv{B}$, i.e. $\bv{B} = \bv{B}^\top$ and $\bv{B}_{ii} > \sum_{j \neq i} \bv{B}_{ij}$. 

\begin{corollary}[Trace of Inverse Hardness]
\label{cor:trace_inv_hard}
Given an algorithm which for any strictly SDD  $B \in \R^{n\times n}$ with non-positive off-diagonal entries returns $X \in (1\pm\epsilon) \tr(B^{-1})$ in $O(n^\gamma \epsilon^{-c})$ time for $\gamma \geq 2$, we can produce an algorithm which detects if an $n$-node graph contains a triangle in $O(n^{\gamma + 13c})$ time.
\end{corollary}
\begin{proof}
For $x \in (0,2)$ we can write $\frac{1}{x} = \sum_{k=0}^\infty (1-x)^k$, and then apply Theorem~\ref{triangle_lowerBound} with $\delta = \frac{1}{10n^4}$ and $\epsilon_1 = \frac{\delta^3}{n} = \frac{1}{1000n^{13}}$. Checking that the $\bv{B}$ in Theorem~\ref{triangle_lowerBound} is strictly SDD with non-positive off-diagonal entries yields the result.
\end{proof}

Using this we prove hardness for precisely computing effective resistances in a graph. Recall that for a weighted undirected graph $G = (V, E, w)$ its Laplacian, $\bv{L} \in \R^{V \times V}$ is given by $\bv{L}_{ij} = -w_{ij}$ if there is an edge between $i$ and $j$ and $0$ otherwise and $\bv{L}_{ii} = - \sum_{i \neq j} \bv{L}_{ij}$ or equivalently $\bv{L} = \bv{D} - \bv{A}$ where $\bv{D}$ is the diagonal degree matrix and $\bv{A}$ is the weighted adjacency matrix associated with $G$. Note that this describes a clear bijection between a Laplacian and its associated graph and we therefore use them fairly interchangeably in the remainder of this section.

The effective resistance between vertices $i$ and $j$ is given by $(\indicVec_i -\indicVec_j)^\top \bv{L}^\dagger (\indicVec_i - \indicVec_j)$ where $\dagger$ denotes the Moore-Penrose pseudoinverse. In the following lemma we prove that compute all the effective resistances between a vertex and its neighbors in the graph can be used to compute the trace of the inverse of any strictly SDD matrix with non-positive off-diagonals and therefore doing this precisely is as hard as triangle detection via Corollary~\ref{cor:trace_inv_hard}. Our proof is based off a fairly standard reduction between solving strictly SDD matrices with negative off-diagonals and solving Laplacian systems.

\begin{lemma}[Effective Resistance Yields Trace]
\label{lem:effres_to_trace_inv}
Suppose we have an algorithm which given Laplacian $\bv{L} \in \R^{n \times n}$ with $m$-non-zero entries, entry $i \in [n]$, and error $\epsilon \in (0, 1)$ computes a $1 \pm \epsilon$ approximation to the total effective resistance between $i$ and the neighbors of $i$ in the graph associated with $\bv{L}$, that is
\[
X \in (1 \pm \epsilon) \sum_{j \in [n] : L_{ij} \neq 0} 
(\indicVec_i - \indicVec_j)^\top 
\bv{L}^\dagger (\indicVec_i - \indicVec_j)
\]
in time $O(m^\gamma \epsilon^{-c})$. Then there is an algorithm that computes the trace of the inverse of $n \times n$ strict SDD matrix with $m$ non-zero entries and non-positive off-diagonals in $O(m^\gamma \epsilon^{-c})$ time.
\end{lemma}

\begin{proof}
Let $\bv{M} \in \R^{n \times n}$ be an arbitrary strictly SDD matrix with  non-positive off-diagonals, i.e. $\bv{M} = \bv{M}^\top$, $\bv{M}_{ii} > \sum_{j \neq i} |\bv{M}_{ij}|$, and $\bv{M}_{ij} \leq 0$ for all $i \neq j$. Let $\bv{v}  \eqdef \bv{M} 1$, 
$\alpha \eqdef \onesVec^\top \bv{M} \onesVec$, and 
\[
\bv{L} \eqdef \left(\begin{array}{cc}
\bv{M} & -v\\
-v^{\top} & \alpha
\end{array}\right)\,.
\]
Now, clearly by our assumptions on $\bv{M}$ we have that $v > \zeroVec$ entrywise and therefore $\alpha > 0$.  Therefore, the
off-diagonal entries of $\bv{L}$ are non-positive and by construction $\bv{L} \onesVec = \zeroVec$. Consequently, $\bv{L}$ is a $(n + 1) \times (n + 1)$ symmetric Laplacian matrix with $\nnz(\bv{M}) + 2n + 1$ non-zero entries.

Now, consider any $x \in \R^{n}$ and $y\in\R$ that satisfy the following
for some $i\in[n]$
\[
\left(\begin{array}{cc}
\bv{M} & -\bv{v} \\
-\bv{v}^{\top} & \alpha
\end{array}\right)\left(\begin{array}{c}
\bv{x}\\
y
\end{array}\right)=\left(\begin{array}{c}
\indicVec_{i}\\
-1
\end{array}\right) ~.
\]
Since $\bv{L}$ is a symmetric Laplacian and the associated graph is connected by construction we know that $\ker(\bv{L}) = \mathrm{span}(\{\onesVec\})$  and there there must exist such $x$ and $y$. Furthermore, since $\bv{M}$ is strictly SDD it is invertible and since
$\bv{M} \onesVec = v$ we have that 
\[
\bv{x} = \bv{M}^{-1}\left(y \cdot \bv{v} + \indicVec_i \right)
=
y \cdot \onesVec + \bv{M}^{-1} \indicVec_i 
\]
and consequently
\[
(\indicVec_{i} - \indicVec_{n+1})^{\top} \bv{L}^{\dagger} (\indicVec_{i} - \indicVec_{n+1})
=
\indicVec_{i}^{\top} \bv{x} - y 
= \indicVec_i^{\top} \bv{M}^{-1} \indicVec_{i}\,.
\] 
Consequently, if we used the algorithm to get a multiplicative approximation $X$ as stated then $X \in (1 \pm \epsilon) \tr(\bv{M}^{-1})$ and the result follows.
\end{proof}

Using this, we also show that computing leverage scores of matrix, a powerful and prevalent notion in convex analysis and numerical linear algebra, is also difficult to compute. This follows from the well known fact that effective resistances in graphs and leverage scores of matrices are the same up to scaling by known quantities.

\begin{corollary}[Leverage Score Hardness]
Suppose we have an algorithm which given $\bv{A} \in \R^{n \times d}$ can compute $\widetilde{\sigma}$ that is a $1 \pm \epsilon$ multiplicative approximation to the leverage scores of $\bv{A}$, i.e. 
\[
\widetilde{\sigma}_i \in (1 \pm \epsilon) \indicVec_i^\top \bv{A} (\bv{A}^\top \bv{A})^\dagger \bv{A}^\top \indicVec_i
\text{ for all } 
i \in [n] 
\]
in time $O(\nnz(\bv{A})^\gamma \epsilon^{-c})$. Then there is a $O(n^{2\gamma + 13c})$ time algorithm for detecting if an $n$-node graph contains a triangle.
\end{corollary}

\begin{proof}
Let $\bv{L} \in \R^{n \times n}$ be a symmetric Laplacian. Let $E = \{\{i, j\} \subseteq [n] \times [n] : \bv{L}_{ij} \neq 0\}$, i.e. the set of edges in the graph associated with $L$. Let $m = |E|$ and $\bv{B} \in \R^{m \times n}$ be the incidence matrix associated with $\bv{L}$, i.e. for all $e = \{i, j\} \in E$ we have $\bv{B}_{e,i} = \sqrt{-\bv{L}_{ij}}$ and $\bv{B}_{e,j} = - \sqrt{-\bv{L}_{ij}}$ for some canonical choice of ordering of $i$ and $j$ and let all other entries of $\bv{B} = 0$. Clearly $\nnz(\bv{B}) = \nnz(\bv{L})$ and we can form $\bv{B}$ in $O(\nnz(\bv{L}))$ time. 

It is well known and easy to check that $\bv{L} = \bv{B}^\top \bv{b}$. Consequently, for all $e = \{i,j\} \in E$ we have
\[
\indicVec_e^\top \bv{B} (\bv{B}^\top \bv{B})^\dagger \bv{B}^\top \indicVec_e = (- \bv{L}_{ij}) \cdot (\indicVec_{i} - \indicVec_{j})^{\top} \bv{L}^{\dagger} (\indicVec_{i} - \indicVec_{j})
\]
Now if we compute $\widetilde{\sigma}$ using the assumed algorithm in $O(\nnz(\bv{L})^\gamma \epsilon^{-c}) = O(n^{2\gamma} \epsilon^{-c})$ time, then since $- \bv{L}_{ij}$ is non-negative in an additional $O(\nnz(\bv{L}) = O(n^2)$ time this yields a $1 \pm \epsilon$ multiplicative approximation to the total effective resistance between any $i$ and all its neighbors in the graph associated with $\bv{L}$. Consequently, by Lemma~\ref{lem:effres_to_trace_inv} and Corollary~\ref{cor:trace_inv_hard} the result follows.
\end{proof}

\section{Improved Error Dependence via Polynomial Approximation}\label{sec:poly}

For constant $\epsilon$, Theorem \ref{thm:histogram_algorithm} and the resulting Corollary \ref{cor:hist_schatten} matches our fastest runtimes for the Schatten $p$-norms. However, it is possible to significantly improve the $\epsilon$ dependence by generalizing our approach. Instead of splitting our spectrum into many small spectral windows, we split into windows of constant multiplicative width and approximate $x^p$ via a low degree polynomial over each window. The degree of this polynomial only varies with $\log(1/\epsilon)$.

In Theorem \ref{thm:histogram_algorithm}, fixing $\delta_f$ to be constant for illustration, each window has width $\alpha = \Theta(\epsilon)$ and so there are $T = \lceil \log_{1-\alpha} \lambda \rceil = \Theta(\log \lambda /\epsilon)$ windows. Additionally, we must set the steepness parameter to be $\gamma = \Theta(\epsilon \alpha) = \Theta(\epsilon^2)$. This loses us a total $1/\epsilon^2$ factor in our runtime as compared to our improved algorithm which will set $\alpha = \Theta(1)$ and so $\gamma = \Theta(\epsilon)$.
%The  algorithm is presented in Algorithm \ref{algo:poly}

%\begin{algorithm}
%\caption{Schatten-$p$ Norm Estimation via Polynomial Approximation}%[1]
%{\bf Input:} $\bv{A}\in\R^{n\times d}$ with $\norm{\bv{A}}_2 \le 1$, accuracy parameter $\epsilon \in (0,1)$, $p > 0$.
%\\
%{\bf Output:} $X \in (1\pm \epsilon) \norm{\bv{A}}_p^p$.
%\begin{algorithmic}
%  \State Set $\gamma = c_1 \epsilon$, $\lambda = c_2 \left (\frac{\epsilon^{2/p}}{d^{2/p}}\right ) \norm{\bv{A}}_p^2$, $T = \lceil \log_{2} 1/\lambda \rceil$, and $S = \frac{\log n}{c_3\epsilon^2}$
%  
%  \State Set $a_0 = 1$ and choose $a_1$ uniformly at random in $[1/2,1]$. Set $a_t =  \frac{a_1}{2^{t-1}}$ for $2 \le t \le T$.
%  \For {$t = 0 :T-1$} \Comment{\textcolor{blue}{Iterate over spectral windows}}
%  	\State Set $\tilde v_t = 0$ \Comment{\textcolor{blue}{Initialize sum estimate.}}
%  	\For{$s = 1:S$} \Comment{\textcolor{blue}{Estimate sum via trace estimation.}}
%		\State Choose $\bv{y} \in \{-1,1\}^d$ uniformly at random.
%  		\State Set $\tilde b_t = \tilde b_t + \frac{1}{S} \cdot \bv{y}^T\algW(\bv{A}^T\bv{A},\bv{y},a_{t+1},a_{t},\gamma,c_3\epsilon_1^2/n)$  \Comment{\textcolor{blue}{Apply soft window via Thm \ref{thm:window}.}}
%		\State If $\tilde b_t \le 1/2$ set $\tilde b_t = 0$.\Comment{\textcolor{blue}{Round small estimates to ensure relative error.}}
%	\EndFor
%  \EndFor\\
% \Return $a_1$ and $\tilde b_t$ for $t=0:T-1$. \Comment{\textcolor{blue}{Output histogram representation.}}
%\end{algorithmic}
%\label{algo:poly}
%\end{algorithm}

We begin by showing that if we well approximate $f$ on each window, then we well approximate $\Sum_f(\bv{A})$ overall. In the following theorems for simplicity we work with PSD matrices, as we will eventually apply these lemmas to $\bv{A}^T\bv{A}$.% in the proof of Theorem \ref{thm:histogram_algorithm}.

\begin{lemma}\label{lem:infinite_window_approx}
	Consider $f: \R \rightarrow \R^+$, parameters $\alpha,\epsilon \in (0,1)$, and gap parameter $\gamma \in (0, \alpha)$. Set $a_0 = 1$, pick $a_1$ uniformly at random from $[1-\alpha,1]$, and set $a_t = a_1(1-\alpha)^{t-1}$ for $t\ge2$. Let $R_t = [a_{t+1},a_t]$ for all $t \geq 0$. Let $h^{\gamma}_{R_t}$ be a $\gamma$-soft window for $R_t$ (Definition \ref{def:window}) and let $\bar R_t =  [(1-\gamma)a_{t+1},(1+ \gamma)a_t]$ be the interval on which $h^\gamma_{R_t}$ is nonzero. Furthermore, for $t \geq 0$, let $q_t$ be a `well-behaved' approximation to $f$ on $\bar R_t$ in the following sense:
	\begin{itemize}
		\item Multiplicative Approximation: $\abs{q_t(x)-f(x)} \leq \epsilon f(x)$ for $x \in \bar R_t$.
		\item Approximate Monotonicity: $q_t(y) \le c_1q_t(x)$ for all $y \leq x$ for some $c_1 \ge 1$.
		\item Range Preserving: $q_t(0) = 0$.
	\end{itemize}
	Then, for any PSD $\bv{A} \in \R^{d \times d}$ with $\norm{\bv{A}}_2 \le 1$, with probability $9/10$:
	\begin{align*}
	\left (1- \epsilon - \frac{40c_1\gamma}{\alpha}\right) \Sum_f(\bv{A}) \le \sum_{t=0}^\infty \Sum_{q_t}(h^{\gamma}_{R_t}(\bv{A})\bv{A}) \le \left (1+ \epsilon + \frac{40c_1\gamma}{\alpha}\right) \Sum_f(\bv{A}).
	\end{align*}
\end{lemma}
\begin{proof}
	Note that the restriction that $q_t(0) = 0$ along with the approximate monotonicity property implies that $q(x)$ is nonnegative, so $\Sum_{q_t}$ is a valid spectral sum and $\Sum_{q_t}(h^{\gamma}_{R_t}(\bv{A})\bv{A}) = \tr(q_t(h^{\gamma}_{R_t}(\bv{A})\bv{A}))$, so we will be able to estimate this sum via stochastic trace estimation. 

Let $t_i$ denote the unique index such that $\sigma_i(\bv{A}) \in R_{t_i}$. Since $\gamma < \alpha$, each $\sigma_i(\bv{A})$ lies in the support of at most two overlapping soft windows (in at most two ranges $\bar R_{t_i}$ and $\bar R_{t_i\pm1}$). Let $T$ be the set of indices whose singular values fall in the support of exactly one soft window and $\bar T$ be its complement. We first bound the error introduced on singular values with indices in $T$.
	\begin{align}
	\sum_{t=0}^\infty  \Sum_{q_t}(h^{\gamma}_{R_t}(\bv{A})\bv{A}) &= \sum_{i=1}^d \sum_{t=0}^\infty q_t\left ( \sigma_i (\bv{A}) \cdot h^{\gamma}_{R_t}(\sigma_i(\bv{A})) \right) \nonumber\\
	&= \sum_{i \in T} q_{t_i}(\sigma_i(\bv{A})) + \sum_{i \in \bar T} \sum_{t=0}^\infty q_t \left ( \sigma_i (\bv{A}) \cdot h^{\gamma}_{R_t}(\sigma_i(\bv{A})) \right)\nonumber\\
	&\in (1\pm\epsilon) \sum_{i \in T} f(\sigma_i(\bv{A}))  + \sum_{i \in \bar T} \sum_{t=0}^\infty q_t \left ( \sigma_i (\bv{A}) \cdot h^{\gamma}_{R_t}(\sigma_i(\bv{A})) \right).\label{eqn:qt-main}
	\end{align}
	The last inequality follows from the multiplicative approximation assumption on $q_t$ that $\abs{q_t(x)-f(x)} \leq \epsilon f(x)$ for $x \in \bar R_t$. Let us now consider a particular $i \in \bar T$ and bound $\sum_{t=0}^\infty q_t \left ( \sigma_i (\bv{A}) \cdot h^{\gamma}_{R_t}(\sigma_i(\bv{A})) \right)$. Note that there are precisely two non-zero terms in this summation -- one corresponding to $t_i$ and the other, to $t_i \pm 1$, which we denote by $t_{\bar i}$. Using the above hypothesis on $q_t$ again, we see that
	\begin{align}
		q_{t_i} \left ( \sigma_i (\bv{A}) \cdot h^{\gamma}_{R_{t_i}}(\sigma_i(\bv{A})) \right) = q_{t_i} \left ( \sigma_i (\bv{A}) \right) \in (1\pm \epsilon) f(\sigma_i (\bv{A})).\label{eqn:qt-1}
	\end{align}
	For the term involving $t_{\bar i}$, we have $\sigma_i (\bv{A}) \in \bar R_{\bar t_i}$ and by the approximate monotonicity requirement that $y \le x \Rightarrow q_t(y) \le c_1q_t(x)$ have:
	\begin{align}
	q_{t_{\bar i}} \left ( \sigma_i (\bv{A}) \cdot h^{\gamma}_{R_{t_{\bar i}}}(\sigma_i(\bv{A}))\right )  \le c_1 q_{t_{\bar i}} ( \sigma_i (\bv{A})) \in c_1 (1\pm \epsilon) f(\sigma_i(\bv{A})).\label{eqn:qt-2}
	%
		%q_{t_{\bar i}} \left ( \sigma_i (\bv{A}) \cdot h^{\gamma}_{R_{t_{\bar i}}}(\sigma_i(\bv{A})) \right) &\leq \sup_{x \in [0,(1+\gamma)a_{t_{\bar i}}]} q_{t_{\bar i}}(x) \le c_1 \inf_{z \in \bar R_t} f(z) \le f(\sigma_i(\bv{A}))
	\end{align}
	
	Due to the random positioning of the windows, in expectation, the total contribution of the singular values lying at the intersection of two windows is small. Specifically, 
	\begin{align}\label{expectedError}
	\E_{\alpha_1} \sum_{i \in \bar T} f(\sigma_i(\bv{A})) = \sum_{i=1}^d \Pr[i\in \bar T] \cdot  f(\sigma_i(\bv{A})) \le \frac{2\gamma}{\alpha} \sum_{i=1}^d f(\sigma_i(\bv{A})) = \frac{2\gamma}{\alpha} \Sum_f(\bv{A})
	\end{align}
	where the bound on $ \Pr[i\in \bar T] $ follows from the fact that $i \in \bar T$ only if 
	$a_1(1-\alpha)^t \in (1\pm \gamma) \sigma_i(\bv{A})$ for some $t$. This holds only if 
	$a_1 \in  (1\pm \gamma ) \left (\frac{\sigma_i(\bv{A})}{(1-\alpha)^{\lceil \log_{\sigma_i(\bv{A})} 1-\alpha \rceil}} \right )$, which occurs with probability $\frac{2\gamma}{\alpha}$ since $a_1$ is chosen uniformly in the range $[1-\alpha,1]$. 
	By a Markov bound applied to \eqref{expectedError}, with probability $9/10$ we have $\sum_{i \in \bar T} f(\sigma_i(\bv{A})) \le \frac{20\gamma}{\alpha} \Sum_f(\bv{A})$. Combining with \eqref{eqn:qt-main}, \eqref{eqn:qt-1}, and\eqref{eqn:qt-2}, we obtain, with probability 9/10:
	\begin{align*}
			\left (1- \epsilon - \frac{40c_1\gamma}{\alpha}  \right ) \Sum_f(\bv{A}) \le \sum_{t=0}^\infty  \Sum_{q_t}(h^{\gamma}_{R_t}(\bv{A})\bv{A}) \le \left (1+ \epsilon + \frac{40c_1\gamma}{\alpha}  \right ) \Sum_f(\bv{A}).
	\end{align*}
\end{proof}

%Of course we cannot compute the entire infinite sum of spectral windows considered in Lemma \ref{lem:infinite_window_approx}. 
We now show that, as long as the contribution of the smaller singular values of $\bv{A}$ to $\Sum_f(\bv{A})$ is not too large, we can truncate this sum and still obtain an accurate approximation. Specifically: %Additionally, we require that approximating $\sum_f$ on a window in inexpensive.

\begin{corollary}\label{lem:truncated_sum_approx} For any PSD $\bv{A} \in \R^{d\times d}$ with $\norm{\bv{A}}_2 \le 1$, let $f: \R \rightarrow \R^+$ be a function such that, for any $\epsilon > 0$ there exists $\lambda_f(\epsilon)$ such that for $x \in [0,\lambda_f(\epsilon)]$, $f(x) \le \frac{\epsilon}{n}\Sum_f(\bv{A})$.
	%, and 
	%3) $|f'(x)| \le \frac{c_2}{x}$ for any $x \in [0,1]$.
	Given parameters $\alpha,\epsilon \in (0,1)$, and gap parameter $\gamma < \alpha$, 
	for $t \ge 0$ define $R_t$, $h^{\gamma}_{R_t}$ and, $\bar R_t$ as in Lemma \ref{lem:infinite_window_approx}. Let $q_t$ be a well-behaved approximation to $f$ on $\bar R_t$ as in Lemma~\ref{lem:infinite_window_approx}. With probability $9/10$:
	\begin{align}\label{truncated_window_sum_approx}
	\left(1 - \frac{40c_1\gamma}{\alpha} - 5c_1\epsilon \right) \Sum_f(\bv{A}) \le \sum_{t=0}^{\lceil \log_{1-\alpha} \lambda_f(\epsilon) \rceil} \Sum_{q_t}(h^{\gamma}_{R_t}(\bv{A})\bv{A}) \le \left(1 + \frac{40c_1\gamma}{\alpha} + 5c_1\epsilon \right) \Sum_f(\bv{A}). 
	\end{align}
\end{corollary}
\begin{proof}
	This follows from Lemma \ref{lem:infinite_window_approx} along with the small tail assumption. Specifically, by \eqref{eqn:qt-2}:
	\begin{align*}
	\sum_{\lceil \log_{1-\alpha} \lambda_f(\epsilon) \rceil+1}^\infty \Sum_{q_t}(h^{\gamma}_{R_t}(\bv{A})\bv{A}) &\le \sum_{\{i | \sigma_i(\bv{A}) \le \lambda_f(\epsilon)\}} 4c_1f(\sigma_i(\bv{A})) \le  4c_1n  \cdot \frac{\epsilon}{n} \Sum_f(\bv{A}) = 4c_1\epsilon \Sum_f(\bv{A}).
	\end{align*}
\end{proof}

We now show that well behaved polynomial approximations exist for the function $x^p$ for general real $p$, whose spectral sums give the Schatten $p$-norms.
\begin{lemma}[Polynomial Approximation of Power Function]\label{lem:poly}
For all $p \in [-1,0)$ and $k \geq 0$ let 
\[
a_k \eqdef \prod_{j = 1}^{k} \left(1 - \frac{p + 1}{j}\right)
\enspace \text{and} \enspace
q_k(x) = \sum_{j = 0}^{k} a_j (1 - x)^j ~.
\]
Then for all $x \in (0, 1]$ and $k \geq 0$ we have $0 \leq a_k \leq 1$ and 
\[
0 \leq x^p - q_{k}(x) \leq \frac{\exp(-kx)}{x} ~.
\]
\end{lemma}

\begin{proof}
Induction shows that the $k$-th derivative of $f(x)= x^p$ at $x \in \R$, is given by
\[
f^{(k)} (x) = \left(\prod_{j = 1}^{k}  (p + 1 - j)\right)  x^{p - k} ~.
\]
Furthermore, direct calculation reveals that for all $x, t \in \R$ and $k \geq 1$
\begin{equation}
\label{eq:why_ai}
\frac{f^{(k)}(t)}{k!} (x - t)^k
= \left(\prod_{j = 1}^{k} \frac{p + 1 - j}{j}\right) \cdot t^{p - k}
\cdot (t - x)^k \cdot (-1)^k
= a_k t^p \left(1 - \frac{x}{t}\right)^k ~.
\end{equation}
Consequently, $q_k(x)$ is the degree $k$ Taylor expansion of $x^p$ about $1$ evaluated at $x$. Furthermore, since $f^{(k + 1)}(t) = \frac{p - k}{t} \cdot f^{(k)}(t)$, the integral form of the remainder of a Taylor series expansion shows
\begin{align*}
|x^p -  q_k(x)|
&= 
\left|
\int_1^{x}
\frac{f^{(k + 1)}(t)}{k!} (x - t)^k dt
\right|
= 
|a_k| \cdot |p - k| \cdot 
\left|
\int_1^{x}
t^{p - 1} \left(1 - \frac{x}{t}\right)^k dt
\right|
\leq
\frac{k - p}{x} \exp(- k x)
\,.
\end{align*}
In the last step we took the worst case of $t = 1$ in the integral and used that since $p \in [-1,0)$ it is the case that $1 - (p + 1)/j \in (0, 1]$ and therefore $0 \leq a_k \leq 1$. Consequently, $\lim_{i \rightarrow \infty} q_k(x) = x^p$ and
\[
0 \leq x^p - q_k(x) = \sum_{j = k + 1}^\infty a_j (1 - x)^j
\leq (1 - x)^{k + 1} \sum_{j = 0}^{\infty} (1 - x)^j
\leq \frac{\exp(- kx)}{x} \,. 
\]
\end{proof}

\begin{corollary}\label{cor:poly}
For any $p,\epsilon > 0$, $a,b \in (0,1]$ with $a < b$, there is a polynomial $q$ of degree $O \left (\log \left (\frac{b}{a\epsilon}\right) \cdot \frac{b}{a} + p\right)$ such that:
\begin{itemize}
\item Multiplicative Approximation: $|q(x) - x^p| \le \epsilon x^p$ for $x \in [a,b]$.
\item Monotonicity: $q(y) < q(x)$ for all $y < x$.
\item Range Preserving: $q(0) = 0$. 
\end{itemize}
\end{corollary}
\begin{proof}
Set $p' = p-\lceil p \rceil$ and $i = c_1 \log \left (\frac{b}{a\epsilon}\right) \cdot  \frac{b}{a}$ for large enough constant $c_1$.
Instantiate Lemma \ref{lem:poly} with $p'$ and $i$ to obtain $q_i(x)$.
Set $q(x) = b^p \cdot (x/b)^{\lceil p \rceil} \cdot q_{ i}(x/b)$. It is clear that $q(0) = 0$. Further, 
\begin{align*}
|x^p - q(x)| &= b^p (x/b)^{\lceil p \rceil} \cdot | (x/b)^{p'} - q_{i}(x/b)|\\
&\le x^{p} (b/x)^{1+p'} \cdot \exp \left (-c_1 \log \left (\frac{b}{a\epsilon}\right) \frac{b}{a} \cdot \frac{x}{b} \right )
\end{align*}
For $x \in [a,b]$, $x/b \ge a/b$. Further, $(b/x)^{1+p'} \le b/a$ and so if we choose $c_1$ large enough we have $|x^p - q(x)| < \epsilon x^p$. We finally show monotonicity. We have $x^{p'} - q_i(x) = \sum_{j=i+1}^\infty a_j (1-x)^j$. All $a_j$ are positive, so this difference is monotonically decreasing in $x$. We thus have, for $y < x$:
\begin{align*}
\frac{q(y)}{q(x)} = \frac{(y/b)^{\lceil p \rceil} q_i(y/b)}{(x/b)^{\lceil p \rceil} q_i(x/b)} \le \frac{y}{x} \cdot \frac{q_i(y/b)}{q_i(x/b)} \le \frac{y}{x} \cdot \frac{(y/b)^{p'}}{(x/b)^{p'}} \le 1
\end{align*}
since $p' \in [0,1]$. This gives us monotonicity.
%in Lemma \ref{lem:poly} we have $f'_{p',i}(x) = p' x^{p'-1}$ and hence:
%\begin{align*}
% q'(x) = b^{-p'} \cdot \left (\lceil p \rceil x^{\lceil p\rceil -1} \cdot f_{p',i}(x/b) + x^{\lceil p\rceil} p'x^{p'-1} \right )
%\end{align*}
%\todo{Finish up. Got confused by derivative computations.}
\end{proof}
We now combine the approximations of Corollary \ref{cor:poly} with Corollary \ref{lem:truncated_sum_approx} to give our improved Schatten norm estimation algorithm.

\begin{lemma}[Schatten Norms via Polynomial Approximation]\label{lem:polyP} For any $\bv{A} \in \R^{n \times d}$ with $\norm{\bv{A}}_2 \le 1$, $p > 0$ and error parameter $\epsilon \in (0,1)$, for sufficiently small $c_1,c_2$, let $\alpha = 1/2$, $\lambda = \left (\frac{c_1^{2/p}\epsilon^{2/p}}{d^{2/p}} \right ) \norm{\bv{A}}_p^2$, and $\gamma = c_2 \epsilon$.
For $t \ge 0$ define $R_t$, $h^{\gamma}_{R_t}$, and $\bar R_t$ as in Lemma \ref{lem:infinite_window_approx}. Let $q_t(x)$ be as defined in Corollary \ref{cor:poly} with $p' = p/2$, $\epsilon' = c_3 \epsilon$ for sufficiently small $c_3$ and $[a,b] = R_t$. With probability $9/10$:
\begin{align}\label{truncated_window_sum_poly}
\sum_{t=0}^{\lceil \log 1/\lambda \rceil} \Sum_{q_t}(h^{\gamma}_{R_t}(\bv{A}^T\bv{A}) \bv{A}^T\bv{A}) \in (1\pm \epsilon) \norm{\bv{A}}_p^p. 
\end{align}
\begin{proof}
We apply Corollary \ref{lem:truncated_sum_approx} to $\bv{A}^T\bv{A}$, with $f(x) = x^{p/2}$, $\epsilon = c_1 \epsilon$, $\gamma = c_2 \epsilon$ and $\alpha = 1/2$.
$q_t$ satisfies the necessary conditions by Corollary \ref{cor:poly}. Further, we have $\lambda_f(c_1\epsilon) =  \left (\frac{c_1^{2/p}\epsilon^{2/p}}{d^{2/p}} \right ) \norm{\bv{A}}_p^2 = \lambda$.
Plugging in $\alpha = 1/2$, $\gamma = c_2\epsilon$ into \eqref{truncated_window_sum_approx} we have:
\begin{align*}
\sum_{t=0}^{\lceil \log 1/\lambda \rceil} \Sum_{q_t}(h^{\gamma}_{R_t}(\bv{A}^T\bv{A})\bv{A}^T\bv{A}) \in \left(1\pm c_3\epsilon \pm \frac{40c_2\epsilon}{1/2} \pm c_1\epsilon \right) \Sum_f(\bv{A}\bv{A}^T) = \norm{\bv{A}}_p^p
\end{align*}
which gives $(1 \pm \epsilon)$ approximation if we set $c_1,c_2,c_3$ small enough.
\end{proof}
\end{lemma}

\begin{theorem}[Schatten Norm Polynomial Approximation Algorithm]\label{thm:polyAlgo} For any $\bv{A} \in \R^{n \times d}$, $p > 0$ and $k \in [d]$ there is an algorithm returning $X \in (1 \pm \epsilon) \norm{\bv{A}}_p^p$ that runs in time:
\begin{align*}
\tilde O \left (\nnz(\bv{A})k + dk^{\omega-1} + \frac{\max \{p, 1/p^3 \}}{\epsilon^{\max \{3,1+1/p\}}} \cdot \left [\nnz(\bv{A}) + \sqrt{\nnz(\bv{A}) [d\cdot d_s(\bv{A}) + dk ]}\cdot (d/k)^{\max\{0,1/p-1/2\}}\right ] \right )
\\\text{ or } \tilde O \left (\nnz(\bv{A})k + dk^{\omega-1}  + \frac{\max\{p, 1/p^3 \}}{\epsilon^{\max \{3,1+1/p\}}} \cdot \left [(\nnz(\bv{A})+dk) (d/k)^{1/p} \right ] \right )
\end{align*} time for sparse $\bv{A}$ or $$ \tilde O \left (nd^{\omega(\log_d k)-1} + \frac{\max\{p, 1/p^3 \}}{\epsilon^{\max \{3,1+1/p\}}} \cdot \left [ nd + \frac{n^{1/2}d^{1 + 1/p}}{k^{\max\{0,1/p-1/2\}}}  \right ] \right )$$ for dense $\bv{A}$.
\end{theorem}
\begin{proof}
We apply Lemma \ref{lem:polyP}, first scaling $\bv{A}$ so that $\norm{\bv{A}}_2 \le 1$. We apply a trace estimator to approximate $\Sum_{q_t}(h^{\gamma}_{R_t}(\bv{A}^T\bv{A})\bv{A}^T\bv{A}) = \tr(q_t(h^{\gamma}_{R_t}(\bv{A}^T\bv{A})\bv{A}^T\bv{A}))$ up to $(1 \pm \epsilon)$ multiplicative error plus very small additive error for each of our $O(\log 1/\lambda) = \tilde O(\max\{1,1/p\})$ windows. The trace estimation requires $\tilde O \left ( \frac{\lceil p \rceil }{\epsilon^2} \right)$ applications of $h^{\gamma}_{R_t}(\bv{A}^T\bv{A})\bv{A}^T\bv{A}$, since the degree of $q_t$ as given in Corollary \ref{cor:poly} is $\tilde O(1) + p$.

The cost to apply $h^{\gamma}_{R_t}(\bv{A}^T\bv{A})\bv{A}^T\bv{A}$ to a vector is given by Theorem \ref{thm:window}. %with $R_{\lceil \log \lambda \rceil}$ costing the most. For $R_{\lceil \log \lambda \rceil}$ we have $\gamma = \Theta(\epsilon)$ and $a = \lambda$ so, 
Following the argument in Corollary \ref{cor:hist_schatten}, if we write $a_t = \frac{\delta_t^{2/p}}{d^{2/p}} \cdot\norm{\bv{A}}_p^2$ for some $\delta_t \ge \epsilon$,
\begin{align*}
\bar \kappa_t = \frac{k\sigma^2_k(\bv{A}) + \sum_{k=1}^d \sigma_i^2(\bv{A})}{d a_t} \le \frac{1}{\delta_t^{2/p}} \left (\frac{d}{k}\right)^{2/p-1}.
\end{align*}
Similarly, $\hat \kappa_t = \frac{\sigma_{k+1}^2(\bv{A})}{a_t} \le \frac{1}{\delta_t^{2/p}} \left (\frac{d}{k}\right )^{2/p}$.
We can further optimize our $\epsilon$ dependence by noting that it suffices to approximate $\tr(q_t(h^{\gamma}_{R_t}(\bv{A}^T\bv{A})\bv{A}^T\bv{A}))$ up to multiplicative error $1 \pm \tilde O \left (\frac{\epsilon}{\min \{1,\delta_t\}\log(1/\lambda)}  \right )$ since even if there are $d$ singular values below $a_{t-1}$, they will contribute at most a $\delta_t$ fraction of $\norm{\bv{A}}_p^p$. So our total cost of approximating the trace for each of our windows using the first runtime of Theorem \ref{thm:window} is:
\begin{align*}
\tilde O \left (\sum_{t=0}^{\lceil \log 1/\lambda \rceil} \frac{\min \{1,\delta_t \}^2 \cdot \lceil p \rceil (\log 1/\lambda)^2}{\epsilon^2} \cdot \frac{\nnz(\bv{A}) + \sqrt{\nnz(\bv{A}) [d\cdot d_s(\bv{A}) + dk ]}\cdot (d/k)^{1/p-1/2}}{\epsilon \cdot \delta_t^{1/p}} \right )% \\
%= \tilde O \left (\frac{\lceil 1/p \rceil}{\epsilon^3} \cdot \left [\nnz(\bv{A}) + \sqrt{\nnz(\bv{A}) [d\cdot d_s(\bv{A}) + dk ]}\cdot (d/k)^{1/p-1/2}\right ] \right ).
\end{align*}
Factoring out the $\nnz(\bv{A}) + \sqrt{\nnz(\bv{A}) [d\cdot d_s(\bv{A}) + dk ]}\cdot (d/k)^{1/p-1/2}$ term, we have:
\begin{align*}
\tilde O &\left (\sum_{t=0}^{\lceil \log 1/\lambda \rceil} \frac{\min \{1,\delta_t \}^2 \cdot \max \{p,1/p^2 \} }{\epsilon^3 \delta_t^{1/p}} \right ) = \tilde O \left (\frac{\max \{p,1/p^2 \}}{\epsilon^3} \cdot \left [\sum_{\{t: \delta_t < 1\}} \frac{1}{\delta_t^{1/p-2}} + \sum_{\{t: \delta_t \ge 1\}} \frac{1}{\delta_t^{1/p}} \right ]  \right )\\
&=\tilde O \left (\frac{\max \{p,1/p^2 \}}{\epsilon^3} \cdot \left [ \frac{1}{\epsilon^{1/p-2}} \left (1 + \frac{1}{2^{1/p-2}} + \frac{1}{4^{1/p-2}} + ... + \epsilon^{1/p-2}\right)  + \log(1/\lambda) \right ]  \right ).
\end{align*}
Note that if $p > 2$, $1/p-2 < 0$ the terms in the geometric sum are increasing so the largest term is $\epsilon^{1/p-2}$ and so the whole thing is dominated by $\tilde O(\frac{ p \log(1/\lambda)}{\epsilon^3}) = \tilde O(\frac{p}{\epsilon^3})$. If $p < 2$, then the largest term in the geometric sum is $1$ as so similarly the whole thing is $\tilde O(\frac{1/p^2 \log(1/\lambda)}{\epsilon^{\max \{3,1+1/p\}}}) = \tilde O(\frac{ 1/p^3}{\epsilon^{\max \{3,1+1/p\}}})$.

%Similarly we can use standard iterative methods to achieve $\tilde O \left (\frac{\lceil 1/p \rceil}{\epsilon^3} \cdot \left [(\nnz(\bv{A})+dk) (d/k)^{1/p} \right ] \right )$ or dense methods to achieve $\tilde O \left (\frac{\lceil 1/p \rceil}{\epsilon^3} \cdot \left [ nd + n^{1/2}d^{3/2} (d/k)^{1/p-1/2} \right ] \right )$.

A similar argument gives the runtimes for standard iterative and dense methods.
Adding in the cost for deflation-based preconditioning gives the final runtimes. Note that while $\lambda$ depends on $\norm{\bv{A}}_p^p$ which is unknown, it again suffices to lower bound the truncation point using $\sigma_k^p(\bv{A}) \le \frac{1}{k} \norm{\bv{A}}_p^p$. This lower bound gives only a better approximation and our bounds on $\bar \kappa_t$ and $\hat \kappa_t$ still hold.
\end{proof}

% !TEX root = normEstimation.tex
\section{Optimized Runtime Results}\label{sec:upper}

The runtimes given in Theorem \ref{thm:polyAlgo} are quite complex, with many tradeoffs depending on the properties of $\bv{A}$. In this section, we instantiate the theorem showing how to optimize $k$ and achieve our final runtime bounds. For simplicity, we consider the case when $n = d$. 

\subsection{Schatten $p$-Norms}
\label{sec:schatten-p-norms}
We first tackle the relatively simply case when $\bv{A}$ is dense.
\begin{theorem}[Schatten $p$-Norms for Dense Matrices]\label{finalthm:dense}
For any $p \ge 0$ and $\bv{A} \in \mathbb{R}^{n \times n}$ there is an algorithm returning $X \in (1\pm\epsilon)\norm{\bv{A}}_p^p$ which runs in $\tilde O(p \cdot n^2/\epsilon^3)$ time for $p \ge 2$ and 
\begin{align*}
\tilde O \left (\frac{1}{p^3 \cdot\epsilon^{\max \{3,1+1/p\}}} \cdot  n^{\frac{2.3729 - .0994p}{1 + .0435p}} \right )
\end{align*}
time for $p < 2$. If we do not use fast matrix multiplication the runtime is $\tilde O \left (\frac{1}{p^3 \cdot\epsilon^{\max \{3,1+1/p\}}} \cdot  n^{\frac{3 + p/2}{1 + p/2}} \right )$.
\end{theorem}
Note that for the important case of $p = 1$ our runtime is $\tilde O(n^{2.18}/\epsilon^3)$ or $\tilde O(n^{2.33}/\epsilon^3)$ if we do not use fast matrix multiplication. As $p$ approaches $0$, our runtime approaches $\tilde O(n^\omega)$ which is natural, as $p=0$ gives the matrix rank, which seems difficult to determine. As $p$ approaches $2$ the runtime smoothly approaches $\tilde O(n^2/\epsilon^3)$, which is then required for all $p \ge 2$. 
\begin{proof}
The bound for $p \ge 2$ follows immediately from Theorem \ref{thm:polyAlgo} with $k$ set to $0$. For $p < 2$ we have runtime:
\begin{align*}
\tilde O \left (n^{\omega(\log_n k)} + \frac{1}{p^3 \cdot\epsilon^{\max \{3,1+1/p\}}} \cdot \left [ \frac{n^{3/2+1/p}}{k^{1/p-1/2}} \right ] \right )
\end{align*}
We can partially optimize this term, ignoring $p$ and $\epsilon$ factors for simplicity and setting $k$ to equalize the coefficients on $n$. Our optimization is identical to the argument in Corollary \ref{cor:hist_schatten} 
(set $k = n^{\frac{1/p-1/2}{1/p+1/2}}$ when not using fast matrix multiplication), giving the stated runtimes.
\end{proof}

We now tackle the more complex case when $\bv{A}$ is sparse. We first consider $p \ge 2$.

\begin{theorem}[Schatten $p$-norms, $p \ge 2$, for Sparse Matrices]\label{thm:largep}
For any $p \ge 2$, and $\bv{A} \in \R^{n \times n}$ there is an algorithm returning $X \in (1\pm \epsilon) \norm{\bv{A}}_p^p$ with high probability in time: 
\begin{align*}
\tilde O \left (\frac{p}{\epsilon^3} \cdot \sqrt{\nnz(\bv{A}) n d_s(\bv{A})} \right ) \text{ or } \tilde O \left (\frac{p}{\epsilon^3} \left [\nnz(\bv{A}) n^{\frac{1}{1+p}} + n^{1+\frac{2}{1+p}} \right ] \right)
\end{align*}
\end{theorem}
Note that if our matrix is uniformly sparse, $d_s(\bv{A}) = O(\nnz(\bv{A})/n)$ and so the first runtime becomes $\tilde O(p \cdot \nnz(\bv{A})/\epsilon^{3})$. The second runtime can be better when $\bv{A}$ is not uniformly sparse.
\begin{proof}
The first runtime follows by setting $k = 0$ in Theorem \ref{thm:polyAlgo}. Note that $\nnz(\bv{A}) \le \sqrt{\nnz(\bv{A}) n d_s(\bv{A})}$. For the second, we use the second runtime of Theorem \ref{thm:polyAlgo} which gives
\begin{align*}
\tilde O \left (\nnz(\bv{A})k + nk^{\omega-1}  + \frac{p}{\epsilon^3} \cdot \left [(\nnz(\bv{A})+nk) (n/k)^{1/p} \right ] \right )
\end{align*}
Setting $k = \left (\frac{n}{k}\right )^{1/p}$ to balance the coefficients on the $\nnz(\bv{A})$ terms gives $k = n^{\frac{1}{1+p}}$ which yields the result. Note that the $nk^{\omega-1}$ term is dominated by the cost of the regressions, even if we do not use fast matrix multiplication this term will be $nk^2$.
\end{proof}

Finally, we consider the most complex cost, $p \le 2$ for sparse matrices. We have:
\begin{theorem}[Schatten $p$-norms, $p \le 2$, for Sparse Matrices]\label{thm:smallp}
For any $p \in (0,2]$, and $\bv{A} \in \R^{n \times n}$ there is an algorithm returning $X \in (1\pm \epsilon) \norm{\bv{A}}_p^p$ with high probability in time: 
\begin{align*}
\tilde O \left (\frac{1}{p^3 \cdot \epsilon^{\max \{3,1/p\}}} \cdot \left [\nnz(\bv{A}) n^{\frac{1/p-1/2}{1/p+1/2}} \sqrt{\gamma_s} + \sqrt{\nnz(\bv{A})} \cdot n^{\frac{4/p-1}{2/p+1}} \right ]\right )\\ \text{ or } \tilde O \left (\frac{1}{p^3 \cdot \epsilon^{\max \{3,1/p\}}} \left [\nnz(\bv{A}) n^{\frac{1}{1+p}} + n^{1+\frac{2}{1+p}} \right ] \right)
\end{align*}
where $\gamma_s = \frac{d_s(\bv{A})n}{\nnz(\bv{A})} \ge 1$.
\end{theorem}
Note that in the special case of $p=1$ the first runtime gives $\tilde O \left (\frac{1}{\epsilon^3} \left [\nnz(\bv{A})n^{1/3} \sqrt{\gamma_s} + \sqrt{\nnz(\bv{A})} n\right ] \right )$. The second term here is at worst $n^2$, and could be much smaller for sparse $\bv{A}$.
\begin{proof}
For the second runtime, we use the second runtime of Theorem \ref{thm:polyAlgo}, balancing costs exactly as in Theorem \ref{thm:largep} (setting  $k = n^{\frac{1}{1+p}}$). For the first runtime,
 we consider the first runtime of Theorem~\ref{thm:polyAlgo} which gives:
\begin{align*}
\tilde O \left (\nnz(\bv{A})k + nk^{\omega-1} + \frac{1}{p^3 \cdot \epsilon^{\max \{3,1/p\}}}\cdot \left [\nnz(\bv{A}) + \sqrt{\nnz(\bv{A}) [n\cdot d_s(\bv{A}) + nk ]}\cdot (n/k)^{1/p-1/2}\right ] \right )\\
\tilde O \left (\frac{n d_s(\bv{A})k}{\gamma_s} + nk^{\omega-1} + \frac{1}{p^3 \cdot \epsilon^{\max \{3,1/p\}}}\cdot \left [\frac{nd_s(\bv{A}) + n\sqrt{d_s(\bv{A})}\sqrt{k}}{\sqrt{\gamma_s}} \right ]\cdot \left (\frac{n}{k} \right )^{1/p-1/2}\right )
\end{align*}
Ignoring $\epsilon$, $p$, and $\gamma_s$ dependence, one way we can balance the costs is by setting:
\begin{align*}
nd_s(\bv{A}) k = nd_s(\bv{A}) \left (\frac{n}{k} \right )^{1/p-1/2}
\end{align*}
which gives $k = n^{\frac{1/p-1/2}{1/p+1/2}}$ and final runtime:
\begin{align*}
\tilde O \left (\frac{1}{p^3 \cdot \epsilon^{\max \{3,1/p\}}} \cdot \left [\nnz(\bv{A}) n^{\frac{1/p-1/2}{1/p+1/2}} \sqrt{\gamma_s} + n\sqrt{d_s(\bv{A})/\gamma_s} \cdot k^{3/2} \right ]\right ) = \\
\tilde O \left (\frac{1}{p^3 \cdot \epsilon^{\max \{3,1/p\}}}\cdot \left [\nnz(\bv{A}) n^{\frac{1/p-1/2}{1/p+1/2}} \sqrt{\gamma_s} + n\sqrt{d_s(\bv{A})/\gamma_s} \cdot n^{\frac{3/p-3/2}{2/p+1}} \right ]\right ).
\end{align*}
Note that we drop the $nk^{\omega-1}$ term as it is dominated by the last term with $k^{3/2}$ in it. This gives our final result by noting that $d_s(\bv{A})/\gamma_s = \nnz(\bv{A})/n$, the average row sparsity.
\end{proof} 

\subsection{Constant Factor Approximation without Row-Sparsity Dependence}
\begin{theorem}[Removing the Uniform Sparsity Constraint]\label{uniformSparsity}
  Let $\gamma \in (0,1)$, $p \geq 1$, and $\bv{A} \in \mathbb{R}^{n \times n}$. There is an algorithm returning
  $X$ with $\|\bv{A}\|_p \leq X = O(1/\gamma) \|\bv{A}\|_p$ with high probability in time
    $\tilde{O} (p\nnz(\bv{A}) n^{\gamma})$ for $p \geq 2$, and in time 
    $\tilde O \left (\frac{1}{p^3} \left [\nnz(\bv{A}) n^{\frac{1/p-1/2}{1/p+1/2} + \gamma/2} + \sqrt{\nnz(\bv{A})} \cdot n^{\frac{4/p-1}{2/p+1}}\right ] \right)$ for $p < 2$. 
  \end{theorem}
\begin{proof}
For $i \in \{1, 2, \ldots, \lceil \frac{1}{\gamma} \rceil \}$, 
let $\bv{A^i} \in \mathbb{R}^{n \times n}$ be the matrix whose  rows agree with that of $\bv{A}$
provided the number of non-zero entries in these rows 
are in the interval 
$[n^{\gamma i} \cdot \frac{\nnz(\bv{A})}{n}, n^{\gamma (i+1)} \cdot \frac{\nnz(\bv{A})}{n})$, while
  the remaining rows of $\bv{A^i}$ are set to $0^n$. 
  Let $\bv{A^0} = \bv{A} - \sum_{i=1}^{\lceil \frac{1}{\gamma} \rceil} \bv{A^i}$. 
  Since $p \geq 1$ we can apply the triangle inequality,  
  \begin{eqnarray}\label{eqn:upper}
    \|\bv{A}\|_p \leq \sum_{i=0}^{\lceil \frac{1}{\gamma} \rceil} \|\bv{A^i}\|_p \leq (\lceil \frac{1}{\gamma} \rceil + 1 ) \max_i \|\bv{A^i}\|_p.
    \end{eqnarray}
  For any matrix $\bv{B}$ obtained from a matrix $\bv{A}$ by replacing some rows of $\bv{A}$ with $0^n$ and preserving the remaining rows, 
  $\|\bv{B}\|_p \leq \|\bv{A}\|_p$. This follows from the fact that
  $\|\bv{B}x\|_2 \leq \|\bv{A}x\|_2$ for all vectors $x$, together with the min-max
  theorem for singular values. Hence,
  \begin{eqnarray}\label{eqn:lower}
\|\bv{A}\|_p \geq \max_i \|\bv{A^i}\|_p.
    \end{eqnarray}
  Combining (\ref{eqn:upper}) and (\ref{eqn:lower}), an $O(1)$-approximation to $\max_i \|\bv{A^i}\|_p$ is an
  $O(\frac{1}{\gamma})$-approximation to $\|\bv{A}\|_p$.

  We remove the zero rows from
  the $\bv{A^i}$, obtaining corresponding matrices $\bv{B^i}$. By definition of $\bv{A^i}$, the number of rows in
  $\bv{B^i}$ is at most $n^{1-\gamma i}$.
  For each $\bv{B^i}$, which is an $s_i \times n$ matrix for $s_i \leq n^{1-\gamma i}$, we can right-multiply it by an $n \times t_i$
  OSNAP matrix $\bv{T^i}$, with $t_i = O(n^{1-\gamma i} \log n)$ columns and $O(\log n)$ non-zero entries per column, so that
  $\|\bv{B^i} \bv{T^i}\|_p = (1 \pm 1/2) \|\bv{B^i}\|_p$ for all $i$ with probability $1-1/\poly(n)$, see \cite{c16} (the fact
  that $\|\bv{B^i} \bv{T^i}\|_p = (1 \pm 1/2) \|\bv{B^i}\|_p$ follows from the fact that $\bv{T^i}$ is a subspace embedding
  of the row space of $\bv{B^i}$ together with Lemma C.2 of \cite{lnw14}). 
  The time to compute $\bv{C^i} = \bv{B^i} \bv{T^i}$ is $\tilde{O}(\nnz(\bv{A}))$. Since $\bv{T^i}$ has $O(\log n)$ non-zero entries
  per row, each row of $\bv{C^i}$ has a number of non-zero entries within an $O(\log n)$ factor of the corresponding
  row of $\bv{B^i}$. 
  
  We compute $O(1)$-approximations to the $\|\bv{C^i}\|_p$ for each $i$, and then take the maximum over $i$. To do so,  
  for $p \geq 2$ we run the algorithm of Theorem \ref{thm:largep} on each $\bv{C^i}$, while for $p < 2$ we
  run the algorithm of Theorem \ref{thm:smallp}.

  Let $d_s^i$ denote the parameter $d_s$ when run on $\bv{C^i}$.
  Then $d_s^i = O(n^{\gamma i + \gamma} \frac{\nnz(\bv{A})}{n} \log n)$. We can assume $\bv{C^i}$ is a square $t_i \times t_i$ matrix by padding with zeros,
  where $t_i = O(n^{1-\gamma i} \log n)$. Hence, $t_i d_s^i = O(\nnz(\bv{A}) n^{\gamma} \log n)$. Also,
  $\nnz(\bv{C^i}) = O(\nnz(\bv{B^i}) \log n) = O(\nnz(\bv{A})\log n)$.

  The total time to apply Theorem \ref{thm:largep} for $p \geq 2$ 
  across all $\bv{C^i}$ is $\tilde{O} (p \nnz(\bv{A}) n^{\gamma})$.
  The total time to apply Theorem \ref{thm:smallp} for $p < 2$ across all $\bv{C^i}$ is
  $\tilde O \left (\frac{1}{p^3} \left [\nnz(\bv{A}) n^{\frac{1/p-1/2}{1/p+1/2} + \gamma/2} + \sqrt{\nnz(\bv{A})} \cdot n^{\frac{4/p-1}{2/p+1}}\right ] \right ) $, 
  using that the $\gamma_s$ of that theorem is $O(n^{\gamma} \log n)$. 
\end{proof}

\subsection{SVD Entropy}

%The SVD entropy, which we denote by $H(\bv{A})$ is given by $\Sum_f(\bv{A})$ where $f(x) = -\frac{x}{\norm{\bv{A}}_1} \log (x/\norm{\bv{A}}_1)$. This function can be estimated using our general histogram algorithm (Theorem \ref{thm:histogram_algorithm}). %We can write $g(x) = -x \log x$ and have 
%\begin{align*}
%H(\bv{A}) = \frac{1}{\norm{\bv{A}}_1} \sum_{i=1}^d g(\sigma_i(\bv{A})) + d \cdot \frac{\log(1/\norm{\bv{A}}_1)}{\norm{\bv{A}}_1}
%\end{align*}
%We have $f'(x) = \frac{\log (x/\norm{\bv{A}}_1)}{\norm{\bv{A}}_1} + \frac{1}{\norm{\bv{A}}_1}$ and $\frac{f(x)}{x} = \log (x/\norm{\bv{A}}_1)$. Thus, as long as all singlar values are bounded away from $\norm{\bv{A}}_1$ by a constant, we have smoothness $\delta_f = O(1)$ on the spectral range. There is at most one singular value $\ge \norm{\bv{A}}_1/2$, and if this is the case we can simply deflate off this direction and compute its contribution to the spectral sum separately.
In this section, we will show how to approximately estimate the SVD entropy of a matrix $\bv{A}$ 
assuming its condition number $K \eqdef \sigma_1(\bv{A})/\sigma_n(\bv{A})$ to be bounded by $n^{c_1}$.
Recall that the SVD entropy of a matrix $\bv{A}$ is given by 
$H(\bv{B}) \eqdef - \sum_i \sigma_i(\bv{B}) \log \sigma_i(\bv{B})$, where 
$\bv{B} \eqdef \frac{\bv{A}}{\norm{\bv{A}}_1}$ is the matrix $\bv{A}$ scaled inversely
by its Schatten $1$-norm. Using the results in Section~\ref{sec:schatten-p-norms},
we obtain a $(1 + \epsilon)$ approximation $Z$ to $\norm{\bv{A}}_1$ i.e., $\norm{\bv{A}}_1 \leq Z \leq (1+\epsilon)\norm{\bv{A}}_1$. Define $\widetilde{\bv{B}}\eqdef \frac{\bv{A}}{Z}$. This means that
\begin{align*}
\abs{H(\widetilde{\bv{B}}) - H(\bv{B})} &= \abs{\sum_i \left(\frac{\sigma_i(\bv{A})}{\norm{\bv{A}}_1} \log \frac{\sigma_i(\bv{A})}{\norm{\bv{A}}_1} - \frac{\sigma_i(\bv{A})}{Z} \log \frac{\sigma_i(\bv{A})}{Z}\right)} \\
&\leq \abs{\sum_i \sigma_i(\bv{A}) \left(\frac{1}{\norm{\bv{A}}_1}-\frac{1}{Z}\right) \log \frac{\sigma_i(\bv{A})}{\norm{\bv{A}}_1}} + \abs{\sum_i \frac{\sigma_i(\bv{A})}{Z} \log \frac{\norm{\bv{A}}_1}{Z}} \\
&\leq \frac{\epsilon}{1+\epsilon} \abs{\sum_i \frac{\sigma_i(\bv{A})}{\norm{\bv{A}}_1} \log \frac{\sigma_i(\bv{A})}{\norm{\bv{A}}_1}} + \epsilon
\leq 2 \epsilon \log n.
\end{align*}

The following thoerem gives our main result for SVD entropy.
\begin{theorem}\label{theo:SVDentropy}
 Given any PSD $\bv{A} \in \R^{n \times n}$ and error parameter $\epsilon \in (0,1)$.  
Then, there exists an algorithm that computes the approximate SVD entropy of $\bv{A}$  and outputs 
$\hat S$, such that, with probability $9/10$
\begin{align}\label{approxsymnorm}
\hat{S} \in (1\pm O(\epsilon))H(\bv{B}),
\end{align}
where $\bv{B}= \frac{\bv{A}}{\norm{\bv{A}}_1}$.
The runtime of the algorithm will be
\begin{align*}
 \tilde O \left (\frac{\nnz(\bv{A}) n^{\frac{1}{3}}+n^{\frac{3}{2}}\sqrt{d_s}}{\tilde{\epsilon}^{6}}\log (1/\epsilon)  \right ),
\end{align*}
where $d_s(\bv{A}) = O(\nnz(\bv{A})/n)$,	and 
$\tilde{\epsilon}=O(\frac{\epsilon}{\log n})$.
\end{theorem}
The algorithm we consider to estimate the SVD entropy of $\widetilde{\bv{B}}$
follows the techniques of \cite{harvey2008sketching}. %We first define an auxilary quantity spectral sum. %The $\alpha^{th}$ R\'{e}nyi entropy $H_\alpha(\bv{A})$ is given by $\Sum_f(\bv{A})$ with $f(x) = \log(\norm{\bv{A}}_\alpha^\alpha)/(1-\alpha)$.
The algorithm depends on the ${\beta}^{th}$ SVD Tsallis entropy 
$T^{SVD}_\beta(\widetilde{\bv{B}})$ of matrix $\widetilde{\bv{B}}$, which 
is given by  
$$T^{SVD}_\beta(\widetilde{\bv{B}})= \frac{1-\norm{\widetilde{\bv{B}}}_\beta^\beta}{\beta-1}.$$
We can describe the \emph{multi-point interpolation} method for SVD entropy estimation with additive and multiplicative
error approximations using~\cite[Algorithm 1]{harvey2008sketching}, see Algorithm~\ref{algo:SVDentropy}.

Here, given an error parameter $\epsilon$, we compute an approximate SVD Tsallis entropy $\tilde{T}^{SVD}_{1+\alpha_i}$ at $k_1$ different points
$\{1+\alpha_0,\ldots,1+\alpha_{k_1}\}$, where $k_1 \eqdef \log \frac{1}{\epsilon} + \log c_1 + \log \log n$. $\alpha_1, \cdots, \alpha_{k_1}$ are defined as follows. Let $\ell \eqdef 1/(2c_1(k_1+1)\log n)$ and let $g(\cdot)$ be defined as:
\begin{align}\label{eq:points}
	g(y) = \frac{k_1^2 \ell y - \ell (k_1^2 + 1)}{2k_1^2 + 1}, \mbox{ then, } \alpha_i \eqdef g(\cos(i \pi / k_1)).
\end{align}
Now define an error parameter $\widetilde{\epsilon} \eqdef \epsilon / (12c_1(k_1+1)^3 \log n)$.

\begin{algorithm}[tb!]
\caption{SVD entropy estimation via multi-point interpolation.\label{algo:SVDentropy}}
\begin{algorithmic}%[1]
  \State {\bf Input:}    $\bv{A}\in\R^{n\times n}$, $\epsilon\in(0,1)$. 
  \State {\bf Output:} Approximate SVD entropy $\hat S$.
  \State Compute $\widetilde{\bv{B}}= \frac{\bv{A}}{Z}$, where $Z$ is  
  a $(1 + \epsilon)$ approximation to $\norm{\bv{A}}_1$.
  \State Choose $k_1$ points $\alpha_1, \cdots, \alpha_{k_1}$ as in \eqref{eq:points}, 
  and set $\widetilde{\epsilon} = \epsilon / (12c_1(k_1+1)^3 \log n)$.
    \For {$i=1:k_1$}
        \State Compute $Z_{1+\alpha_i}$, a $(1+\widetilde{\epsilon})$ -approximation for Schatten norm
        $\norm{\widetilde{\bv{B}}}_{1+\alpha_{i}}$.
        \State Compute $\tilde{T}^{SVD}_{1+\alpha_i}(\widetilde{\bv{B}})=(1-Z_{1+\alpha_i}/Z^{1+\alpha_i})/\alpha_i$.
      \EndFor
    \State Return $\hat S$ an estimate of $T(0)$ by interpolation using the points $\tilde{T}^{SVD}_{1+\alpha_i}(\widetilde{\bv{B}})$.
\end{algorithmic}
\end{algorithm}

Algorithm~\ref{algo:SVDentropy} gives the stepwise algorithm. The runtime of the algorithm is dominated
by the cost of computing  $Z_{1+\alpha_i}$, a $(1+\widetilde{\epsilon})$ -approximation for Schatten norm
        $\norm{\widetilde{\bv{B}}}_{1+\alpha_{i}}$, particularly the smallest power Schatten norm.
        We get the smallest power when 
        $\alpha_i=\frac{-1}{2c_1(k_1+1)\log n}$.  We have $k_1=\log(1/\epsilon)$ such Schatten norms to be
        estimated. Thus, we obtain the runtime in Theorem~\ref{theo:SVDentropy} by using 
        $p=\alpha =O\left(1-\frac{1}{\log (n/{\epsilon})}\right)$
        in the runtime for Schatten norm estimation given in Corollary~\ref{cor:hist_schatten},
        and observing that $n^{1/\log(n/\epsilon)}$ is between 1 and 2 (we set $p\approx 1$).

\paragraph{Additive approximation}: Section 3.3.2 of~\cite{harvey2008sketching}
tells us that if in Algorithm~\ref{algo:SVDentropy},
the approximate $\tilde{T}^{SVD}_{1+\alpha_i}(\widetilde{\bv{B}})$ is
an additive $\widetilde{\epsilon}$ approximation to
${T}^{SVD}_{1+\alpha_i}(\widetilde{\bv{B}})$ for every $i \in [k_1]$,
then we can use these to compute an additive-$\epsilon$ approximation to $H(\widetilde{\bv{B}})$. 
Since $\frac{-1}{2c_1 k_1 \log n} < \alpha_i < \frac{-1}{16c_1 k_1^3 \log n}$, and
since $\norm{\widetilde{\bv{B}}}_1=1$, obtaining $\frac{1}{16c_1 k_1^3 \log n}$ 
multiplicative approximation to $\norm{\widetilde{\bv{B}}}_{1+\alpha_{k_1}}$ suffices.
This approximation can be obtained using the results of Section~\ref{sec:schatten-p-norms}.

\paragraph{Multiplicative approximation}: %\textbf{TODO}
%Then, we estimate the SVD entropy 
%using interpolation of the points
%$\tilde{H}^{SVD}_{1+\alpha_i}$ or $\tilde{T}^{SVD}_{1+\alpha_i}$ at $\alpha_i=0$. 
Article~\cite{harvey2008sketching} further extends the multi-point interpolation method
to achieve multiplicative approximation for Shannon entropy 
(equivalently for SVD entropy) using the following modifications.
We set the number of interpolation points $k_1=\max\{5,\log(1/\epsilon)\}$. Then,
section 3.4 in~\cite{harvey2008sketching} shows that if  
$\tilde{T}^{SVD}_{1+\alpha_i}(\widetilde{\bv{B}})$ (in Algorithm~\ref{algo:SVDentropy}) for every $i \in [k_1]$ computed are to be
$(1+\widetilde{\epsilon})$-multiplicative approximation to ${T}^{SVD}_{1+\alpha_i}({\bv{B}})$
(where $\widetilde{\epsilon}$ is as defined above),
 then we can achieve multiplicative approximation for $H({\bv{B}})$ using 
 the multi-point interpolation method with these $\tilde{T}^{SVD}_{1+\alpha_i}(\widetilde{\bv{B}})$.
 So, we need to obtain $(1+\widetilde{\epsilon})$-relative error approximations 
 to the  ${T}^{SVD}_{1+\alpha_i}({\bv{B}})$ at each $\alpha_i$.
 
 If the matrix has reasonable (large) SVD entropy (meaning the singular values are uniform 
 and there is no one singular value that is very large), then we can obtain 
 $(1\pm\widetilde{\epsilon})$ approximation
 to the  ${T}^{SVD}_{1+\alpha_i}({\bv{B}})$ by simply computing $(1\pm\widetilde{\epsilon})$
 approximation to $\norm{\widetilde{\bv{B}}}_{1+\alpha_{i}}$ using results 
 from section~\ref{sec:schatten-p-norms}.
 However, if the matrix has very small entropy, i.e., we have one singular value with very large magnitude
 and remaining singular values are very small, then achieving a multiplicative
 approximation will be difficult. 
 This is because, we are approximating  ${T}^{SVD}_{1+\alpha_i}$ of $\widetilde{\bv{B}}$, which
 is matrix $\bv{A}$ scaled by a $(1\pm \epsilon)$ approximation of its Schatten $1$-norm. 
 The $(1+\widetilde{\epsilon})$ approximation obtained for  ${T}^{SVD}_{1+\alpha_i}(\widetilde{\bv{B}})$
 might not be close to ${T}^{SVD}_{1+\alpha_i}({\bv{B}})$ in this case.
 
 This issue is equivalent to the `heavy element' issue (one of the entries in the vector is very large)
 discussed in~\cite{harvey2008sketching}.
 Hence, to overcome the above issue, we can follow the workaround proposed in~\cite{harvey2008sketching}. 
 Specifically, Lemma 5.5 in~\cite{harvey2008sketching} shows that a 
 $(1+\widetilde{\epsilon})$-approximation
 to $1-\hat{\sigma}_1$ together with a $(1+\widetilde{\epsilon})$-approximation
 to $\sum_{j > 1} \hat{\sigma}_j^{1+ \alpha_i}$,
 where $\hat{\sigma}_1 \ge\hat{\sigma}_2 \ge\ldots\ge \hat{\sigma}_n$ 
 are the singular values of $\bv{B}$ with $\sum_i \hat{\sigma}_i = 1$, 
 suffices to get a $(1+\widetilde{\epsilon})$-approximation 
 to  ${T}^{SVD}_{1+\alpha_i}({\bv{B}})$ at $\alpha_i$. 
 It follows that we just need to obtain $(1+\widetilde{\epsilon})$-approximations for 
 these latter two quantities for each $\alpha_i$.
 
%  The first quantity can be computed upto $(1+\widetilde{\epsilon})$-approximation 
%  by estimatiing the top singular value of $\bv{A}$ using 
%  $\Theta(\log n/\sqrt{\widetilde{\epsilon}})$ steps of Lanczos algorithm~\cite{musco2015randomized}.
 It can be shown that when the Krylov method~\cite{musco2015randomized} (or power method) is 
 used to deflate the top singular vector, we have that any unitarily invariant norm 
 of the tail (remaining part of the spectrum) is preserved, see Appendix~\ref{sec:Krylovpreserve}
 for the proof. 
This means we can get $(1+\widetilde{\epsilon})$-approximations 
 to both $\|\bv{A}_{-1}\|_1$ and $(\|\bv{A}_{-1}\|_{1+\alpha_i})^{1+\alpha_i}$, where $\bv{A}_{-1}$
 is matrix $\bv{A}$ with the top singular vector exactly deflated off,
 by running the algorithm presented in this paper. That is, get the Schatten $1$ and Schatten $(1+\alpha_i)$
 norms for the matrix $\bv{A}$ with the top singular vector deflated off.
 We can also approximate  $\|\bv{A}\|_1$ up to  $(1+\widetilde{\epsilon})$
 relative error using our  algorithm. 
 Then, the two quantities above ($1-\hat{\sigma}_1$ and $\sum_{j > 1} \hat{\sigma}_j^{1+ \alpha_i}$)
 are exactly $\|\bv{A}_{-1}\|_1/\|\bv{A}\|_1$ and $(\|\bv{A}_{-1}\|_{1+\alpha_i})^{1+\alpha_i} /
 \|\bv{A}\|_1$, respectively.
 Since we have relative $(1+\widetilde{\epsilon})$-approximations 
 to the numerators and denominators of both these quantities, we obtain then up to
 $(1+\widetilde{\epsilon})$-relative error  the quantities $1-\hat{\sigma}_1$ and 
 $\sum_{j > 1} \hat{\sigma}_j^{1+ \alpha_i}$,
 as needed to achieve a multiplicative approximation to the SVD entropy.
 Note that we do not need to compute $\|\bv{A}\|_1$ exactly here even when the matrix has very 
 low SVD entropy.
A similar argument can be seen in section 6 of~\cite{harvey2008sketching}.
\section*{Acknowledgements} We thank Vladimir Braverman for pointing out an error in our original proof of Theorem \ref{thm:histogram}, which has been corrected.

\bibliographystyle{alpha}
\bibliography{normEstimation}
\clearpage

\appendix
\section{Linear System Solvers}\label{sec:solverAppendix}
In this section we give runtimes for solving ridge regression using both traditional iterative methods and stochastic gradient descent equipped with deflation-based preconditioning.
We start with a few standard notions from convex optimization, which are necessary for our stochastic method bounds.
\begin{definition}[Strong convexity]\label{strongConvexity} A function $f: \mathbb{R}^d \rightarrow \mathbb{R}$ is $\mu$-strongly convex if, for all $\bv{x},\bv{y} \in \mathbb{R}^d$,
\begin{align*}
f(\bv{x})-f(\bv{y}) \ge \grad f(\bv{y})^T (\bv{x}-\bv{y}) + \frac{\mu}{2} \norm{\bv{x}-\bv{y}}_2^2.
\end{align*}
Equivalently, if $\grad^2 f \succeq \mu \bv{I}$.
\end{definition}

\begin{definition}[Smoothness]\label{smoothness} A function $f: \mathbb{R}^d \rightarrow \mathbb{R}$ is $\beta$-smooth if for all $\bv{x},\bv{y} \in \mathbb{R}^d$,
\begin{align*}
f(\bv{x})-f(\bv{y}) \le \grad f(\bv{y})^T (\bv{x}-\bv{y}) + \frac{\beta}{2} \norm{\bv{x}-\bv{y}}_2^2.
\end{align*}
Equivalently, if $\grad^2 f \preceq \beta \bv{I}$.
\end{definition}

The rate of convergence that iterative methods achieve on $f$ typically depends on the ratio $\beta/\mu$, which corresponds to the condition number of a linear system. We next show how ridge regression can be recast as minimizing a convex function $f$, and show that our error from the true ridge solution is proportional to error in minimizing $f$.

\begin{fact}[Ridge Regression as Convex Optimization]\label{ridgeConvex}
For any $\bv A \in \mathbb{R}^{n \times d}$, $\bv{b} \in \mathbb{R}^d$, and $\lambda > 0$, let $\bv{M}_\lambda \eqdef \bv{A}^T\bv{A} + \lambda \bv{I}$ and $\bv{x}^* \eqdef \bv{M}_\lambda^{-1} \bv{b}$. $\bv{x}^*$ is the minimizer of the convex function:
\begin{align}\label{systemFunction}
f(\bv{x}) = \frac{1}{2} \bv{x}^T(\bv{A}^T\bv{A} + \lambda \bv{I})\bv{x} - \bv{b}^T \bv{x},
\end{align}
which has gradient $\grad f(\bv{x}) = (\bv{A}^T\bv{A} + \lambda \bv{I})\bv{x} - \bv{b}$. $f$ is $\lambda$-strongly convex and $(\sigma_1(\bv{A})^2 + \lambda)$-smooth.
\end{fact}
\begin{proof}
To check that $\bv{x}^*$ is the minimizer, notice that:
\begin{align*}
\grad f(\bv{x}^*) = (\bv{A}^T\bv{A} +\lambda \bv{I})(\bv{A}^T\bv{A} +\lambda \bv{I})^{-1} \bv{b} - \bv{b} = \bv{0}.
\end{align*}
Since the function is quadratic, $\bv{x}^*$ is the unique minimizer. $\grad^2 f = \bv{M}_\lambda$, so by Definitions \ref{strongConvexity} and \ref{smoothness}, $f$ is $\lambda$-strongly convex and $(\sigma_1(\bv{A})^2 + \lambda)$-smooth.
\end{proof}

\begin{fact}[Function Error Equals Norm Error]\label{norm2FunctionFact}
For any $\bv{x} \in \mathbb{R}^d$, letting $\bv{M}_\lambda$, $\bv{x}^*$ and $f$ be defined as in Fact \ref{ridgeConvex}, 
\begin{align}\label{norm2Function}
\norm{\bv{x} - \bv{x}^*}_\bv{M_\lambda}^2 = 2[f(\bv{x})-f(\bv{x}^*)].
\end{align}
\end{fact}
\begin{proof}
Since $\bv{x}^* = \bv{M}_\lambda^{-1}\bv{b}$: 
\begin{align*}
\norm{\bv{x}-\bv{x}^*}_{\bv{M}_\lambda}^2 &\eqdef (\bv{x}-\bv{x}^*)^T \bv{M}_\lambda (\bv{x}-\bv{x}^*)\nonumber\\ 
&= \bv{x}^T\bv{M}_\lambda \bv{x} - 2\bv{x}^T \bv{M}_\lambda \bv{x}^* + \bv{x}^{*T} \bv{M}_\lambda \bv{x}^*\nonumber\\
&= \bv{x}^T\bv{M}_\lambda\bv{x} - 2\bv{x}^T \bv{b} + \bv{b}^{T} \bv{M}_\lambda^{-1} \bv{b}\nonumber\\
&= 2 \left [ f(\bv{x}) - f(\bv{x}^*) \right ].
\end{align*}
\end{proof}
We will focus on making multiplicative progress in $\left [ f(\bv{x}) - f(\bv{x}^*) \right ]$ which will lead to multiplicative progress in $\norm{\bv{x}-\bv{x}^*}_\bv{M}^2$ and a close approximation in $\log(1/\epsilon)$ iterations.

\subsection{Unaccelerated SVRG}
We first prove an unaccelerated and unpreconditioned runtime for the Stochastic Variance Reduced Gradient (SVRG) algorithm, introduced in \cite{johnson2013accelerating}. 
\begin{theorem}[Standard SVRG Runtime]\label{svrgThm} For any $\bv A \in \mathbb{R}^{n \times d}$, $\bv{b} \in \mathbb{R}^d$, and $\lambda > 0$, let $\bv{M}_\lambda \eqdef \bv{A}^T\bv{A} + \lambda \bv{I}$ and $\bv{x}^* \eqdef \bv{M}_\lambda^{-1} \bv{y}$.
There is an algorithm that returns $\bv{x}$ with: 
$\E \norm{\bv{x}-\bv{x}^*}_{\bv{M}_\lambda} \le \epsilon \norm{\bv{x}^*}_{\bv{M}_\lambda}$
in $$O\left (\left [\nnz(\bv{A})+\frac{d_s(\bv{A})\norm{\bv{A}}_F^2}{\lambda} \right ]\cdot \log(1/\epsilon) \right )$$ time.
\end{theorem}

SVRG proceeds in epochs. In each epoch we will make one full gradient computation -- $\grad f(\bv{x}_0) = \bv{M}_\lambda \bv{x}_0 - \bv{b}$ where $\bv{x}_0$ is the value of our iterate at the beginning of the epoch. We will then make a number of stochastic gradient updates each `variance reduced' using $\grad f(\bv{x}_0)$ and show that we make constant factor progress on our function value in expectation. Stringing together $\log(1/\epsilon)$ epochs yields Theorem \ref{svrgThm}.
We write our function $f(\bv{x})$ given by \eqref{systemFunction} as a sum:
\begin{align}
f(\bv{x}) &= \frac{1}{2} \bv{x}^T(\bv{A}^T\bv{A} + \lambda \bv{I})\bv{x} - \bv{b}^T \bv{x}\nonumber\\
 &= \sum_{i=1}^n \frac{1}{2} \bv{x}^T \left (\bv{a}_i \bv{a}_i^T + \frac{\lambda \norm{\bv{a}_i}_2^2}{\norm{\bv{A}}_F^2} \bv{I} \right )\bv{x} - \frac{1}{n}\bv{b}^T \bv{x}.\label{systemFunctionSum}
\end{align}
where $\bv{a}_i \in \mathbb{R}^d$ is the $i^{th}$ row of $\bv{A}$ and $\bv{a}_i \bv{a}_i^T \in \mathbb{R}^{d \times d}$ is the rank-$1$ matrix which is its contribution to $\bv{A}^T\bv{A}$. We start with a well known lemma showing that it is possible to make constant progress on the value of $f(\bv{x})$ in each epoch of SVRG:

\begin{lemma}[SVRG Epoch]\label{lem:epoch}
Consider a set of convex functions $\{\psi_1,\psi_2,...,\psi_n\}$ mapping $\mathbb{R}^d \rightarrow \mathbb{R}$. Let $f(\bv{x}) = \sum_{i=1}^n \psi_i(\bv{x})$ and $\bv{x}^*= \argmin_{\bv{x}\in\mathbb{R}^d} f(\bv{x})$. Suppose we have a probability distribution $p$ on $[1,2,...,n]$ and that starting from some initial point $\bv{x}_0 \in \mathbb{R}^d$ in each iteration we select $i \in [1,...,n]$ with probability $p_i$ and set:
\begin{align}\label{sgdStep}
\bv{x}_{k+1} := \bv{x}_k - \frac{\eta}{p_i} \left (\grad \psi_i(\bv{x}_k) - \grad \psi_i(\bv{x}_0)\right ) + \eta \grad f(\bv{x}_0)
\end{align}
for some step size $\eta$. If $f$ is $\mu$-strongly convex and if for all $\bv{x}\in \mathbb{R}^d$ we have 
\begin{align}\label{smoothnessCondition}
\sum_{i=1}^n \frac{1}{p_i} \norm{\grad \psi_i(\bv{x})-\grad \psi_i(\bv{x}^*)}_2^2 \le 2\bar S \left [f(\bv{x})-f(\bv{x}^*) \right]
\end{align}
for some variance parameter $\bar S$ then for all $m \ge 1$ we have:
\begin{align*}
\E \left [\frac{1}{m} \sum_{k=1}^m f(\bv{x}_k)-f(\bv{x}^*) \right ] \le \frac{1}{1-2\eta \bar S} \left [\frac{1}{\mu \eta m} + 2 \eta \bar S \right ] \cdot \left [f(\bv{x}_0) - f(\bv{x}^*) \right ]
\end{align*}
Consequently, if we pick $\eta$ to be a small multiple of $1/\bar S$, then for $m = O(\bar S/\mu)$ we will decrease the error by a constant multiplicative factor in expectation.
\end{lemma}
\begin{proof}
See for example Theorem 9 of \cite{newVersion}.
\end{proof}

To apply the above Lemma we need the variance bound:
\begin{lemma}[SVRG Variance Bound]\label{varianceBound}
If $\psi_i(\bv{x}) = \frac{1}{2}\bv{x}^T\left (\bv{a}_i \bv{a}_i^T + \frac{\lambda\norm{\bv{a}_i}_2^2}{\norm{\bv{A}}_F^2} \bv{I} \right )\bv{x} - \frac{1}{n}\bv{b}^T \bv{x}$ and $p_i = \frac{\norm{\bv{a}_i}_2^2}{\norm{\bv{A}}_F^2}$ then letting $\bar S = \norm{\bv{A}}_F^2 + 2 \lambda$ we have:
\begin{align*}
\sum_{i=1}^n \frac{1}{p_i} \norm{\grad \psi_i(\bv{x})-\grad \psi_i(\bv{x}^*)}_2^2 \le 2\bar S \left [f(\bv{x})-f(\bv{x}^*) \right]
\end{align*}
\end{lemma}
\begin{proof}
We have
$\grad \psi_i(\bv{x}) = \left (\bv{a}_i \bv{a}_i^T + \frac{\lambda\norm{\bv{a}_i}_2^2}{\norm{\bv{A}}_F^2}\bv{I} \right )\bv{x} - \frac{1}{n}\bv{b}.$
Write $\bv{x}-\bv{x}^* = \bv{y}$ for simplicity of notation.
\begin{align*}
\sum_{i=1}^n& \frac{1}{p_i} \norm{\grad \psi_i(\bv{x})-\grad \psi_i(\bv{x}^*)}_2^2\\
 &= \sum_{i=1}^n \frac{\norm{\bv{A}}_F^2}{\norm{\bv{a}_i}_2^2} \cdot \norm{\left(\bv{a}_i \bv{a}_i^T + \frac{\lambda\norm{\bv{a}_i}_2^2}{\norm{\bv{A}}_F^2}\bv{I} \right )\bv{y}}_2^2\\
&= \norm{\bv{A}}_F^2 \sum_{i=1}^n \left ( \frac{1}{\norm{\bv{a}_i}_2^2} \cdot \bv{y}^T\left((\bv{a}_i \bv{a}_i^T)^2 + 2\frac{\lambda\norm{\bv{a}_i}_2^2}{\norm{\bv{A}}_F^2}\bv{a}_i\bv{a}^T + \left (\frac{\lambda\norm{\bv{a}_i}_2^2}{\norm{\bv{A}}_F^2}\right )^2\bv{I} \right )\bv{y} \right )\\
&=\norm{\bv{A}}_F^2 \cdot \left [ \sum_{i=1}^n \left (\frac{\bv{y}^T(\bv{a}_i\bv{a}_i^T)^2\bv{y}}{\norm{\bv{a}_i}_2^2}  \right ) + 2\sum_{i=1}^n \frac{\lambda \bv{y}^T\bv{a}_i\bv{a}_i^T\bv{y}}{\norm{\bv{A}}_F^2} +\sum_{i=1}^n \left (\frac{\norm{\bv{y}}_2^2 \lambda^2 \norm{\bv{a}_i}_2^2}{\norm{\bv{A}}_F^4} \right )\right ]\\
&=\norm{\bv{A}}_F^2 \cdot \left [ \bv{y}^T\bv{A}^T\bv{A}\bv{y} + \frac{2\lambda \bv{y}^T\bv{A}\bv{A}^T\bv{y}}{\norm{\bv{A}}_F^2} + \frac{\norm{\bv{y}}_2^2 \lambda^2}{\norm{\bv{A}}_F^2} \right ]\\
&\le \norm{\bv{A}}_F^2 \left [\norm{\bv{y}}_{\bv{M}_\lambda}^2 + \frac{2\lambda\norm{\bv{y}}_{\bv{M}_\lambda}^2}{\norm{\bv{A}}_F^2} \right ]\\
&\le 2(\norm{\bv{A}}_F^2 + 2\lambda) [f(\bv{x})-f(\bv{x}^*)]
\end{align*}
where the last step uses Fact \ref{norm2FunctionFact}.
\end{proof}

We can now plug this variance bound into Lemma \ref{lem:epoch} to prove Theorem \ref{svrgThm}
\begin{proof}[Proof of Theorem \ref{svrgThm}]
Using Lemma \ref{varianceBound} we can instantiate Lemma \ref{lem:epoch} with $\bar S = \norm{\bv{A}}_F^2 + 2\lambda $. If we set 
\begin{align*}
m = O\left (\frac{\bar S}{\mu}\right) =  O\left (\frac{\norm{\bv{A}}_F^2 + \lambda}{\lambda}\right) = O\left (\frac{\norm{\bv{A}}_F^2}{\lambda}\right)
\end{align*}
and $\eta = O\left (\frac{1}{\norm{\bv{A}}_F^2 + \lambda}\right)$ (which we can compute explicitly)
then after an $m$ step SVRG epoch:
\begin{align*}
\E \left [\frac{1}{m} \sum_{k=1}^m f(\bv{x}_k)-f(\bv{x}^*) \right ] \le \frac{1}{2} \left [f(\bv{x}_0) - f(\bv{x}^*) \right ].
\end{align*}
If we choose $k$ uniformly from $[1,...,m]$ this gives us $\E\left [ f(\bv{x}_k) - f(\bv{x}^*)\right] \le \frac{1}{2}\left [f(\bv{x}_0) - f(\bv{x}^*) \right ]$.
So stringing together $\log(1/\epsilon)$ of these epochs and letting $\bv{x}_0 = \bv{0}$ in the first epoch gives the theorem. Each epoch requires $m$ stochastic gradient steps given by \eqref{sgdStep}, taking $O(md_s) = O\left (\frac{\norm{d_s\bv{A}}_F^2 }{\lambda}\right )$ time plus $\nnz(\bv{A})$ time to compute $\grad f(\bv{x}_0)$, giving us the stated runtime.

Note that, naively, to produce $\bv{x}_{k+1}$ we need $O(d)$ time, not $O(d_s)$ time since our gradient step includes adding multiples of $\bv{b}$ and $\eta \grad f(\bv{x}_0)$ both of which might be dense vectors in $\R^d$. However, in each epoch, we will just keep track of the coefficients of these two components in our iterate, allowing us to still compute \eqref{sgdStep} in $O(d_s)$ time.
\end{proof}

\subsection{Unaccelerated SVRG with Deflation-Based Preconditioning}

If we assume that $d_s = O\left (\frac{\nnz(\bv{A})}{n} \right)$ (i.e. our rows are uniformly sparse) we see that the
runtime of Theorem \ref{svrgThm} is dominated by $\frac{d_s(\bv{A}) \norm{\bv{A}}_F^2}{\lambda} = \nnz(\bv{A}) \cdot \frac{\sum_{i=1}^d \sigma_i^2(\bv{A})}{n \lambda}$. $\bar \kappa \eqdef \frac{\sum_{i=1}^d \sigma_i^2(\bv{A})}{n \lambda}$ can be seen as an \emph{average condition number}, which is always smaller than the condition number $\kappa \eqdef \frac{\sigma_1^2}{\lambda}$. This average condition number dependence means that SVRG can significantly outperform traditional iterative solvers that require a number of iterations depending on $\kappa$.  However, this advantage is limited if $\norm{\bv{A}}_F^2$ is very concentrated in a few large singular values, and hence $\bar \kappa \approx \kappa$. We can perform better in such situations by deflating off these large singular values and preconditioning with our deflated matrix, significantly `flattening' the spectrum of $\bv{A}$. This method was used in \cite{gonen2016solving}, and we follow their approach closely, giving our own proof for completeness and so that we can express runtimes in terms of the necessary parameters for our downstream results. In particular, we show that the preconditioning methods of \cite{gonen2016solving} can be implemented efficiently for sparse systems. 

\begin{theorem}[Preconditioned SVRG Runtime]\label{preCondsvrgThm} For any $\bv A \in \mathbb{R}^{n \times d}$, $\bv{b} \in \mathbb{R}^d$, and $\lambda > 0$, let $\bv{M}_\lambda \eqdef \bv{A}^T\bv{A} + \lambda \bv{I}$. Let $ \bar \kappa \eqdef \frac{k \sigma_k^2(\bv{A}) + \sum_{i=k+1}^d \sigma_i^2(\bv{A})}{d\lambda} $ where $k \in [0,...,d]$ is an input parameter.
There is an algorithm that uses $O( \nnz(\bv{A}) k\log d + dk^{\omega-1})$ precomputation time for sparse matrices or $O (nd^{\omega(1,1,\log_d k) - 1} \log d)$  for dense matrices after which, given any input $\bv{y} \in \R^d$, letting $\bv{x}^* \eqdef \bv{M}_\lambda^{-1} \bv{y}$ the algorithm 
returns $\bv{x}$ with
$\E \norm{\bv{x}-\bv{x}^*}_{\bv{M}_\lambda} \le \epsilon \norm{\bv{x}^*}_{\bv{M}_\lambda}$
in $$O \left ( \nnz(\bv{A}) \log (1/\epsilon) +  \log(1/\epsilon) \left [d\cdot d_s(\bv{A}) + d k \right ]\bar \kappa \right )$$ time for sparse $\bv{A}$ or 
$$O \left (\log(1/\epsilon) (nd + d^2 \lceil \bar \kappa \rceil )\right )$$ time for dense $\bv{A}$.
\end{theorem}

Our ideal algorithm would compute the top $k$ singular vectors of $\bv{A}$ and deflate these off our matrix to flatten the spectrum. However, for efficiency we will instead compute approximate singular vectors, using an iterative method, like simultaneous iteration or a block Krylov iteration. These algorithms give the following guarantee \cite{musco2015randomized}:
\begin{lemma}\label{blockLanczos}
There is an algorithm that, with high probability in $O(\nnz(\bv{A})k \log d + dk^{\omega-1})$ time returns $\bv{ Z} \in \R^{d \times k}$ such that $\norm{\bv{A}-\bv{ZZ}^T \bv{A}}_2^2 \le 2\sigma_{k+1}(\bv{A})^2$ and, for all $i \le k$, letting $\tilde \sigma_i^2 \eqdef \bv{z}_i^T \bv{A}^T\bv{A}\bv{z}_i$,  $\left |\tilde \sigma_i - \sigma_i^2(\bv{A}) \right | \le 2 \sigma_{k+1}^2(\bv{A})$. For dense inputs the runtime can be sped up to $O(n d^{\omega(1,1,\log_d k)-1} \log d)$ by applying fast matrix multiplication at each iteration.\footnote{Each iteration requires $n d^{\omega(1,1,\log_d k)-1} \ge d^{\omega(1,1,\log_d k)}$ time since we assume $n \ge d$, which dominates the $O(dk^{\omega-1})$ time required to orthogonalize the approximate singular directions.}
\end{lemma}

\cite{gonen2016solving} shows that we can build a good preconditioner from such a $\bv{Z}$. Specifically:
\begin{lemma}[Theorem 5 of \cite{gonen2016solving}]\label{traceBound}
For any $\bv A \in \mathbb{R}^{n \times d}$, $\bv{b} \in \mathbb{R}^d$, and $\lambda > 0$,
given $\bv{Z} \in \R^{d\times k}$ satisfying the guarantees of Lemma \ref{blockLanczos}, let 
$$\bv{P}^{-1/2} = \bv{Z} \bv{\tilde \Sigma}^{-1/2} \bv{Z}^T + \frac{(\bv{I}-\bv{ZZ}^T)}{\sqrt{\tilde \sigma_k^2 + \lambda}}\text{ where }\bv{\tilde \Sigma}^{-1/2}_{i,i} = \frac{1}{\sqrt{\tilde \sigma_i^2 + \lambda}}.$$
Then we have:
$$\frac{\tr(\bv{P}^{-1/2}(\bv{A}^T \bv{A} + \lambda \bv{I}) \bv{P}^{-1/2})}{\lambda_d(\bv{P}^{-1/2} (\bv{A}^T\bv{A} + \lambda \bv{I}) \bv{P}^{-1/2})} = O \left ( \frac{k \sigma_k^2(\bv{A}) + \sum_{i=k+1}^d \sigma_i^2(\bv{A})}{\lambda} + d \right )
%= O \left (\frac{\sigma_k(\bv{A})\norm{\bv{A}}_1}{\lambda} + d\right )
.$$
\end{lemma}
Intuitively, after applying the preconditioner, all top singular values are capped at $\sigma_k^2(\bv{A})$, giving a much flatter spectrum and better performance when optimizing with SVRG.
To make use of the above bound, we first define a preconditioned ridge regression problem.
\begin{definition}[Preconditioned Ridge Regression]\label{def:precondFunction} For any $\bv A \in \mathbb{R}^{n \times d}$, $\bv{b} \in \mathbb{R}^d$, and $\lambda > 0$, let $\bv{M}_\lambda \eqdef \bv{A}^T\bv{A} + \lambda \bv{I}$ and $\bv{x}^* \eqdef \bv{M}_\lambda^{-1} \bv{b}$. Letting $\bv{P}^{-1/2}$ be as described in Lemma \ref{traceBound}, let $\bv{\hat M}_\lambda = \bv{P}^{-1/2}\bv{M}_\lambda \bv{P}^{-1/2}$ and
define the preconditioned regression objective function by:
\begin{align*}
\hat f(\bv{x}) = \frac{1}{2} \bv{x}^T \bv{\hat M}_\lambda \bv{x} - \bv{b}^T \bv{P}^{-1/2} \bv{x}.
\end{align*}
$\hat f$ is minimized at $\bv{\hat x}^* =\bv{\hat M}_\lambda^{-1} \bv{P}^{-1/2} \bv{b} = \bv{P}^{1/2} (\bv{A}^T \bv{A} + \lambda \bv{I})^{-1} \bv{b}$.
\end{definition}

\begin{fact}[Preconditioned Solution]\label{preconditionedSolution}
For any $\bv{x} \in \R^d$, if $\norm{\bv{x} - \bv{\hat x}^*}_{\bv{\hat M}_\lambda} \le \epsilon \norm{\bv{\hat x}^*}_{\bv{\hat M}_\lambda}$ then $\norm{\bv{P}^{-1/2}\bv{x} - \bv{x}^*}_{\bv{M}_\lambda} \le \epsilon \norm{\bv{ x}^*}_{\bv{ M}_\lambda}$ where $\bv{x}^*$ is the minimizer of $f(\bv{x})$.
\end{fact}
\begin{proof}
We have $\bv{\hat x}^* = \bv{P}^{1/2} \bv{x}^*$ and so using the fact that $\norm{\bv{x} - \bv{\hat x}^*}_{\bv{\hat M}_\lambda} \le \epsilon \norm{\bv{\hat x}^*}_{\bv{\hat M}_\lambda}$ can write:
\begin{align*}
\norm{\bv{x} - \bv{P}^{1/2}\bv{x}^*}_{\bv{\hat M}_\lambda} &\le \epsilon \norm{\bv{P}^{1/2}\bv{ x}^*}_{\bv{\hat M}_\lambda}\\
\norm{\bv{P}^{-1/2}\bv{x} - \bv{x}^*}_{\bv{ M}_\lambda} &\le \epsilon \norm{\bv{ x}^*}_{\bv{ M}_\lambda}.
\end{align*}
\end{proof}

By Fact \ref{preconditionedSolution}, we can find a near optimal solution to our preconditioned system and by multiplying by $\bv{P}^{-1/2}$ obtain a near optimal solution to the original ridge regression problem. We now show that our preconditioned problem can be solved quickly. We first give a variance bound as we did in Lemma \ref{varianceBound} in the non-preconditioned case.
%We first note that it is possible to efficiently apply our preconditioner to any vector $\bv{x} \in \mathbb{R}^{d}$.
%\begin{lemma}
%\end{lemma}

\begin{lemma}[Preconditioned Variance Bound]\label{precondvarianceBound}
Let $\bv{B}^T = [\bv{P}^{-1/2}\bv{A}^T, \sqrt{\lambda} \bv{P}^{-1/2}]$ so that $\bv{B}^T\bv{B} = \bv{P}^{-1/2} (\bv{A}^T\bv{A} + \lambda \bv{I}) \bv{P}^{-1/2} = \bv{\hat M}_\lambda$.
For $i \in 1,...,(n+d)$, let 
$$\hat \psi_i(\bv{x}) = \frac{1}{2}\bv{x}^T \left (\bv{b}_i \bv{b}_i^T\right )\bv{x} - \frac{1}{n+d}\bv{b}^T \bv{P}^{-1/2}\bv{x}$$
so that $\sum_{i=1}^{n+d} \hat \psi_i(\bv{x}) = \hat f(\bv{x})$. Set $p_i = \frac{\norm{\bv{b}_i}_2^2}{\norm{\bv{B}}_F^2}$ then letting $\bar S= \tr(\bv{P}^{-1/2} (\bv{A}^T\bv{A} + \lambda \bv{I} )\bv{P}^{-1/2}$:
\begin{align*}
\sum_{i=1}^n \frac{1}{p_i} \norm{\grad \hat \psi_i(\bv{x})-\grad \hat \psi_i(\bv{x}^*)}_2^2 \le 2\bar S \left [f(\bv{x})-f(\bv{x}^*) \right].
\end{align*}
\end{lemma}
\begin{proof}
Write $\bv{x}-\bv{x}^* = \bv{y}$ for simplicity of notation. Then we have:
\begin{align*}
\sum_{i=1}^{n+d} \frac{1}{p_i} \norm{\grad \hat \psi_i(\bv{x})-\grad \hat \psi_i(\bv{x}^*)}_2^2
= \norm{\bv{B}}_F^2 \sum_{i=1}^{n+d} \frac{\norm{\bv{b}_i\bv{b}_i^T \bv{y}}_2^2}{\norm{\bv{b}_i}_2^2} &= \norm{\bv{B}}_F^2\sum_{i=1}^{n+d} \frac{\bv{y}^T \bv{b}_i(\bv{b}_i^T \bv{b}_i)\bv{b}_i^T \bv{y}}{\norm{\bv{b}_i}_2^2}\\
&= \norm{\bv{B}}_F^2 \cdot \bv{y}^T \bv{B}^T\bv{B} \bv{y}\\
&= 2\norm{\bv{B}}_F^2 \cdot  [\hat f(\bv{x}) - \hat f(\bv{\hat x}^*)]
\end{align*}
where the last step uses Fact \ref{norm2FunctionFact}. Finally we write:
\begin{align*}
\bar S = \norm{\bv{B}}_F^2 = \tr(\bv{B}^T\bv{B}) = \tr(\bv{P}^{-1/2} (\bv{A}^T\bv{A} + \lambda \bv{I} )\bv{P}^{-1/2}).
\end{align*}
\end{proof}
We can now combine this variance bound with Lemmas \ref{lem:epoch} and \ref{traceBound} to prove our preconditioned SVRG runtime.

\begin{proof}[Proof of Theorem \ref{preCondsvrgThm}]
Using Lemma \ref{precondvarianceBound} we can instantiate Lemma \ref{lem:epoch} with $\bar S = \norm{\bv{B}}_F^2 = \tr(\bv{P}^{-1/2} (\bv{A}^T\bv{A} + \lambda \bv{I} )\bv{P}^{-1/2})$. Letting $\mu = \lambda_d(\bv{P}^{-1/2} (\bv{A}^T\bv{A} + \lambda \bv{I}) \bv{P}^{-1/2})$ be a lower bound on the strong convexity of $\hat f$, by Lemma \ref{traceBound} we have:
\begin{align*}
m = O \left (\frac{\bar S}{\mu} \right ) &= O \left (\frac{\tr(\bv{P}^{-1/2} (\bv{A}^T\bv{A} + \lambda \bv{I} )\bv{P}^{-1/2})}{\lambda_d(\bv{P}^{-1/2} (\bv{A}^T\bv{A} + \lambda \bv{I}) \bv{P}^{-1/2})} \right )\\
&= O \left ( \frac{k \sigma_k^2(\bv{A}) + \sum_{i=k+1}^d \sigma_i^2(\bv{A})}{\lambda} + d \right ).
\end{align*}
If we set the step size $\eta = O \left (\frac{1}{\bar S} \right ) = O \left (\frac{1}{\norm{\bv{B}}_F^2} \right )$ which we will compute explicitly, then after an $m$ step SVRG epoch running on the functions $\hat \psi_i(\bv{x}) = \frac{1}{2}\bv{x}^T \left (\bv{b}_i \bv{b}_i^T\right )\bv{x} - \frac{1}{n+d}\bv{b}^T \bv{P}^{-1/2}\bv{x}$ each selected with probability $p_i = \frac{\norm{\bv{b}_i}_2^2}{\norm{\bv{B}}_F^2}$, we make constant progress in expectation on $\hat f(\bv{x})$. After $\log (1/\epsilon)$ iterations we have $\bv{x}$ with $\E \norm{\bv{x}-\bv{\hat x}^*}_{\bv{\hat M}_\lambda} \le \epsilon \norm{\bv{\hat x}^*}_{\bv{\hat M}_\lambda}$ and hence, $\E \norm{\bv{P}^{-1/2} \bv{x}-\bv{x}^*}_{\bv{M}_\lambda} \le \epsilon \norm{\bv{x}^*}_{\bv{M}_\lambda}$ by Fact \ref{preconditionedSolution}. It remains to bound the cost of each SVRG epoch applied to $\hat f(\bv{x})$.

For dense inputs, the argument is relatively simple. We can compute the approximate top singular vector space $\bv{Z} \in \R^{d \times k}$ in $O(nd^{\omega(1,1,\log_d k) -1} \log d)$ time by Lemma \ref{blockLanczos}. Using the factored structure of $\bv{P}^{-1/2}$, we can explicitly form $\bv{A}\bv{P}^{-1/2}$ and $\bv{P}^{-1/2}$ also in $(O(nd^{\omega(1,1,\log_d k) -1})$ time. This allows us to compute a full gradient in $O(nd)$ time (which we do once per epoch) and perform each stochastic gradient step in $O(d)$ time (which we do $m$ times per epoch). This gives final runtime 
\begin{align*}
O \left ( (nd + md) \log(1/\epsilon) \right ) = O \left (\log(1/\epsilon) (nd + d^2 \lceil \bar \kappa \rceil) \right )
\end{align*}
where $\bar \kappa \eqdef \frac{k \sigma_k^2(\bv{A}) + \sum_{i=k+1}^d \sigma_i^2(\bv{A})}{d\lambda}$.

For sparse $\bv{A}$ we have to be much more careful to fully exploit sparsity when preconditioning.
Computing $\bv{Z}$ takes time $O(\nnz(\bv{A}) k \log d + dk^{\omega-1})$ by Lemma \ref{blockLanczos}. We will not explicitly form $\bv{P}^{-1/2}$ but will show how to apply it when needed.  Recall that $\bv{P}^{-1/2} = \bv{Z} \bv{\tilde \Sigma}^{-1/2} \bv{Z}^T + \delta (\bv{I}-\bv{ZZ}^T)$ where we denote $\delta \eqdef 1/(\tilde \sigma_i^2 + \lambda)$.
First, in order to determine our step size $\eta$ we must compute $\norm{\bv{B}}_F^2$.
\begin{align*}
\norm{\bv{B}}_F^2 &= \norm{\bv{A}\bv{P}^{-1/2}}_F^2 + \lambda \norm{\bv{P}^{-1/2}}_F^2\\
&= \norm{\bv{AZ}\bv{\tilde \Sigma}^{-1/2} \bv{Z}^T}_F^2 + \delta \norm{\bv{A}(\bv{I}-\bv{ZZ}^T)}_F^2 + \lambda \norm{\bv{Z}\bv{\tilde \Sigma}^{-1/2} \bv{Z}^T}_F^2 + \lambda \delta \norm{\bv{I}-\bv{ZZ}^T}_F^2.
\end{align*}
Applying Pythagorean theorem and submultiplicativity we have:
\begin{align*}
\norm{\bv{B}}_F^2 &=  \norm{\bv{AZ}\bv{\tilde \Sigma}^{-1/2}}_F^2 + \delta \norm{\bv{A}}_F^2 - \delta \norm{\bv{A}\bv{Z}}_F^2 +  \lambda\norm{\bv{Z}\bv{\tilde \Sigma}^{-1/2}}_F^2 + \lambda\delta(d-k)
\end{align*}
which can be computed in $O(\nnz(\bv{A}) k)$ time with just a single multiplication of $\bv{A}$ by $\bv{Z}$. 

We can similarly compute our sampling probabilities $p_i = \frac{\norm{\bv{b}_i}_2^2}{\bv{B}_F^2}$ quickly by noting that for $i = 1,...,n$, $\norm{\bv{b}_i}_2^2 = \norm{\bv{a}_i^T\bv{Z}\bv{\tilde \Sigma}^{-1/2}}_2^2 + \delta \norm{\bv{a}_i}_2^2 - \delta \norm{\bv{a}_i^T \bv{Z}}_2^2$ and for $i = n+1,...,n+d$, $\norm{\bv{b}_i}_2^2 = \lambda \left (\norm{\bv{z}_i^T \bv{\tilde \Sigma}^{-1/2}}_2^2 + \delta (1-\norm{\bv{z}_i}_2^2) \right )$. This again requires just $O(\nnz(\bv{A})k)$ time.

 In each SVRG epoch we must compute one full gradient of the form:
\begin{align*}
\grad \hat f(\bv{x}_0) = \bv{B}^T\bv{B} \bv{x}_0 - \bv{P}^{-1/2}\bv{b} = \bv{P}^{-1/2}  \bv{A}^T\bv{A} \bv{P}^{-1/2}  \bv{x}_0 + \lambda \bv{P}^{-1} \bv{x}_0 - \bv{P}^{-1/2}\bv{b}
\end{align*}
which takes $O(\nnz(\bv{A}) + dk)$ time as $\bv{P}^{-1/2}$ can be applied to a vector in $O(dk)$ time. Naively, we then must make $m$ stochastic gradient steps, each requiring $O(dk)$ time. However we can do better. Instead of storing our iterate $\bv{x}_k$ explicitly, we will store it as the sum:
\begin{align*}
\bv{x}_k = \left (\bv{x}_k^{(0)} - \bv{Z}\bv{x}_k^{(1)}\right ) +  \bv{Z}\bv{x}_k^{(2)} +x_k^{(3)} \cdot \grad \hat f(\bv{x}_0).
\end{align*}
Note that $\bv{x}_k^{(0)} \in \R^{d}$, $\bv{x}_k^{(1)},\bv{x}_k^{(2)} \in \R^{k}$ and $x_k^{(3)}$ is a scalar. At the beginning of each epoch we will set $\bv{x}_0^{(0)} = \bv{x}_0$, $\bv{x}_0^{(1)} =  \bv{Z}^T \bv{x}_0$, $\bv{x}_0^{(2)}  = \bv{Z}^T \bv{x}_0$ and $x_0^{(3)} = 0$. We can compute $\bv{Z}^T \bv{x}_0$ in $O(dk)$ time. We will maintain the invariant that $\left (\bv{x}_k^{(0)} - \bv{Z}\bv{x}_k^{(1)}\right)$ is perpendicular to the span of $\bv{Z}$ while $\bv{x}_k^{(2)}$ will give the component of $\bv{x}_k$ within this span.

For ease of notation let $\bv{a}_{n+i}$ denote $\sqrt{\lambda}\bv{e}_i$ where $\bv{e}_i \in \R^{d}$ the $i^{th}$ standard  basis vector. In this way we have $\bv{b}_i = \bv{P}^{-1/2} \bv{a}_i$ for all $i \in 1,...,n+d$.
Each stochastic gradient step is of the form:
\begin{align*}
\bv{x}_{k+1} &= \bv{x}_k - \frac{\eta}{p_i} \left (\grad \hat \psi_i(\bv{x}_k) - \grad \hat \psi_i(\bv{x}_0)\right ) + \eta \grad \hat f(\bv{x}_0)\\
&= \bv{x}_k + \frac{\eta}{p_i} \bv{P}^{-1/2}(\bv{a}_i\bv{a}_i^T) \bv{P}^{-1/2} (\bv{x}_{k} - \bv{x}_0) + \eta \grad \hat f(\bv{x}_0)\\
&= \bv{x}_k + \frac{\eta}{p_i} (\bv{P}^{-1/2}\bv{a}_i) \left ((\bv{a}_i^T \bv{P}^{-1/2})\bv{x}_{k} - \bv{a}_i^T (\bv{P}^{-1/2}\bv{x}_0) \right) + \eta \grad \hat f(\bv{x}_0).
\end{align*}

We can precompute $\bv{P}^{-1/2}\bv{x}_0$ in $O(dk)$ time per epoch and then can compute the dot product $\bv{a}_i^T (\bv{P}^{-1/2}\bv{x}_0)$ in $\nnz(\bv{a}_i) = O(d_s(\bv{A}))$ time.  Using the fact that $\left ( \bv{x}_k^{(0)} - \bv{Z}\bv{x}_k^{(1)} \right )$ is always perpendicular to the span of $\bv{Z}$, we can write the dot product $(\bv{a}_i^T \bv{P}^{-1/2})\bv{x}_{k}$ as:
\begin{align*}
\bv{a}_i^T &\left (\bv{Z} \bv{\tilde \Sigma}^{-1/2} \bv{Z}^T + \delta (\bv{I}-\bv{ZZ}^T) \right ) \left (\left ( \bv{x}_k^{(0)} - \bv{Z}\bv{x}_k^{(1)} \right ) + \bv{Z}\bv{x}_k^{(2)} + x_k^{(3)} \grad \hat f(\bv{x}_0)\right)\\
&= \delta \bv{a}_i^T \left ( \bv{x}_k^{(0)} - \bv{Z}\bv{x}_k^{(1)} \right ) + \bv{a}_i^T \bv{Z} \bv{\tilde \Sigma}^{-1/2} \bv{x}_k^{(2)} + x_k^{(3)} \bv{a}_i^T \bv{Z} \bv{\tilde \Sigma}^{-1/2} \bv{Z}^T \grad \hat f(\bv{x}_0) + \delta x_k^{(3)} \bv{a}_i^T  (\bv{I}-\bv{ZZ}^T)\grad \hat f(\bv{x}_0).
\end{align*}
We can precompute $\bv{A} \bv{Z}$ in $O(\nnz(\bv{A}) k)$ time. We can also precompute $\bv{Z} \bv{\tilde \Sigma}^{-1/2} \bv{Z}^T \grad \hat f(\bv{x}_0) $ and $(\bv{I}-\bv{ZZ}^T)\grad \hat f(\bv{x}_0)$ in $O(dk)$ time. With these values in hand, computing the above dot product takes just $O(\nnz(\bv{a}_i) + k) = O(d_s(\bv{A}) + k)$ time.

Now, we can write $\bv{P}^{-1/2} \bv{a}_i = \bv{Z} (\bv{\tilde \Sigma} \bv{Z}^T \bv{a}_i) -\delta(\bv{I}-\bv{ZZ}^T) \bv{a}_i$. Since we have precomputed $\bv{A}\bv{Z}$, $\bv{\tilde \Sigma} \bv{Z}^T \bv{a}_i$ can be computed and added to $\bv{x}_k^{(2)}$ with the appropriate weight in $O(k)$ time. We can then compute for the appropriate weight $w$:
\begin{align*}
\left ( \bv{x}_{k+1}^{(0)} - \bv{Z}\bv{x}_{k+1}^{(1)} \right ) &= \left ( \bv{x}_{k}^{(0)} - \bv{Z}\bv{x}_{k}^{(1)} \right ) + w \cdot (\bv{I}-\bv{ZZ}^T) \bv{a}_i\\
&= \left (\bv{x}_{k}^{(0)} + w\bv{a}_i \right ) - \bv{Z} \left (\bv{x}_{k}^{(1)} + w\cdot \bv{Z}^T\bv{a}_i \right )
\end{align*}
which takes time $O(d_s(\bv{A}) + k)$. Finally, we set $x_{k+1}^{(3)} = x_{k}^{(3)} + \eta$.
Overall, our runtime per epoch is $O(\nnz(\bv{A}) + dk + m(d_s(\bv{A}) + k))$ and so our total runtime is:
\begin{align*}
%O \left (\nnz(\bv{A})k \log n + nk^{\omega-1} + \log(1/\epsilon) \left [\nnz(\bv{A}) + dk + m(d_s+k) \right] \right )
O \left ( \nnz(\bv{A}) \log (1/\epsilon) +  \log(1/\epsilon) \left [d\cdot d_s(\bv{A}) + d k \right ]\bar \kappa \right )
\end{align*}
where $\bar \kappa \eqdef \frac{k \sigma_k^2(\bv{A}) + \sum_{i=k+1}^d \sigma_i^2(\bv{A})}{d\lambda}.$
\end{proof}

\subsection{Accelerated and Preconditioned SVRG}

We now combine the deflation-based preconditioning described above with accelerated gradient methods to give our strongest runtime bound for ridge regression using via stochastic solvers.

\begin{theorem}[Accelerated Preconditioned SVRG Runtime]\label{accPreCondSvrg} For any $\bv A \in \mathbb{R}^{n \times d}$, $\bv{b} \in \mathbb{R}^d$, and $\lambda > 0$, let $\bv{M}_\lambda \eqdef \bv{A}^T\bv{A} + \lambda \bv{I}$ and $ \bar \kappa \eqdef \frac{k \sigma_k^2(\bv{A}) + \sum_{i=k+1}^d \sigma_i^2(\bv{A})}{d\lambda}$ where $k \in [0,...,d]$ is an input parameter.
There is an algorithm that uses $O( \nnz(\bv{A}) k\log d + dk^{\omega-1})$ precomputation time for sparse matrices or $O (nd^{\omega(1,1,\log_d k) - 1} \log d)$  for dense matrices after which, given any input $\bv{y} \in \R^d$, letting $\bv{x}^* \eqdef \bv{M}_\lambda^{-1} \bv{y}$ the algorithm 
returns $\bv{x}$ with
$\E \norm{\bv{x}-\bv{x}^*}_{\bv{M}_\lambda} \le \epsilon \norm{\bv{x}^*}_{\bv{M}_\lambda}$
in $$O \left (\nnz(\bv{A}) \log(1/\epsilon) + \log(1/\epsilon) \log(d\bar \kappa) \cdot \sqrt{\nnz(\bv{A}) [d \cdot d_s(\bv{A}) + dk ] \bar \kappa} \right )$$ time for sparse $\bv{A}$ or $$O \left (\log(1/\epsilon) (nd + n^{1/2}d^{3/2} \log(\bar \kappa) \sqrt{\bar \kappa} ) \right )$$ time for dense $\bv{A}$.
\end{theorem}

The above runtime will follow from applying a blackbox technique for accelerating the runtime of convex optimization methods. While there are a number of improvements over this method \cite{allen2016katyusha}, we at most lose logarithmic factors and gain significant simplicity in our proofs.

%Assume $\nnz(\bv A) \le \frac{d_s \norm{\bv{A}}_F^2}{\lambda}$ or else Theorem \ref{svrgThm} already gives us $O(\log(1/\epsilon) \nnz(\bv{A}))$ runtime.
%Otherwise, we will apply the following Lemma black box, although again it would be nice to prove eventually:
\begin{lemma}[Theorem 1.1 of \cite{frostig2015regularizing}]\label{acceleration_primitive}
Let $f(\bv{x})$ be $\mu$-strongly convex and let $\bv{x}^* \eqdef \argmin_{\bv{x}\in\mathbb{R}^d} f(\bv{x})$. For any $\gamma > 0$ and any $\bv{x}_0 \in \mathbb{R}^d$, let $f_{\gamma,\bv{x}_0}(\bv{x}) \eqdef f(\bv{x}) + \frac{\gamma}{2} \norm{\bv{x}-\bv{x}_0}_2^2$. Let $\bv{x}^*_{\gamma,\bv{x}_0}  \eqdef \argmin_{\bv{x}\in\mathbb{R}^d} f_{\gamma,\bv{x}_0} (\bv{x})$.
Suppose that, for all $\bv{x}_0 \in \mathbb{R}^d$, $c > 0$, $\gamma > 2\mu$, we can compute $\bv{x}_c$ such that
\begin{align*}
\E \left [f_{\gamma,\bv{x}_0}(\bv{x}_c) - f_{\gamma,\bv{x}_0}(\bv{x}^*_{\gamma,\bv{x}_0} )\right ] \le \frac{1}{c} \left [f_{\gamma,\bv{x}_0} - f_{\gamma,\bv{x}_0}(\bv{x}^*_{\gamma,\bv{x}_0} ) \right ]
\end{align*}
in time $\mathcal{T}_c$. Then we can compute $x_1$ such that $\E \left [f(\bv{x}_1) - f(\bv{x}^*) \right ] \le \frac{1}{c} \left [f(\bv{x}_0) - f(\bv{x}^*) \right ]$
in time $$O\left(\mathcal{T}_{4\left (\frac{2\gamma + \mu}{\mu} \right )^{3/2}} \sqrt{ \gamma/\mu} \log c \right ).$$
\end{lemma}

\begin{proof}[Proof of Theorem \ref{accPreCondSvrg}]
We focus again on optimizing the preconditioned function $\hat f$ in Definition \ref{def:precondFunction} as by Fact \ref{preconditionedSolution} a near optimal minimizer for this function yields a near optimal solution for our original ridge regression problem. 
We split $\hat f_{\gamma,\bv{x}_0} = \frac{1}{2} \bv{x}^T \bv{\hat M}_\lambda \bv{x} - \bv{b}^T \bv{P}^{-1/2} \bv{x} + \frac{\gamma}{2} \norm{\bv{x}-\bv{x}_0}_2^2$ as $\hat f_{\gamma,\bv{x}_0} = \sum_{i=1}^ {n+d} \hat \psi^i_{\gamma,\bv{x}_0}(\bv{x})$ where 
\begin{align*}
\hat \psi^i_{\gamma,\bv{x}_0}(\bv{x}) = \frac{1}{2} \bv{x}^T (\bv{b}_i \bv{b}_i^T) \bv{x} - \frac{1}{n+d} \bv{b}^T \bv{P}^{-1/2} \bv{x} + \frac{\gamma \norm{\bv{b}_i}_2^2}{2\norm{\bv{B}_F^2}} \norm{\bv{x}-\bv{x}_0}_2^2.
\end{align*}
$\bv{B}$ is as defined in Lemma \ref{precondvarianceBound} with $\bv{B}^T = [\bv{P}^{-1/2}\bv{A}^T, \sqrt{\lambda} \bv{P}^{-1/2}]$.

We have $\grad \hat \psi^i_{\gamma,\bv{x}_0}(\bv{x}) = (\bv{b}_i \bv{b}_i^T)\bv{x} - \frac{1}{n+d}\bv{b}\bv{P}^{-1/2}  + \frac{\gamma \norm{\bv{b}_i}_2^2}{\norm{\bv{B}}_F^2} (\bv{x}-\bv{x}_0)$. Letting $\bv{y} = \bv{x}-\bv{x}_{\gamma,\bv{x}_0}^*$,
we can follow a similar calculation to Lemma \ref{varianceBound} to show:
\begin{align*}
\sum_{i=1}^n& \frac{1}{p_i} \norm{\grad \hat \psi^i_{\gamma,\bv{x}_0}(\bv{x})-\grad \hat \psi^i_{\gamma,\bv{x}_0}(\bv{x}_{\gamma,\bv{x}_0}^*)}_2^2\\
&= \norm{\bv{B}}_F^2 \sum_{i=1}^{n+d} \frac{\norm{(\bv{b}_i \bv{b}_i^T + \frac{\gamma \norm{\bv{b}_i}_2^2}{\norm{\bv{B}}_F^2} \bv{I}) \bv{y}}_2^2}{\norm{\bv{b}_i}_2^2}\\
&\le 2(\norm{\bv{B}}_F^2 + 2\gamma) [f(\bv{x})-f(\bv{x}_{\gamma,\bv{x}_0}^*)].
\end{align*}

Let $\bar S = \norm{\bv{B}}_F^2$ and $\mu = \lambda_d(\bv{P}^{-1/2} (\bv{A}^T\bv{A} + \lambda \bv{I}) \bv{P}^{-1/2})$ be an lower bound on the strong convexity of $\hat f$. 
%In \cite{gonen2016solving} it is shown that $\mu \ge \frac{\lambda}{19(\sigma_k^2(\bv{A}) + \lambda}$. 
Denote $\frac{\gamma}{\mu} = r$. The strong convexity of $\hat f_{\gamma,\bv{x}_0}$ is lower bounded by
 $\mu_\gamma = \lambda_d(\bv{P}^{-1/2} (\bv{A}^T\bv{A} + \lambda \bv{I}) \bv{P}^{-1/2}) + \gamma = \Theta(r \cdot \lambda_d(\bv{P}^{-1/2} (\bv{A}^T\bv{A} + \lambda \bv{I}) \bv{P}^{-1/2}))$ which gives:
\begin{align*}
\frac{\bar S}{\mu_\gamma} = \frac{\tr(\bv{P}^{-1/2}(\bv{A}^T \bv{A} + \lambda \bv{I}) \bv{P}^{-1/2}) + 2\gamma}{\lambda_d(\bv{P}^{-1/2} (\bv{A}^T\bv{A} + \lambda \bv{I}) \bv{P}^{-1/2}) + \gamma} = O \left ( \frac{k \sigma_k^2(\bv{A}) + \sum_{i=k+1}^d \sigma_i^2(\bv{A})}{r\lambda} + \frac{d}{r} \right )
\end{align*}
So, by Theorem \ref{preCondsvrgThm}, for dense inputs, ignoring the $O(nd^{\omega(1,1,\log_d k)-1} \log d )$ precomputation cost to compute $\bv{Z}$, $\bv{P}^{-1/2}$ and $\bv{A}\bv{P}^{-1/2}$, which we only pay once, letting $\hat \kappa \eqdef \frac{k \sigma_k^2(\bv{A}) + \sum_{i=k+1}^d \sigma_i^2(\bv{A})}{d\lambda}$, we have 
$\mathcal{T}_c = O \left (\log c \cdot \left ( nd + \frac{d^2\hat \kappa}{r}\right) \right )$. If $nd \ge d^2 \hat \kappa$ then we already have runtime $O(nd\log(1/\epsilon))$ by Theorem \ref{preCondsvrgThm}. Otherwise, setting $\frac{\gamma}{\mu} \eqdef r = \frac{d \hat \kappa}{n}$ we can solve $\hat f$ up to $\epsilon$ accuracy in time:
\begin{align*}
O\left(\mathcal{T}_{4\left (\frac{2\gamma + \mu}{\mu} \right )^{3/2}} \sqrt{ \gamma/\mu} \log (1/\epsilon) \right ) = O \left (\log(1/\epsilon) n^{1/2}d^{3/2} \log(\bar \kappa) \sqrt{\bar \kappa} \right).
\end{align*}
For sparse inputs, we again ignore the one time precomputation cost of $O \left ( \nnz(\bv{A})k \log d + dk^{\omega-1} \right )$ to compute $\bv{Z}$ and $\bv{A}\bv{Z}$.
We have:
\begin{align*} 
\mathcal{T}_c = O \left ( \log c \cdot \left ( \nnz(\bv{A}) + \left [d \cdot d_s(\bv{A}) + d k \right ]\frac{\bar \kappa}{r} \right ) \right ).
\end{align*}
If $\nnz(\bv{A}) \ge \left [d\cdot d_s(\bv{A}) + d k \right ]\bar \kappa)$ then we already have runtime $O(\nnz(\bv{A})\log(1/\epsilon))$ by Theorem \ref{preCondsvrgThm}. Otherwise, 
$\nnz(\bv{A}) \le \left [d\cdot d_s(\bv{A}) + d k \right ]\bar \kappa$. So setting $r = \frac{\bar \kappa \left [d\cdot d_s(\bv{A}) + d k \right ]}{\sqrt{\nnz(\bv{A})}}$ we have:
\begin{align*}
O\left(\mathcal{T}_{4\left (\frac{2\gamma + \mu}{\mu} \right )^{3/2}} \sqrt{ \gamma/\mu} \log (1/\epsilon) \right ) = O \left (\log(1/\epsilon) \log(d\bar \kappa) \cdot \sqrt{\nnz(\bv{A}) [d\cdot d_s(\bv{A}) + dk ] \bar \kappa} \right ).
\end{align*}
\end{proof}

\subsection{Preconditioned Iterative Methods}
Finally, we describe how to combine deflation-based preconditioning with standard iterative methods, which can give runtime advantages over Theorem \ref{accPreCondSvrg} in some parameter regimes.
\begin{theorem}[Preconditioned Iterative Method]\label{preCondIte} For any $\bv A \in \mathbb{R}^{n \times d}$, $\bv{b} \in \mathbb{R}^d$, and $\lambda > 0$, let $\bv{M}_\lambda \eqdef \bv{A}^T\bv{A} + \lambda \bv{I}$ and $ \hat \kappa =\frac{\sigma_{k+1}^2(\bv{A})}{\lambda}$ where $k \in [0,...,d]$ is an input parameter.
There is an algorithm that uses $O( \nnz(\bv{A}) k\log d + dk^{\omega-1})$ precomputation time after which, given any input $\bv{y} \in \R^d$, letting $\bv{x}^* \eqdef \bv{M}_\lambda^{-1} \bv{y}$ the algorithm returns $\bv{x}$ such that with high probability
$\norm{\bv{x}-\bv{x}^*}_{\bv{M}_\lambda} \le \epsilon \norm{\bv{x}^*}_{\bv{M}_\lambda}$
in time $$O \left (\log(1/\epsilon) (\nnz(\bv{A}) + dk) \sqrt{\hat \kappa}\right ).$$
\end{theorem}
\begin{proof}
If we form the preconditioner $\bv{P}^{-1/2}$ as in Lemma \ref{traceBound}, by Lemmas 2 and 4 of \cite{gonen2016solving} we have the preconditioned condition number bound:
\begin{align*}
\hat \kappa \eqdef \frac{\lambda_1(\bv{P}^{-1/2}  \bv{M}_\lambda \bv{P}^{-1/2})}{\lambda_d(\bv{P}^{-1/2}  \bv{M}_\lambda \bv{P}^{-1/2})} = O \left ( \frac{\sigma_{k+1}^2(\bv{A})}{\lambda} \right )
\end{align*}
Note that the bound in \cite{gonen2016solving} is actually in terms of $\sigma_{k}^2(\bv{A})$, however we write $\sigma_{k+1}^2(\bv{A})$ so that our theorem holds in the case that $k = 0$. This can be achieved with no effect on the asymptotic runtime by setting $k$ to $k+1$.

We can now apply any accelerated linear system solver, such as Chebyshev iteration, Conjugate Gradient, or Accelerated Gradient Descent \cite{saad2003iterative,nesterov2013introductory} to obtain $\bv{x}$ with $\norm{\bv{x}-\bv{x}^*}_{\bv{M}_\lambda} \le \epsilon \norm{\bv{x}^*}_{\bv{M}_\lambda}$ in $O \left (\log(1/\epsilon) \sqrt{\hat \kappa} \cdot \text{matvec}(\bv{P}^{-1/2}(\bv{A}^T\bv{A} + \lambda \bv{I})\bv{P}^{-1/2} \right )$ where $\text{matvec}(\bv{P}^{-1/2}(\bv{A}^T\bv{A} + \lambda \bv{I})\bv{P}^{-1/2}$ is the time required to multiply a single vector by this matrix. $\text{matvec}(\bv{P}^{-1/2}(\bv{A}^T\bv{A} + \lambda \bv{I})\bv{P}^{-1/2}) = O(\nnz(\bv{A}) + dk)$ since $\bv{P}^{-1/2} = \bv{Z}\bv{\tilde \Sigma}^{-1/2} \bv{Z}^T + \delta (\bv{I}-\bv{ZZ}^T)$ can be applied in $O(dk)$ time. This gives the result combined with the precomputation time for $\bv{P}^{-1/2}$ from Lemma \ref{blockLanczos}.
\end{proof}

\section{Additional Proofs: Lower Bounds}\label{sec:hardnessAppendix}

\begin{lemma}[Determinant Hardness]\label{detHardness}
Given algorithm $\mathcal{A}$ which given $\bv{B} \in \R^{n\times n}$ returns $X \in (1\pm \epsilon) \det(\bv{B})$ in $O(n^\gamma \epsilon^{-c})$ time, we can detect if an $n$-node graph contains a triangle in $O(n^{\gamma + 12c})$ time.
\end{lemma}
\begin{proof}
Let $\bv{A} \in \R^{n \times n}$ be the adjacency matrix of an $n$-node graph $G$. Let $\lambda_1,...,\lambda_n$ denote its eigenvalues. Let $\bv{B} = \bv{I} + \delta \bv{A}$ for some $\delta$ which we will set later. We can write:
\begin{align}\label{detsum}
\det(\bv{B}) = \prod_{i=1}^n \lambda_i(\bv{B}) = \prod_{i=1}^n (1 + \delta \lambda_i) = \sum_{k=0}^n \left ( \delta^k \cdot \sum_{i_1 < i_2 < ... < i_k} \lambda_{i_1} \lambda_{i_2}...\lambda_{i_k}\right ).
\end{align}
The $k=0$ term in \eqref{detsum} is $1$, and the next two are easy to compute. $\delta \sum_{i=1}^n \lambda_i = \delta \tr(\bv{A}) = 0$, and $\delta^2 \sum_{i < j} \lambda_i\lambda_j = \frac{\delta^2}{2} \left (\sum_{i,j}  \lambda_i\lambda_j - \sum_i \lambda_i^2 \right ) = \frac{\delta^2}{2} \sum_i \lambda_i \tr(\bv{A}) - \frac{\delta^2}{2} \norm{\bv{A}}_F^2 = -\delta^2 \norm{\bv{A}}_F^2/2.$
For $k=3$ we have:
\begin{align*}
\delta^3 \sum_{i< j < k} \lambda_i \lambda_j \lambda_k &= \frac{\delta^3}{3} \left (\sum_{i< j} \lambda_i \lambda_j \tr(\bv{A}) - \sum_{i \neq j} \lambda_i^2 \lambda_j\right)\\
&= 0 - \frac{\delta^3}{3}\norm{\bv{A}}_F^2 \cdot \tr(\bv{A}) + \frac{\delta^3}{3}\tr(\bv{A}^3)\\
&= \frac{\delta^3}{3}\tr(\bv{A}^3).
\end{align*}
%Finally, writing $T = \sum_{k=4}^\infty \left ( \delta^k \sum_{i_1 \neq i_2 \neq ... \neq i_k} \lambda_{i_1} \lambda_{i_2}...\lambda_{i_k} \right)$ and setting $\delta = \frac{1}{10n^2}$ we can bound:
%\begin{align*}
%T \le \lambda_1 \cdot n \cdot \delta \left (\sum_{i\neq j \neq k} \lambda_i \lambda_j \lambda_k + T \right)\\
%\frac{9}{10} T \le \frac{1}{10} \sum_{i\neq j \neq k} \lambda_i \lambda_j \lambda_k \tag{Since $\lambda_1 \le n$ so $\lambda_1 n \delta \le \frac{1}{10}$}
%\end{align*}

We will bound the $k > 3$ terms by: $\left | \delta^k \cdot  \sum_{i_1 < i_2 < ... < i_k} \lambda_{i_1} \lambda_{i_2}...\lambda_{i_k}\right | \le {n \choose k} \delta^k \lambda_1^k \le (n\delta \lambda_1)^k \le (n^2 \delta)^k$ since $\lambda_1 \le n$. However, in order to obtain a tighter result, we will use stronger bounds for $k=4,5$. These bounds are very tedious but straightforward. Specifically:
\begin{align*}
\left |\delta^4 \sum_{i < j < k < l} \lambda_i \lambda_j \lambda_k \lambda_l \right |&= \frac{\delta^4}{4} \left |\tr(\bv{A}) \sum_{i < j < k} \lambda_i \lambda_j \lambda_k - \frac{1}{2}\sum_{i \neq j \neq k} \lambda_i^2 \lambda_j \lambda_k \right |\\
&= \frac{\delta^4}{8} \left | \tr(\bv{A}) \sum_{i \neq j} \lambda_i^2 \lambda_j - \sum_{i \neq j} \lambda_i^2 \lambda_j^2 - \sum_{i\neq j} \lambda_i^3 \lambda_j \right |\\
&= \frac{\delta^4}{8} \left |  \sum_{i \neq j} \lambda_i^2 \lambda_j^2 + \tr(\bv{A}) \sum_{i\neq j} \lambda_i^3 - \sum_i \lambda_i^4 \right |\\
&= \frac{\delta^4}{8} \left | \norm{\bv{A}}_F^2 - 2 \tr(\bv{A}^4) \right | \le \frac{\delta^4 n^4}{4}.
\end{align*}
And similarly:
\begin{align*}
\left |\delta^5 \sum_{i < j < k < l < m} \lambda_i \lambda_j \lambda_k \lambda_l \lambda_m \right |&= \frac{\delta^5}{30} \left | \sum_{i \neq j \neq k \neq l} \lambda_i^2 \lambda_j \lambda_k \lambda_l \right |\\
&= \frac{\delta^5}{30} \left | 2\sum_{i \neq j \neq k} \lambda_i^2 \lambda_j^2 \lambda_k + \sum_{i \neq j \neq k } \lambda_i^3 \lambda_j \lambda_k \right |\\
&= \frac{\delta^5}{30} \left | 5\sum_{i \neq j} \lambda_i^2 \lambda_j^3 + \sum_{i \neq j } \lambda_i^4 \lambda_j\right |\\
&= \frac{\delta^5}{30} \left | 5(\sum \lambda_i^2) (\sum \lambda_i^3) - 6\sum \lambda_i^5 \right |\\
&\le \frac{\delta^5 n^2}{6} \tr(\bv{A}^3) + \frac{\delta^5}{5} \lambda_1 \sum \lambda_i^4 \\
&\le \frac{\delta^5 n^2}{6} \tr(\bv{A}^3) + \frac{\delta^5 n^5}{5}.
\end{align*}
Finally, if we set $\delta = \frac{1}{10n^4}$ then we have:
\begin{align*}
\left |\sum_{k=4}^n \left ( \delta^k \cdot \sum_{i_1 < i_2 < ... < i_k} \lambda_{i_1} \lambda_{i_2}...\lambda_{i_k}\right ) \right | &\le \frac{\delta^4 n^4}{4} + \frac{\delta^5 n^5}{5} + \frac{\delta^5 n^2}{6} \tr(\bv{A}^3) + \sum_{k=6}^\infty (n^2\delta)^k\\
& \le \delta^3 \left ( \frac{1}{40} + \frac{1}{500} + \frac{1}{600}\tr(\bv{A}^3) + \left (\frac{1}{10^3} + \frac{1}{10^5} + ... \right)\right )\\
&\le \frac{\delta^3}{30} + \frac{\delta^3}{600} \tr(\bv{A}^3).
\end{align*}
We then write:
\begin{align*}
\det(\bv{B}) &= 1 - \frac{\delta^2 \norm{\bv{A}}_F^2}{2} + \frac{\delta^3 \tr(\bv{A}^3)}{3} \pm \frac{\delta^3}{30} \pm \frac{\delta^3}{600} \tr(\bv{A}^3).
\end{align*}
Since $1 \le \delta^3 \cdot 10^3 n^{12}$ and $\frac{\delta^2 \norm{\bv{A}}_F^2}{2} \le \delta^3 \cdot 5 n^6$
if we compute $X \in (1\pm c_1/n^{12})\det(\bv{B})$ for sufficiently small constant $c_1$ and subtract off $\left (1 - \frac{\delta^2 \norm{\bv{A}}_F^2}{2}\right)$, we will be able to determine if $\tr(\bv{A}^3) > 0$ and hence detect if $G$ has a triangle. So any algorithm approximating $\det(\bv{B})$ to $(1\pm \epsilon)$ error in $O(n^\gamma \epsilon^{-c})$ time yields a triangle detection algorithm running in $O(n^{\gamma + 12c})$ time.
\end{proof}
% !TEX root = normEstimation.tex
\section{Krylov and Power Methods preserve Invariant Norms}\label{sec:Krylovpreserve}
\newcommand\nc\newcommand
\nc\eps{\epsilon} \nc\PhI{\Phi_I}\nc\del{\delta}
\nc\Om{\bv{\Omega}}

In this section, we show that when the Krylov method~\cite{musco2015randomized} or the power method is 
 used to deflate the top singular vectors of the matrix,  any unitarily invariant norm 
 of the tail (remaining part of the spectrum) is preserved.
 Let $\bv{P}=\bv{Z}\bv{Z}^\top$ be the projector obtained for the top $k$ singular vectors of $\bv{A}$ using
 the Krylov method~\cite{musco2015randomized} or the power method. Then, for any  invariant norm 
 of the tail to be preserved, we just need to show the following:
 \begin{align}\label{mainBound}
\sigma_i({(\bv{I}-\bv{P})\bv{A}}) \le (1+\epsilon) \sigma_i({\bv{A}}-{\bv{A}}_k) + \frac{\epsilon}{n}\sigma_1({\bv{A}}).
\end{align}
 Since $\bv{P}$ is a $(1+\eps)$ approximation obtained
in terms of the spectral norm ,
we have by the spectral low rank approximation guarantee of \cite{musco2015randomized}
that $\sigma_1((\bv{I}-\bv{P})\bv{A}) \leq (1+\epsilon) \sigma_1(\bv{A}-\bv{A}_k)$,
where $\bv{A}_k$ is the best rank $k$ approximation of $\bv{A}$. %From the arguments in
%\cite{muscorand} (see sec.3.2, Lemma 9 and its proof),
Consider first any $i$ for which
$\sigma_i(\bv{A}-\bv{A}_k) \geq (1-\epsilon)\sigma_1(\bv{A}-\bv{A}_k) = (1-\epsilon) \sigma_{k+1}(\bv{A})$. Then we have, 
\[
 \sigma_i((\bv{I}-\bv{P})\bv{A}) \leq \sigma_1((\bv{I}-\bv{P})\bv{A}) \leq (1+\eps) \sigma_{k+1}(\bv{A}) 
 \leq \frac{(1+\eps)}{(1-\eps)} \sigma_i(\bv{A}-\bv{A}_k)\leq (1+3\eps)\sigma_i(\bv{A}-\bv{A}_k).
\]
For any $i$ with $\sigma_i(\bv{A}-\bv{A}_k) < (1-\epsilon) \sigma_{k+1}(\bv{A})$, i.e.,
with $\sigma_{i+k}({\bv{A}}) \le (1-\epsilon) \sigma_{k+1}({\bv{A}})$
there is a large ($> \epsilon$)
relative gap between this singular value and $\sigma_k$.
 We have by the min-max characterization of singular values:
\begin{align*}
\sigma_i(({\bv{I}}-{\bv{P}}){\bv{A}}) &= \min_{{\bv{Y}} | rank({\bv{Y}}) = n-i+1} \left ( \max_{\bv{y} \in span({\bv{Y}}) |
\norm{\bv{y}}_2 = 1} \norm{\bv{y}^T({\bv{I}}-{\bv{P}}){\bv{A}}}_2 \right )
\end{align*}
If we just set ${\bv{Y}} = {\bv{Z}}$, then we have:
\begin{align*}
\sigma_{n-k+1}(({\bv{I}}-{\bv{P}}){\bv{A}}) \le \max_{\bv{y} \in span({\bv{Z}}) | \norm{\bv{y}}_2 = 1}
\norm{\bv{y}^T({\bv{I}}-{\bv{P}}){\bv{A}}}_2 = 0.
\end{align*}
It follows that the bottom $k$ singular values are all $0$, and so equal to the
bottom $k$ singular values of $\sigma_i({\bv{A}}-{\bv{A}}_k)$.

Now, for $i < n-k+1$, 
let ${\bv{U}}_U$ denote the top $i+k-1$ singular vectors of ${\bv{A}}$ and ${\bv{U}}_L$ denote 
the bottom $n-(i+k-1)$ singular vectors. 
We set ${\bv{Y}} = [{\bv{Z}}, {\bv{U}}_L]$. First, we note that ${\bv{Y}}$ has 
$k + n- (i+k-1) = n-i+1$ columns and also $\rank({\bv{Y}}) = n - i + 1$. 
This is because, if we consider ${\bv{Y}}^T {\bv{Y}}$, the top left
$k \times k$ blocks is ${\bv{Z}}^T{\bv{Z}} = {\bv{I}}$ and the bottom right $n-i-k +1$ 
block is ${\bv{U}}_L^T {\bv{U}}_L = {\bv{I}}$. The off-diagonal entries in the top
right and bottom left blocks are all bounded by $1/\poly(n)$; the proof of this
latter statement is given in 
the latter part of this section, where we bound $\|\bv{Z}^T\bv{U}_L\|_2$. By the 
Gershgorin circle theorem~\cite{golub2012matrix}, all eigenvalues of ${\bv{Y}}^T{\bv{Y}}$ are in the 
range $1 \pm 1/\poly(n)$ and so the matrix is full rank.
So, we have:
\begin{align}\label{maxMinBound}
\sigma_{i}(({\bv{I}}-{\bv{P}})\bv{A}) \le \max_{\bv{y} \in span({\bv{Y}}) | \norm{\bv{y}}_2 = 1} \norm{\bv{y}^T({\bv{I}}-{\bv{P}}){\bv{A}}}_2.
\end{align}

Next, we can write $\bv{y} = [{\bv{Z}}, {\bv{U}}_L] \bv{w}$ for some $\bv{w}$. By our argument above, 
every singular value of $[{\bv{Z}}, {\bv{U}}_L]$ lies in $1 \pm 1/\poly(n)$. Then, we have
$\norm{\bv{w}}_2 \le 1 + 1/\poly(n)$. Splitting $\bv{w} = \bv{w}_1 + \bv{w}_2$, where $\bv{w}_1$
contains the first $k$ coordinates of the vector and $\bv{w}_2$ contains the rest, we have:
\begin{align*}
({\bv{I}}-{\bv{P}}) {\bv{y}} &= ({\bv{I}}-{\bv{P}}) {\bv{Z}} \bv{w}_1 + ({\bv{I}}-{\bv{P}}) {\bv{U}}_L \bv{w}_2\\
&= {0} + ({\bv{I}}-{\bv{P}}) {\bv{U}}_L \bv{w}_2\\
&= {\bv{U}}_L\bv{w}_2 - {\bv{Z}\bv{Z}}^T {\bv{U}}_L \bv{w}_2.
\end{align*} 
Then, we have
$$\norm{{\bv{Z}\bv{Z}}^T {\bv{U}}_L \bv{w}_2}_2 \le \norm{{\bv{Z}}^T {\bv{U}}_L}_2 \cdot \norm{\bv{w}_2}_2 \le 1/\poly(n),$$
where the inequality comes from bounding $\norm{{\bv{Z}}^T {\bv{U}}_L}_2$
using the fact that all its entries are at most $1/\poly(n)$ (see the latter part of this section),
and
bounding $\norm{\bv{w}_2}_2 \le \norm{\bv{w}}_2 \le 1+ 1/\poly(n)$.
Thus, finally we obtain:
\begin{align*}
\norm{{\bv{y}}^T ({\bv{I}}-{\bv{P}}) {\bv{A}}}_2 %&= \norm{{x}^T ({\bv{I}}-{\bv{P}})^2 {A}}_2\\
&\le \norm{\bv{w}_2^T {\bv{U}}_L^T {\bv{A}}}_2 + \norm{\bv{w}_2^T {\bv{U}}_L^T {\bv{Z}\bv{Z}}^T {\bv{A}}}_2\\
&\le (1+1/\poly(n)) \cdot \norm{{\bv{U}}_L^T {\bv{A}}}_2 + 1/\poly(n) \cdot \norm{{\bv{A}}}_2\\
&\le  (1+1/\poly(n))\sigma_i({\bv{A}}-{\bv{A}}_k) + \frac{1}{\poly(n)}\sigma_1({\bv{A}}).
\end{align*}

Plugging into \eqref{maxMinBound} gives the proof of \eqref{mainBound}.
Note that we can assume $\frac{1}{\poly(n)} \le \frac{\epsilon}{n}$ 
since if $\epsilon = o(1/\poly(n))$, we can just compute the SVD. 

\paragraph{Bound for  $\norm{{\bv{Z}}^T {\bv{U}}_L}_2$ :} 
  Let $\bv{Z}$ be an orthonormal basis for $\bv{A}^q\Om$ from the power method
  (or $p_q(\bv{A})\Om$ from block Krylov method).
  Then, for any singular vector $\bv{u}_i$ with corresponding singular value
  $\sigma_i(\bv{A})$ with $\sigma_i(A) \leq (1-\eps)\sigma_k(\bv{A})$, we have $\|\bv{u}_i^T\bv{Z}\|_2 \leq 1/\poly(n)$.

To see this, we first write $\bv{A}^q\Om = \bv{Z} \bv{T} \bv{W}^T$ in its SVD form. Then, we have
$\|\bv{u}_i^T \bv{A}^q \Om\|_2 = \|\bv{u}_i^T \bv{Z} \bv{T}\|_2 \geq \|\bv{u}_i^T \bv{Z}\|_2 \tau_{k,k}$, where $\tau_{k,k}$
is the $k$th diagonal entry of $\bv{T}$. 
On the other hand, we also have $\|\bv{u}_i^T \bv{A}^q \Om\|_2 = \sigma_i(\bv{A})^q \|\bv{v}_i^T \Om\|_2$, 
where $\bv{v}_i^T$ is the right singular vector corresponding to $\bv{u}_i$ of $\bv{A}$.
Since $\Om$ is i.i.d. Gaussian, we have $\|\bv{v}_i^T \Om\|_2^2 = O(k \log n)$ 
with probability $1-1/n^2$ 
for a fixed $\bv{v}_i$ (see proof of Lemma 2.12 in~\cite{w14}). 
So, we can union bound over all $n$ right singular vectors 
and the relation holds for all right singular vectors $\bv{v}\in \bv{V}$. We condition on this event. 

This event implies that
$\|\bv{u}_i^T \bv{A}^q \Om\|_2 = O(\sigma_i(\bv{A})^q \sqrt{k \log n})$. It follows that, if we show
$\tau_{k,k} \geq  \sigma_i(\bv{A})^q \poly(n)$, then$ \|\bv{u}_i^T \bv{Z}\|_2 \leq 1/\poly(n)$.

Next, to show that $\tau_{k,k} \geq  \sigma_i(\bv{A})^q \poly(n)$, it suffices to show 
$\tau_{k,k} \geq \sigma_k(\bv{A})^q \sqrt{k}$. Indeed, note that 
for $q = \Theta((\log n) / \eps)$ 
(for the block Krylov method, we need to use higher number
of iterations than required, which is
$q = \Theta((\log n) / \sqrt{\eps})$~\cite{musco2015randomized}), 
by definition of $i$, we have $\sigma_k(\bv{A}^q)/\sigma_i(\bv{A}^q) \geq \poly(n)$.

For this, we have
$\tau_{k,k}=\sigma_k(\bv{U} \bv{\Sigma}^q \bv{V}^T \Om)\geq \sigma_k (\bv{U}_k \bv{U}_k^T \bv{U} \bv{\Sigma}^q \bv{V}^T \Om)$ 
since $\sigma_i(\bv{P}\bv{A}) \leq \sigma_i(\bv{A})$
for any projection matrix $\bv{P}$ and any matrix $\bv{A}$.
Then, we have
$\sigma_k (\bv{U}_k \bv{U}_k^T \bv{U} \bv{\Sigma}^q \bv{V}^T \Om)= \sigma_k (\bv{U}_k \bv{\Sigma}_k^q \bv{V}_k^T \Om)
= \sigma_k (\bv{\Sigma}_k^q \bv{V}_k^T \Om)$ 
by the definition of the SVD, and that the columns of $\bv{U}_k$ are orthonormal.
So, we get
$\tau_{k,k}\geq \sigma_k(\bv{\Sigma}_k^q \bv{H})$, where $\bv{H}=\bv{V}_k^T\Om$ is a $k \times k$
matrix with i.i.d. Gaussian distribution
due to the rotational invariance of the Gaussian distribution. 
Since the $k$-th singular value is the smallest singular value of $\bv{H}$ and $\bv{\Sigma}_k$, we then have
$\sigma_k(\bv{\Sigma}_k^q \bv{H})\geq \sigma_k^q(\bv{A}) \sigma_k(\bv{H}) \geq \sigma_k^q \sqrt{k}\cdot C$
with probability at least $9/10$, for an arbitrary constant $C>0$, using
standard properties of minimum singular values of squared Gaussian matrices
(see e.g., Fact 6 in section 4.3 of~\cite{w14}). 
Thus, we have $\tau_{k,k}\geq C\sqrt{k}\sigma_k^q(\bv{A})$, which completes the proof. 

\section{Additional Proofs}\label{sec:generalAppendix}

\begin{repclaim}{claim:hypo}
	Let $f: \R^+ \rightarrow \R^+$ be a $\delta_f$-multiplicatively smooth function on the range $[a,b]$. For any $x,y \in [a,b]$ and $c \in (0,\frac{1}{3\delta_f})$ $$y \in \left [(1-c)x, (1+c)x\right] \Rightarrow f(y) \in [(1-3\delta_f c) f(x), (1+3\delta_f c) f(x)].$$
\end{repclaim}
\begin{proof}
Let $R$ denote the range $[\min(x,y), \max(x,y)]$. For $y \in \left [(1-c)x, (1+c)x\right]$ we have:
\begin{align}\label{eq:z1}
|f(x) - f(y)| \le |x - y| \cdot \sup_{z \in R} |f'(z)| \le c x \cdot \frac{\delta_f\sup_{z \in R} f(z)}{\min(x,y)} \le \frac{c\delta_f\sup_{z \in R} f(z)}{1-c}.
\end{align}
Similarly, letting ${Z} = \sup_{z \in R} f(z)$ we have:
\begin{align*}
{Z} - f(x) \le | x - y | \frac{\delta_f{Z}}{\min(x,y)} \le \frac{c\delta_f {Z}}{1-c} \le \frac{c}{\delta_f}{1-1/3-1/3}
\end{align*}
and hence ${Z} \le \frac{1-c}{1-(\delta_f + 1)c} f(x)$. So plugging into \eqref{eq:z1} we have:
\begin{align*}
|f(x) - f(y)| \le \frac{c \delta_f }{1-(\delta_f+1)c} \cdot f(x) \le \frac{c\delta_f}{1-1/3 -1/3} f(x) =3\delta_f f(x)
\end{align*}
which gives the result.
\end{proof}

\end{document}